\documentclass[nonacm, acmsmall]{acmart}

\AtBeginDocument{%
  }

\usepackage[british]{babel}

\usepackage{textcomp}
\usepackage{soul}
\usepackage{amsmath,amsfonts}
\usepackage{pifont}
\usepackage{nicefrac}
\usepackage{textcomp}
\usepackage{mathtools}
\usepackage{array}
\usepackage{tabularx,longtable}
\usepackage{rotating}
\usepackage{verbatim}
\usepackage{etoolbox}
\usepackage{calc}
\usepackage{arydshln}

\usepackage{graphicx}

\usepackage{xcolor}

\usepackage{xspace}

\usepackage{csquotes}
\usepackage{dirtytalk}

\usepackage{float}

\usepackage{pdflscape}

\usepackage{booktabs}
\usepackage{multirow}

\usepackage{datatool}
\DTLsetseparator{ = }
\newcommand{\placeholder}[2]{\DTLfetch{#1}{key}{#2}{value}}

\usepackage{breakcites}

\usepackage{tikz}

\usepackage[
algo2e,
algosection,
ruled,vlined,
linesnumbered,
norelsize]{algorithm2e}
\DontPrintSemicolon
\SetAlCapHSkip{0pt}
\IncMargin{-\parindent}
\SetProgSty{}
\SetArgSty{}

\SetKwComment{tcp}{{\upshape{}/\!/}\,}{}%
\SetCommentSty{textit}

\SetKwProg{Function}{Function}{}{}
\SetKwFunction{LCPCompare}{LCP-Compare}

\SetKwFor{Lock}{lock}{}{}
\SetKwFor{ParallelFor}{for}{do in parallel}{}
\SetKwFor{ParallelWhile}{while}{do in parallel}{}
\SetKwIF{If}{ElseIf}{Else}{if}{then}{else if}{else}{end if}

\newcommand{\arro}[1]{[\, #1 \kern.1em )}

\newcommand{\Oh}[1]{\ensuremath{\mathcal{O}(#1)}\xspace}

\usepackage{wasysym}
\usepackage{marvosym}

\usepackage{xstring}
\newcommand{\githash}[1]{$\FuncSty{\StrLeft{#1}{7}}$}

\usepackage{tikz}

\newif\ifshowcomments
\showcommentstrue

\ifshowcomments
\definecolor{fuchsiapink}{rgb}{1.0, 0.47, 1.0}
\newcommand{\tobi}[1]{{\color{fuchsiapink}[TH: #1]}}
\newcommand{\lars}[1]{{\color{cyan}{[Lars: #1]}}}
\newcommand{\peter}[1]{{\color{blue}[PS: #1]}}
\newcommand{\seb}[1]{{\color{orange}[SS: #1]}}
\newcommand{\todo}[1]{{\textcolor{red}{\bf [TODO]} \emph{#1}}}

\else
\newcommand{\tobi}[1]{}
\newcommand{\lars}[1]{}
\newcommand{\peter}[1]{}
\newcommand{\seb}[1]{}
\newcommand{\todo}[1]{}

\fi

\newcommand{\etal}{{et al}.}
\newcommand{\wrt}{w.r.t.\@\xspace}

\newcommand{\pluseq}{\mathrel{+}=}
\newcommand{\minuseq}{\mathrel{-}=}
\newcommand{\Increment}{\raisebox{.025ex}{\hbox{\tt ++}}}
\newcommand{\Decrement}{\raisebox{.025ex}{\hbox{\tt -}{\tt -}}}
\newcommand{\EndOfStatement}{;\quad}
\newcommand{\pluseqatomic}{\overset{\text{\tiny atomic}}{\mathrel{+}=}}
\newcommand{\minuseqatomic}{\overset{\text{\tiny atomic}}{\mathrel{-}=}}

\newcommand{\neighbors}{\ensuremath{{\Gamma}}}%
\newcommand{\incnets}{\ensuremath{{I}}}%

\newcommand{\maxsize}[1]{\ensuremath{\Delta_{#1}}}

\newcommand{\meddeg}{\ensuremath{\widetilde{d(v)}}}
\newcommand{\medsize}{\ensuremath{\widetilde{|e|}}}

\newcommand{\Partition}{\ensuremath{{\Pi}}}%
\newcommand{\DeltaPartition}{\ensuremath{\Delta \Partition}}%
\newcommand{\pinsinpart}{\ensuremath{{\Phi}}}

\newcommand{\con}{\ensuremath{\lambda}}
\newcommand{\conset}{\ensuremath{\Lambda}}
\newcommand{\quotientgraph}{\ensuremath{\mathcal{Q}}}

\newcommand{\nodeblock}[1]{\ensuremath{\Partition[#1]}}

\newcommand{\ocut}{\ensuremath{\mathfrak{f}_c}}%
\newcommand{\ocon}{\ensuremath{\mathfrak{f}_{\lambda-1}}}%
\newcommand{\osoed}{\ensuremath{\mathfrak{f}_s}}%
\newcommand{\gain}[2]{\ensuremath{g_{#1}(V_{#2})}}

\newcommand{\cutnets}{\ensuremath{E_{\text{Cut}}(\Partition)}}

\newcommand{\expgain}{\ensuremath{\Delta_{\text{exp}}}}
\newcommand{\attrgain}{\ensuremath{\Delta}}

\newcommand{\cluster}{\mathcal{C}}

\newcommand{\statearray}{\ensuremath{\text{\textsf{state}}}}
\newcommand{\statevar}[1]{\ensuremath{\text{\textsf{#1}}}}

\newcommand{\pending}{\ensuremath{\text{\textsf{pending}}}}
\newcommand{\rep}{\ensuremath{\text{\textsf{rep}}}}
\newcommand{\contractions}{\ensuremath{\mathcal{C}}}
\newcommand{\contractionforest}{\ensuremath{\mathcal{F}}}
\newcommand{\maxbatchsize}{\ensuremath{b_{\max}}}
\newcommand{\batches}{\ensuremath{\mathcal{B}}}

\newcommand{\flownetwork}{\ensuremath{\mathcal{N}}}
\newcommand{\flowhypergraph}{\ensuremath{\mathcal{H}}}
\newcommand{\flownodeset}{\ensuremath{\mathcal{V}}}
\newcommand{\flowedgeset}{\ensuremath{\mathcal{E}}}
\DeclareMathAlphabet{\mathpzc}{OT1}{pzc}{m}{n}
\newcommand{\capacity}{\ensuremath{\mathpzc{c}}}
\newcommand{\source}{\ensuremath{s}}
\newcommand{\sink}{\ensuremath{t}}
\newcommand{\sourcesink}{\ensuremath{(\source,\sink)}}
\newcommand{\innode}{\ensuremath{e_{\text{in}}}}
\newcommand{\outnode}{\ensuremath{e_{\text{out}}}}
\newcommand{\distance}{\ensuremath{d}}
\newcommand{\excess}{\ensuremath{\operatorname{exc}}}
\newcommand{\activequeue}{\ensuremath{A}}
\newcommand{\overbar}[1]{\mkern 1.5mu\overline{\mkern-1.5mu#1\mkern-1.5mu}\mkern 1.5mu}

\newcommand{\subhypergraph}[1]{\ensuremath{H[#1]}}

\newcommand{\balancedconstraint}[1]{\ensuremath{L_{#1}}}

\newcommand{\matharray}[1]{\ensuremath{\FuncSty{#1}}}

\newcommand{\atomicfunc}[1]{\ensuremath{\text{\texttt{#1}}}}
\newcommand{\textfunc}[1]{\texttt{#1}}
\newcommand{\gmean}[1]{\ensuremath{\overbar{#1}}}
\newcommand{\gmeantime}{geo\-metric mean running time}
\newcommand{\splitatcommas}[1]{%
  \begingroup
  \begingroup\lccode`~=`, \lowercase{\endgroup
    \edef~{\mathchar\the\mathcode`, \penalty0 \noexpand\hspace{0pt plus 1em}}%
  }\mathcode`,="8000 #1%
  \endgroup
}

\newcommand{\shortmediumhg}{M$_{\scriptstyle \text{HG}}$}
\newcommand{\shortlargehg}{L$_{\scriptstyle \text{HG}}$}
\newcommand{\shortmediumgr}{M$_{\scriptstyle \text{G}}$}
\newcommand{\shortlargegr}{L$_{\scriptstyle \text{G}}$}
\newcommand{\shortmediumparameterhg}{M$_{\scriptstyle \text{P}}$}

\newcommand{\mediumhg}{set \shortmediumhg}
\newcommand{\largehg}{set \shortlargehg}
\newcommand{\mediumgr}{set \shortmediumgr}
\newcommand{\largegr}{set \shortlargegr}

\newcommand{\SAT}{\textsc{Sat14}}
\newcommand{\Dual}{\textsc{Dual}}
\newcommand{\Primal}{\textsc{Primal}}
\newcommand{\ISPD}{\textsc{Ispd98}}
\newcommand{\Literal}{\textsc{Literal}}
\newcommand{\SPM}{\textsc{Spm}}
\newcommand{\DAC}{\textsc{Dac2012}}
\newcommand{\Random}{\textsc{Random Graphs}}
\newcommand{\Social}{\textsc{Social Networks}}
\newcommand{\Dimacs}{\textsc{Dimacs}}

\newcommand{\Partitioner}[1]{\textsf{#1}} 
\newcommand{\Algorithm}[1]{\textsc{#1}}

\newcommand{\CPP}{\texttt{C++}}
\newcommand{\gpp}{\texttt{g++}}

\newcommand{\ShortTBB}{\texttt{TBB}}

\usepackage{pgfplots}
\usepgfplotslibrary{external}
\newif\ifpdfplots
\pdfplotstrue

\tikzset{external/system call={%
  pdflatex \tikzexternalcheckshellescape -halt-on-error-interaction=batchmode -jobname "\image"
  "\string\PassOptionsToPackage{bookmarks=false}{hyperref}\texsource"}}

\begin{document}

\title{Scalable High-Quality Hypergraph Partitioning}

\author{Lars Gottesbüren}
\email{lars.gottesbueren@kit.edu}
\affiliation{%
  \institution{Karlsruhe Institute of Technology}
  \streetaddress{Am Fasanengarten 5}
  \city{Karlsruhe}
  \state{BW}
  \country{Germany}
}


\author{Tobias Heuer}
\email{tobias.heuer@kit.edu}
\affiliation{%
  \institution{Karlsruhe Institute of Technology}
  \country{Germany}
}

\author{Nikolai Maas}
\email{nikolai.maas@student.kit.edu}
\affiliation{%
  \institution{Karlsruhe Institute of Technology}
  \country{Germany}
}

\author{Peter Sanders}
\email{sanders@kit.edu}
\affiliation{%
  \institution{Karlsruhe Institute of Technology}
  \country{Germany}
}

\author{Sebastian Schlag}
\email{research@sebastianschlag.com}
\affiliation{%
\institution{Independent Researcher}
\city{Sunnyvale}
\state{CA}
\country{USA}
}

\renewcommand{\shortauthors}{Gottesbüren et al.}

\begin{abstract}
Balanced hypergraph partitioning is an NP-hard problem with many applications, e.g., optimizing communication in distributed data placement problems.
The goal is to place all nodes across k different blocks of bounded size, such that hyperedges span as few parts as possible. This problem is well-studied
in sequential and distributed settings, but not in shared-memory. We close this gap by devising efficient and scalable shared-memory algorithms for all
components employed in the best sequential solvers without compromises with regards to solution quality.

This work presents the scalable and high-quality hypergraph partitioning framework \Partitioner{Mt-KaHyPar}. Its most important components are parallel
improvement algorithms based on the FM algorithm and maximum flows, as well as a parallel clustering algorithm for coarsening -- which are used in a
multilevel scheme with $\log(n)$ levels. As additional components, we parallelize the $n$-level partitioning scheme, devise a deterministic version of
our algorithm, and present optimizations for plain graphs.

We evaluate our solver on more than 800 graphs and hypergraphs, and compare it with 25 different algorithms from the literature. Our fastest
configuration outperforms almost all existing hypergraph partitioners with regards to both solution quality and running time. Our highest-quality
configuration achieves the same solution quality as the best sequential partitioner \Partitioner{KaHyPar}, while being an order of magnitude faster
with ten threads. Thus, two of our configurations occupy all fronts of the Pareto curve for hypergraph partitioning.
Furthermore, our solvers exhibit good speedups, e.g., 29.6x in the geometric mean on 64 cores (deterministic), 22.3x ($\log(n)$-level), and 25.9x
($n$-level).

\end{abstract}



\maketitle

\clearpage

\section{Introduction}

The \emph{balanced hypergraph partitioning problem} asks for a partition of the node set of a hypergraph into a fixed number of disjoint blocks with bounded size
such that an objective function defined on the hyperedges is minimized. The two most prominent objective functions
are the \emph{edge cut} and \emph{connectivity} metric. The former counts the number of hyperedges connecting more than
one block, while the latter additionally considers the number of blocks spanned by each hyperedge.
The problem has gained attraction in the field of very-large-scale-integration (VLSI) design already in the 1960s~\cite{HYPERGRAPH-KL,Rutman1964,HAUCK,FM}.
Since then, it has been widely adopted in many other areas, such as minimizing the communication volume in parallel scientific
simulations~\cite{PATOH,DBLP:conf/ipps/CatalyurekA01,DBLP:conf/sc/CatalyurekA01}, storage sharding in distributed
databases~\cite{schism,sword,SHP,clay,hepart,adp}, simulations of distributed quantum circuits~\cite{gray2021hyper,andres2019automated},
and as a branching strategy in satisfiability solvers~\cite{SATApplication}.

Unfortunately, balanced partitioning is NP-hard~\cite{LENGAUER,DBLP:journals/tcs/GareyJS76} and hard to
approximate~\cite{DBLP:conf/mfcs/Feldmann12}.
Thus, heuristic solutions are used in practice  -- with
the multilevel scheme emerging as the most successful method to achieve high solution quality in a reasonable
amount of time~\cite{DBLP:conf/sc/HendricksonL95,BarnardS93}.
Figure~\ref{fig:multilevel_paradigm} illustrates this
technique, which consists of three phases.
First, the hypergraph is \emph{coarsened} to obtain a hierarchy of successively smaller and structurally similar
approximations of the input hypergraph by \emph{contracting} pairs or clusters of highly-connected nodes.
Once the hypergraph is small enough, an \emph{initial partition} into $k$ blocks is computed.
Subsequently, the contractions are reverted level-by-level, and, on each level, \emph{local search} heuristics are
used to improve the partition from the previous level.

\begin{figure}
  \centering
  \includegraphics[width=0.8\textwidth]{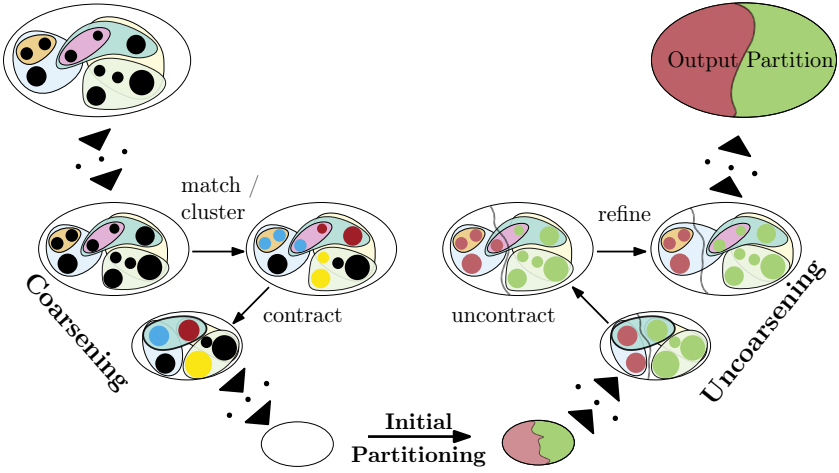}
  \caption{The multilevel paradigm.}
  \label{fig:multilevel_paradigm}
\end{figure}

There is a diverse landscape of algorithms that implement the multilevel framework with different time-quality trade-offs, as illustrated in
Figure~\ref{fig:pareto_plot}.
The plot shows two major shortcomings of existing solvers: (i) higher solution
quality comes at the cost of higher running times often by several orders of magnitude, and (ii) parallel algorithms do
not achieve the same solution quality as the best sequential systems because they use comparatively
weaker components that are easier to parallelize (with the exception of our new solver \Partitioner{Mt-KaHyPar}).
Historically, the parallel partitioning community has focused on algorithms
for the distributed-memory model, which turned out to be not well-suited for the fine-grained parallelism required
to effectively parallelize high-quality techniques. However, as the number of cores and main-memory
capacity in modern machines increases, we believe that the shared-memory model has become a viable alternative for
processing large (hyper)graphs and can be used for closing the quality gap between sequential and parallel partitioning algorithms.

\begin{figure}
	\centering
	\begin{minipage}{\textwidth}
		\centering
  \ifpdfplots
    \includegraphics{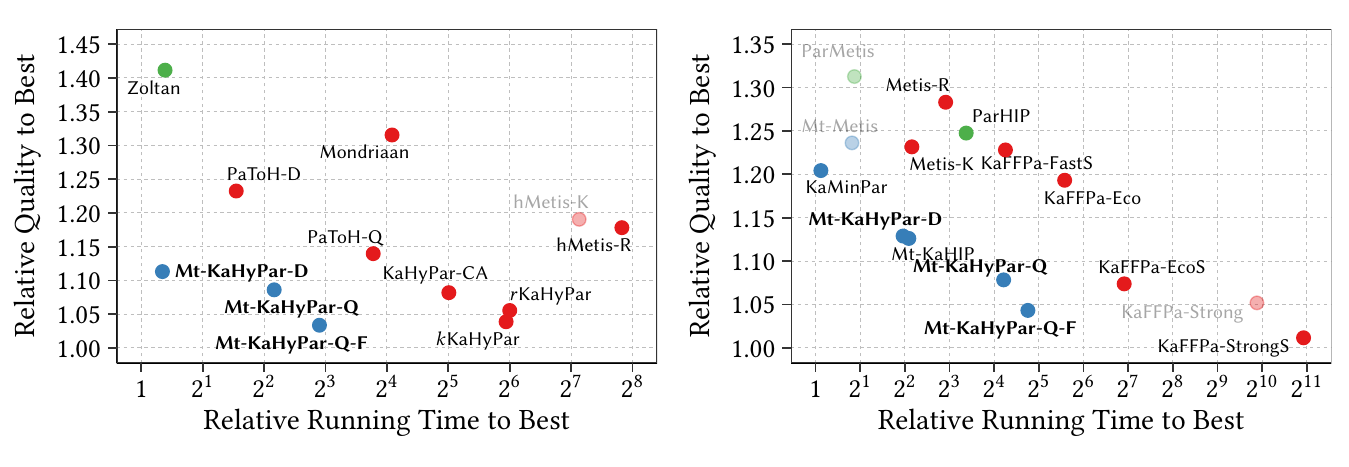}
  \fi
	\end{minipage} %
	\begin{minipage}{\textwidth}
		\vspace{-0.1cm}
		\centering
  \ifpdfplots
    \includegraphics{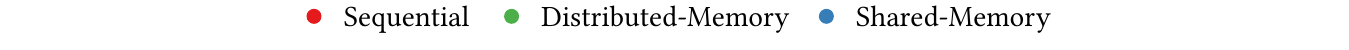}
  \fi
	\end{minipage} %
	\vspace{-0.5cm}
  \caption{Solution quality and running times of existing algorithms for hypergraph partitioning (left, connectivity metric) and
           graph partitioning (right, edge cut metric).
           For the $y$-values in the plot (solution quality), we compute the ratios of the objective values of an algorithm relative to the best value produced by any algorithm for each instance and
           aggregate them using the harmonic mean (similarly for running times on the $x$-axis).
           Markers on the lower left side are considered better.
           We run each parallel algorithm using 10 threads.
           Instances are restricted to more than two million edges/pins to make parallelism worthwile (see \mediumgr~and~\shortmediumhg~in Section~\ref{sec:experiments}).
           Partially transparent markers indicate solvers producing more than 15$\%$ infeasible partitions (either imbalanced or timeout).}
	\label{fig:pareto_plot}
\end{figure}

\paragraph{Main Contributions}

Figure~\ref{fig:pareto_plot} highlights the main contribution of this work: A shared-memory multilevel algorithm (\Partitioner{Mt-KaHyPar}) that achieves
the same solution quality as the best sequential codes, while being faster than most of the relevant parallel algorithms in its fastest configuration.
In particular, our Mt-KaHyPar solvers occupy all points on the Pareto frontier for hypergraphs (left) as well as the middle segment for graphs (right).

This is achieved by implementing parallel formulations for the core techniques used in the best sequential algorithms
without compromises in solution quality.
Our coarsening algorithm contracts a clustering of highly-connected nodes on each level and is guided by the community structure of the hypergraph.
The clustering algorithm uses a less restrictive locking protocol than a previous approach~\cite{PARALLEL-PATOH} and resolves conflicting clustering decisions
\emph{on-the-fly}. Initial partitioning is done via parallel recursive bipartitioning
and a portfolio solver, leveraging work-stealing to account for load imbalances. The key feature distinguishing \Partitioner{Mt-KaHyPar} from previous parallel systems are the substantially stronger local search
algorithms. We present the first fully-parallel implementation of the FM algorithm and
a parallel version of flow-based refinement.
For these algorithms, we propose several novel and easy-to-implement solutions to overcome some fundamental parallelization challenges such as, for example, techniques to
(re)compute correct gain values for concurrent node moves.

We also present several extensions of the core multilevel algorithm.
We devise the first parallel formulation of the $n$-level partitioning scheme -- the most extreme instantiation of the multilevel technique --
contracting only a single node on each level.  Correspondingly, in each refinement step, only a single node is uncontracted followed by a highly-localized
search for improvements around the uncontracted node, leading to more fine-grained refinement and ultimately better solution quality in a single run.
Furthermore, we present a deterministic version of our multilevel algorithm.
This offers reproducible results and thus also stable results, whereas previous algorithms may have large variance from repeated runs.
Furthermore, some applications even require deterministic results or value them highly (e.g., VLSI design due to manual post-processing).
Moreover, we present data structure optimizations that speed up our algorithm by a factor of two when running on plain graphs instead of hypergraphs.

In our extensive experimental evaluation, we compare \Partitioner{Mt-KaHyPar} to $25$ different sequential and parallel
graph and hypergraph partitioners on over $800$ graphs and hypergraphs with up to 2 billion edges/pins.
As of today and to the best of our knowledge, this is the most comprehensive comparison of partitioning algorithms in the literature.
As a main result, the highest-quality configuration of \Partitioner{Mt-KaHyPar}
produces partitions that are on par with \Partitioner{KaHyPar}~\cite{KAHYPAR-DIS} -- the best sequential hypergraph partitioner -- while being
almost an order of magnitude faster with only ten threads. The fastest configuration of \Partitioner{Mt-KaHyPar} achieves a self-relative speedup
of $22.3$ with 64 threads and computes partitions that are $23\%$ better
than those of \Partitioner{Zoltan}~\cite{ZOLTAN} (distributed-memory), while being a factor of $2.72$ faster on average.
Out of all evaluated algorithms, \Partitioner{KaFFPa}~\cite{KAFFPA} (sequential) computes slightly better solutions, while
\Partitioner{KaMinPar}~\cite{KAMINPAR} (shared-memory) is faster than \Partitioner{Mt-KaHyPar}.

The work presents the main results of several conference publications~\cite{MT-KAHYPAR-D,MT-KAHYPAR-Q,MT-KAHYPAR-FLOWS,MT-KAHYPAR-SDET} and summarizes
the dissertations of Gottesbüren~\cite{GOTT-DIS} and Heuer~\cite{HEUER-DIS}. The added value of the paper is the detailed
overview of the overall framework that contains the highest-quality and one of the fastest algorithms for partitioning (hyper)graphs.
This paper puts particular focus on our multilevel partitioning algorithm which provides the best time-quality trade-off.
We describe the algorithm with a greater level of detail compared to the corresponding conference version~\cite{MT-KAHYPAR-D}.
Furthermore, the previously mentioned optimizations for graph partitioning are unpublished. Another key contribution is the large experimental
evaluation, going beyond the scope of the individual publications by including graph partitioning, breaking down the running times of individual components, and including even more competing baseline algorithms.
We included almost all publicly available multilevel graph
and hypergraph partitioning algorithms to provide a comprehensive overview on the landscape of partitioning tools.


\paragraph{Outline}
Section~\ref{sec:preliminaries} introduces basic notation and definitions used throughout this work. We then start
the algorithm description with a high-level overview of the multilevel partitioning algorithm in Section~\ref{sec:overview}.
The following sections are structured according to the different phases of the multilevel scheme: Section~\ref{sec:coarsening}
and~\ref{sec:initial_partitioning} describe the coarsening and initial partitioning algorithm, while we discuss
different concurrent gain (re)computation techniques and the implementation of the parallel FM and flow-based
refinement algorithms in Section~\ref{sec:gain_computation}--\ref{sec:flows}. In Section~\ref{sec:nlevel}, we
present the parallelization of the $n$-level partitioning scheme, and conclude the algorithmic part with our
data structure optimizations for graph partitioning and the deterministic version of the multilevel algorithm
in Section~\ref{sec:parallel_graph} and~\ref{sec:deterministic}. We then turn to the experimental evaluation
in Section~\ref{sec:experiments}.
Here, we evaluate the solution quality and scalability of the different configurations of
\Partitioner{Mt-KaHyPar}, and compare them to existing partitioning algorithms.
Section~\ref{sec:conclusion} concludes the work and presents directions for future research.

As this paper covers a wide range of partitioning techniques, we review relevant literature in the
corresponding sections. For a comprehensive overview on (hyper)graph partitioning, we refer the reader
to existing surveys~\cite{HYPERGRAPH-SURVEY,ALPERT-SURVEY,BulucMSS016,PAPA-MARKOV,GRAPH-SURVEY}
and the literature overviews in the theses of Lars Gottesbüren~\cite{GOTT-DIS}, Tobias Heuer~\cite{HEUER-DIS},
and Sebastian Schlag~\cite{KAHYPAR-DIS}.

\section{Preliminaries}\label{sec:preliminaries}

\paragraph{Hypergraphs}
A \emph{weighted hypergraph} $H=(V,E,c,\omega)$ is defined as a set of  $n$ nodes $V$ and a set of $m$
hyperedges $E$ (also called \emph{nets})
with node weights $c:V \rightarrow \mathbb{R}_{>0}$ and net weights $\omega:E \rightarrow \mathbb{R}_{>0}$, where each net $e$ is
a subset of the node set $V$. The nodes of a net are called its \emph{pins}.
We extend $c$ and $\omega$ to sets in a natural way, i.e., $c(U) :=\sum_{u \in U} c(u)$ and $\omega(F) :=\sum_{e \in F} \omega(e)$.
A node $u$ is \emph{incident} to a net $e$ if $u \in e$.
$\incnets(u) := \{e \mid u \in e\}$ is the set of all incident nets of $u$.
The set $\neighbors(u) := \{ v \mid \exists e \in E : \{u,v\} \subseteq e\}$ denotes the neighbors of $u$.
Two nodes $u$ and $v$ are \emph{adjacent} if $v \in \neighbors(u)$.
The \emph{degree} of a node $u$ is $d(u) := |\mathrm{I}(u)|$.
The \emph{size} $|e|$ of a net $e$ is the number of its pins.
Nets of size one are called \emph{single-pin} nets.
We denote the number of pins of a hypergraph with $p := \sum_{e \in E} |e| = \sum_{v \in V} d(v)$.
We call two nets $e_i$ and $e_j$ \emph{identical} if $e_i = e_j$.
Given a subset $V' \subset V$, the \emph{subhypergraph}
$\subhypergraph{V'}$ is defined as $\subhypergraph{V'}:=(V', \{e \cap V' \mid e \in E : e \cap V' \neq \emptyset \}, c, \omega')$
where $\omega'(e \cap V')$ is the weight of hyperedge $e$ in $H$.
The \emph{bipartite graph representation} $G_x := (V \cup E, E_x)$~\cite{BIPARTITE-GRAPH,HYPERGRAPH-KL}
of an unweighted hypergraph $H = (V,E)$ contains the nodes and nets of $H$ as node set and for each pin $u \in e$, we add an undirected edge $\{u,e\}$ to $E_x$.
More formally, $E_x := \{\{u,e\} \mid \exists e \in E: u \in e\}$.

\paragraph{Clusterings and Partitions}
A \emph{clustering} $\cluster = \{C_1, \ldots, C_l\}$ of a hypergraph $H = (V,E,c,\omega)$ is a partition of the node set $V$ into disjoint subsets.
A cluster $C_i$ is called a \emph{singleton} cluster if $|C_i| = 1$.
A node contained in a singleton cluster is called \emph{unclustered}.
A \emph{$k$-way partition} of a hypergraph $H$ is a clustering into a predefined number of disjoint blocks
$\Partition = \{V_1, \ldots, V_k\}$.
A $2$-way partition is also called a \emph{bipartition}.
We denote the block to which a node $u$ is assigned by $\nodeblock{u}$.
For each net $e$, $\conset(e) := \{V_i \mid  V_i \cap e \neq \emptyset\}$ denotes the \emph{connectivity set} of $e$.
The \emph{connectivity} $\con(e)$ of a net $e$ is $\con(e) := |\conset(e)|$.
A net is called a \emph{cut net} if $\con(e) > 1$.
A node $u$ that is incident to at least one cut net is called \emph{boundary node}.
The number of pins of a net $e$ in block $V_i$ is denoted by $\pinsinpart(e,V_i) := |e \cap V_i|$.
We refer to $\pinsinpart(e,V_i)$ as the \emph{pin count value} for a net $e$ and block $V_i$.
The set $E(V_i,V_j) := \{e \in E \mid \{V_i,V_j\} \subseteq \conset(e)\}$ represents the
cut nets connecting block $V_i$ and $V_j$. Two blocks $V_i$ and $V_j$ are \emph{adjacent} if $E(V_i,V_j) \neq \emptyset$.
The \emph{quotient graph} $\quotientgraph := (\Partition, E_{\Partition} := \{(V_i, V_j) \mid E(V_i,V_j) \neq \emptyset\})$
contains an edge between all adjacent blocks. 

\paragraph{The Balanced Hypergraph Partitioning Problem}

The \emph{balanced hypergraph partitioning problem} is to
find a $k$-way partition $\Partition$ of a hypergraph $H$ that minimizes an objective function defined on the hyperedges where
each block $V' \in \Partition$ satisfies the \emph{balance constraint}: $c(V') \leq \balancedconstraint{\max} := (1+\varepsilon) \lceil \frac{c(V)}{k} \rceil$\footnote{The
$\lceil\cdot\rceil$ in this definition ensures that there is always a feasible solution for inputs with unit node weights. However, this does not hold for general weighted inputs as finding
a balanced solution for its own is an NP-hard problem~\cite{garey1979computers}. There exists several alternative definitions~\cite{KAMINPAR,DEEP-BALANCE-PAPER}, but no commonly
accepted way how to deal with feasibility. In this work, we use the original definition since our benchmark instances are unweighted.}
for some \emph{imbalance ratio} $\varepsilon \in (0,1)$.
If $\Partition$ satisfies the balance constraint, we call $\Partition$ \emph{$\varepsilon$-balanced} or just say \emph{balanced} or \emph{feasible} when $\varepsilon$
is clear from the context.
For $k = 2$, we refer to the problem as the \emph{bipartitioning} problem.
The two most prominent objective functions are the \emph{cut-net} metric $\ocut := \sum_{e \in \cutnets} \omega(e)$ (also called \emph{edge cut} metric for graph partitioning)
and \emph{connectivity} metric $\ocon(\Partition) := \sum_{e \in \cutnets} (\lambda(e) - 1) \cdot \omega(e)$ (also called $(\lambda - 1)$-metric)
where $\cutnets$ denotes the set of all cut nets. The cut-net metric directly generalizes the edge cut metric from graphs to
hypergraphs and minimizes the weight of all cut hyperedges.
The connectivity metric additionally considers the number of blocks connected by a net and thus more accurately models
the communication volume for parallel computations~\cite{PATOH} (e.g., for the parallel sparse matrix-vector multiplication).
The hypergraph partitioning problem is NP-hard for both objective functions~\cite{LENGAUER}.

\paragraph{Recursive Bipartitioning vs Direct $k$-way Partitioning}
A $k$-way partition of a hypergraph can be obtained either by \emph{recursive bipartitioning} or
\emph{direct $k$-way partitioning}. The former first computes a bipartition and then
calls the bipartitioning routine on both blocks recursively until the input hypergraph is divided into the desired number of blocks. The latter partitions
the hypergraph directly into $k$ blocks and applies $k$-way local search algorithms to improve the solution.

\section{A Brief Overview of the Partitioning Algorithm}\label{sec:overview}

Algorithm~\ref{pseudocode:multilevel} shows the high-level structure of our multilevel partitioning algorithm.
While the pseudocode presented does not explicitly exhibit parallelism, it shows the algorithmic components for which we provide
parallel implementations.

The coarsening algorithm proceeds in rounds until the hypergraph is considered as small enough for initial partitioning.
In each round, we find a clustering of highly-connected nodes and subsequently contract the clustering in parallel.
The clustering algorithm iterates over the nodes in parallel and finds the best target cluster for a node according
to a rating function. Afterwards, the node joins its desired cluster for which we implement a novel locking protocol that detects
and resolves conflicting clustering decisions on-the-fly.

Initial partitioning is done via parallel
recursive bipartitioning using a novel work-stealing approach to account for load imbalances within the parallel bipartitioning
calls. To compute an initial bipartition, we use a portfolio of \emph{nine} different bipartitioning techniques, which is run several
times in parallel. The best bipartition out of all runs is then used as initial solution.

In the uncoarsening phase, we project the partition onto the next hypergraph in the hierarchy
by assigning the nodes to the block of their corresponding constituent in the coarser representation.
Subsequently, we improve the partition using three different parallel refinement algorithms:
label propagation refinement (used in most of the existing parallel partitioning algorithms),
a highly-localized version of the FM algorithm (improves an existing implementation used in \Partitioner{Mt-KaHIP}~\cite{MT-KAHIP}),
and a novel parallelization of flow-based refinement. The rationale behind the use of three different local search algorithms executed in this order
is that it allows for increasingly better solution quality at the cost of higher running times.

The following sections are structured according to the different phases of the multilevel scheme, and provide a more detailed
explanation of the different algorithmic components of Algorithm~\ref{pseudocode:multilevel}. In Section~\ref{sec:coarsening}
and~\ref{sec:initial_partitioning}, we present our coarsening and initial partitioning algorithm. The
description of the uncoarsening phase is split into three separate sections: Section~\ref{sec:gain_computation} describes the
partition data structure and several concurrent gain (re)computation techniques, while Section~\ref{sec:fm} and~\ref{sec:flows} presents our parallel
FM and flow-based refinement algorithm.

\begin{algorithm2e}[!t]
\KwIn{Hypergraph $H = (V,E)$, number of blocks $k$}
\KwOut{$k$-way partition $\Partition$ of $H$}
\caption{The Multilevel Partitioning Algorithm
}\label{pseudocode:multilevel}

$H_1 \gets H\EndOfStatement \mathcal{H} \gets \langle H_1 \rangle \EndOfStatement n \gets 1$\;
\While {$V_i$ has too many nodes}{
  $\cluster \gets \FuncSty{ComputeClustering}(H_n)$ \;
  $H_{n + 1} \gets H_n.\FuncSty{Contract}(\cluster)\EndOfStatement \mathcal{H} \gets \mathcal{H} \cup \langle H_{n+1}\rangle \EndOfStatement \Increment n$\; \label{multilevel:contraction}
}

$\Partition \gets$ \FuncSty{InitialPartition}($H_n, k$)\;
\For {$i = n - 1$ {\bf down to} $1$}{
  $\Partition \gets$ project $\Partition$ onto $H_i$\;
  $\FuncSty{LabelPropagationRefinement}(H_i,\Partition)$ \tcp*[r]{finds easy improvements by moving single nodes}
  $\FuncSty{FMRefinement}(H_i,\Partition)$   \tcp*[r]{finds short and non-trivial move sets}
  $\FuncSty{FlowBasedRefinement}(H_i,\Partition)$ \tcp*[r]{global optimization finding long and complex move sets}
}

\Return{$\Partition$}
\end{algorithm2e}


\section{The Coarsening Phase}\label{sec:coarsening}

The goal of the coarsening phase is to find successively smaller and structurally similiar approximations of the input
hypergraph~\cite{Walshaw2003} such that initial partitioning can find a partition of high quality not significantly worse
than the partition that can be found on the input hypergraph~\cite{karypis2003multilevel}.
This can be achieved by grouping highly-connected nodes together and merging each group into a single node,
which can be done by computing either a matching~\cite{DBLP:conf/sc/HendricksonL95,
KarypisK98, DBLP:journals/jpdc/KarypisK98a,WalshawC00,DIBAP,MONDRIAAN,PARJOSTLE,PARMETIS,PT-SCOTCH,KAPPA,MT-METIS,ZOLTAN} or clustering of the
nodes~\cite{KAFFPA,HMETIS,PATOH,KAMINPAR,PARHIP,MT-KAHIP,PARKWAY-2,BIPART}.
The latter was shown to be more effective in reducing the
size of (hyper)graphs with highly-skewed node degree distributions~\cite{Abou-RjeiliK06,PARHIP} (e.g., social networks).
In the following, we present our parallel clustering-based coarsening algorithm that works similar to the shared-memory version
of \Partitioner{PaToH}'s coarsening scheme~\cite{PARALLEL-PATOH}.
However, the algorithm of \citet{PARALLEL-PATOH} excessively locks nodes when
evaluating the rating function and adding nodes to clusters. We therefore propose a less restrive locking protocol that completly
omits locking nodes when computing the best target cluster for a node.
Moreover, it detects and resolves conflicting clustering decisions \emph{on-the-fly}, while previous approaches relied on a postprocessing step~\cite{PARALLEL-PATOH,MT-METIS,MT-KAHIP-DIS}.

\subsection{The Clustering Algorithm}
Our coarsening algorithm repeatedly finds a clustering $\cluster$ of the nodes and subsequently contracts it
until the hypergraph is small enough. We represent the clustering $\cluster$ using an array $\rep$ of size $n$. We then choose
one representative $v \in C$ for each cluster $C \in \cluster$ and store $\rep[u] = v$ for each node $u \in C$.
Initially, each node is unclustered (i.e., $\rep[u] = u$ for each node $u \in V$). The clustering algorithm
then iterates over the nodes in parallel and assigns each unclustered node to the best target cluster
according to a rating function, which we introduce in the subsequent paragraph.

\paragraph{Cluster Join Operation}

\begin{figure}[!t]
  \centering
  \includegraphics[width=0.85\textwidth]{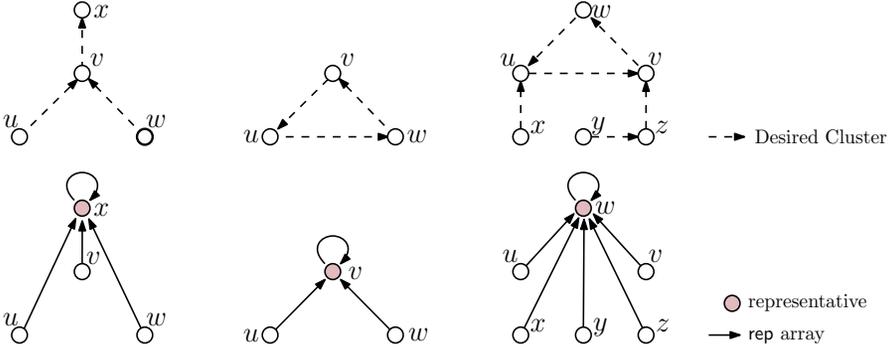}
  \caption{A path (top-left) and cyclic conflict (top-middle), and a combination of both conflicts (top-right)
           with their resolutions (bottom).}
  \label{fig:clustering_conflicts}
\end{figure}

Once a node $u$ chooses its desired target cluster $C$ represented by a node $v$, we have to set $\rep[u] = v$.
Since several nodes can join clusters simultaneously, there may occur conflicts that must
be resolved. As illustrated in Figure~\ref{fig:clustering_conflicts}, there are two types of conflicts:
\emph{path} and \emph{cyclic} conflicts. A path conflict involves several nodes $u_1, \ldots, u_l$ and occurs
when each node $u_i$ tries to join $u_{i+1}$. In a cyclic conflict, the last node $u_l$ additionally
tries to join $u_1$. It is also possible that a combination of both conflicts occurs, as illustrated in Figure~\ref{fig:clustering_conflicts}
(right). We can resolve a path conflict when each node $u_i$ waits until $u_{i+1}$ has joined its desired cluster.
Afterwards, we can set $\rep[u_i] = \rep[u_{i+1}]$ to resolve the conflict. However, applying this resolution scheme
to cyclic conflicts would result in a deadlock. Therefore, the threads must agree on a cluster join operation that
breaks the cycle and reduces it to a path conflict.

Algorithm~\ref{pseudocode:cluster_join_operation} shows the pseudocode of our cluster join operation, which takes a node $u$ as input, and adds it to a cluster represented by a node $v$.
The algorithm associates each node with one of the following three states: \emph{unclustered}, currently \emph{joining} a cluster, or \emph{clustered}.
Unclustered nodes ($\rep[u] = u$) can join clusters, while an already clustered node is not considered by the clustering algorithm
anymore and therefore its representative does not change. If a thread sets the state of a node $u$ from unclustered to joining via an
atomic \textfunc{compare-and-swap} operation, it acquires exclusive ownership for modifying $\rep[u]$ and setting its state to clustered.
Thus, if we succeed in setting the state of $u$ and $v$ to joining or $v$ is already clustered, we can safely set $\rep[u] = \rep[v]$
(see Line~\ref{cluster_join:simple_case_1}--\ref{cluster_join:simple_case_3}) since this guarantees that no other thread modifies $\rep[u]$ and $\rep[v]$.
Note that the representative of $v$ may have changed due to concurrent cluster join operations.
In that case, its representative is stored in $\rep[v]$. We therefore always set $\rep[u] = \rep[v]$ (instead of $\rep[u] = v$).

If another thread sets the state of $v$ to joining, we know that $v$ also tries to join a cluster.
To resolve the conflict, we spin in a busy-waiting loop until
the state of $v$ is updated to clustered (see Line~\ref{cluster_join:busy_waiting}), and then join its new cluster (path conflict).
In the busy-waiting loop, we additionally check if $u$ is part of a cycle of nodes trying to join each other.
To detect a cyclic conflict, each node writes its desired target cluster into a globally shared vector and checks if this induces a cycle.
If so, the node with the smallest ID in the cycle gets to join its desired cluster, thus breaking the cycle.

\begin{algorithm2e}[!t]
\KwIn{A node $u$ that wants to join $v$'s cluster}
\caption{Cluster Join Operation}\label{pseudocode:cluster_join_operation}

\If (\label{cluster_join:simple_case_1}) { $\atomicfunc{compare-and-swap}(\statearray[u], \statevar{Unclustered}, \statevar{Joining})$ } {
  \If (\label{cluster_join:simple_case_2}) { $\statearray[v] = \statevar{Clustered}$ \textbf{or} $\atomicfunc{compare-and-swap}(\statearray[v], \statevar{Unclustered}, \statevar{Joining})$ } {
    $\rep[u] \gets \rep[v]$\; \label{cluster_join:simple_case_3}
  }
  \Else (\tcp*[f]{Another thread tries to add $v$ to a cluster}) {
    \While (\tcp*[f]{busy-waiting loop}\label{cluster_join:busy_waiting}) { $\statearray[v] = \statevar{Joining}$ } {
      \If { cyclic conflict detected and $u$ is node with smallest ID in cycle } {
        $\rep[u] \gets \rep[v]\EndOfStatement \statearray[u], \statearray[v] \gets \statevar{Clustered}\EndOfStatement \textbf{break}$\;
      }
    }

    \lIf (\tcp*[f]{resolves path conflicts}) { $\statearray[u] = \statevar{Joining}$ } { $\rep[u] \gets \rep[v]$ }
  }
  $\statearray[u], \statearray[v] \gets \statevar{Clustered}$\;
}

\end{algorithm2e}

\paragraph{Rating Function}
A node $u$ joins the cluster $C$ maximizing the heavy-edge rating function
\[r(u,C) = \sum_{e \in \incnets(u) \cap \incnets(C)} \frac{\omega(e)}{|e| - 1}.\]
The rating function is commonly used in the partitioning literature~\cite{PATOH,HMETIS,KAHYPAR-K} and prefers clusters connected to $u$ via
a large number of heavy nets with small size. We evaluate the rating function by iterating
over the incident nets $e \in \incnets(u)$ and aggregating the ratings to the representatives $\rep[v]$ of each pin $v \in e$
in a thread-local hash table.
Afterwards, we iterate over the aggregated ratings and determine the representative $\rep[v]$ that maximizes $r(u,\rep[v])$.
Ties are broken uniformly at random. Subsequently, we perform the cluster join operation that sets $\rep[u] = \rep[v]$.

To aggregate ratings, we use fixed-capacity linear probing hash tables with $2^{15}$ entries and resort to a larger hash table
if the fill ratio exceeds $\nicefrac{1}{3}$ of the capacity. This technique can considerably reduce the number of cache misses
since most neighborhoods are small in real-world hypergraphs.
We further note that the representative of a node
can change during the evaluation of the rating function since we do not lock the nodes.
However, it has already been shown that such conflicts rarely happen in practice~\cite{PARALLEL-PATOH}
and therefore have a negligible impact on the partitioning result.

\paragraph{Contraction Limit}
We stop coarsening when the number of nodes in the smallest hypergraph reaches $160k$. This contraction limit
was chosen based on our prior research on sequential hypergraph partitioning~\cite{KAHYPAR-IP}. In addition, we terminate
the clustering algorithm when the number of nodes would drop below $\frac{c(V)}{2.5}$ after the contraction step.
This prevents the coarsening process from reducing the size of the hypergraph too aggressively~\cite{HMETIS,Abou-RjeiliK06}.
Conversely, we also stop coarsening if the contraction step does not reduce the number of nodes by more than $1\%$, even if the
$160k$ node limit is not reached. This can happen since we enforce an upper weight limit $c_{\max}$
on the weight of the heaviest cluster (set to $\frac{c(V)}{160k}$ as in \Partitioner{KaHyPar}~\cite{KAHYPAR-IP}),
which prevents highly-skewed node-weight distributions that would make it difficult for initial partitioning to find a balanced solution~\cite{KAHYPAR-K,PARHIP}.
When adding a node to a cluster $C \in \cluster$, we ensure that $c(C) \le c_{\max}$ by updating cluster weights via atomic \textfunc{fetch-and-add} instructions.
If $c(C) > c_{\max}$ after the update, we reject the corresponding cluster join operation and revert the cluster weight update.
The cluster weight limit can lead to coarsening passes that do not sufficiently reduce the size of the hypergraph.

\subsection{The Contraction Algorithm}

The hypergraph data structure stores the incident nets $I(u)$ of each node $u \in V$ and the pin-lists of each net $e \in E$
using two adjacency arrays. Each node $u$ and net $e$ additionally
stores its weight $c(u)$ and $\omega(e)$.
Contracting a clustering $\cluster = \{C_1,\ldots,C_l\}$ replaces each
cluster $C_i$ with one supernode $u_i$ with weight $c(u_i) = \sum_{v \in C_i} c(v)$. For each net $e \in E$,
we replace each pin $v \in e$ with the node $u_i$ representing the cluster $C_i$ in which $v$ is contained ($\rep[v] = u_i$).
After the replacement, multiple occurrences of the same supernode in a net are discarded.

Our contraction algorithm consists of several simple, easily parallelizable operations including remapping node IDs to a consecutive range,
aggregating cluster weights and degrees using atomic \textfunc{fetch-and-add} instructions,
eliminating duplicated entries in pin-lists, and using parallel prefix sum operations to construct the adjacency arrays of the contracted hypergraph.
As these steps are rather low level, we refer the reader to Ref.~\cite[p.~87]{HEUER-DIS} for more details.

A challenging aspect is removing duplicates from the set of nets.
We identify groups of identical nets and remove all but one representative per group to which we assign their aggregate weight.
This can reduce the number of pins significantly and therefore accelerates the other algorithmic components.
A simple algorithm is to perform pair-wise comparisons between all nets, which is however too expensive in practice.
To this end, we parallelize the \Algorithm{InrSrt} algorithm of Aykanat et al.~\cite{INR-Source, INR} for identical net detection.
It uses \emph{fingerprints} $f(e) := \sum_{v \in e} v^2$ to eliminate unnecessary pairwise comparisons between nets, by grouping nets with equal fingerprints via sorting.
Nets with different fingerprints or different sizes cannot be identical.
We distribute the fingerprints and their associated nets to the threads using a hash function.
Each thread sorts the nets by their fingerprint and size, and then performs pairwise comparisons on the subranges of potentially identical nets.
We aggregate the weights of identical nets at a representative and mark the others as invalid in a bitset.
A parallel prefix sum over the bitset maps the hyperedge IDs to a consecutive range in the contracted hypergraph.
Note that we also remove nets that contain only a single pin since they do not contribute to the cut.

\subsection{Community-Aware Coarsening}\label{sec:community_detection}

A popular approach to improve an existing $k$-way partition $\Partition$ is the \emph{iterated multilevel cycle} technique~\cite{WalshawVcycle} (also called \emph{V-cycle}).
In the coarsening phase, the algorithm forbids contractions between nodes that are not in the same block in $\Partition$, thus preserving the already identified cut structure. While the technique can be effective, using it as a postprocessing step in a multilevel algorithm almost doubles the running time.
As a more lightweight alternative, Heuer and Schlag~\cite{KAHYPAR-CA} proposed using a clustering of the nodes
computed via a community detection algorithm instead of an existing $k$-way partition.
Community detection still captures the sparse cut patterns that are often found in good $k$-way partitions.
The authors showed that this substantially improves the
quality of both the initial and the final partition, and only slightly increases the running time of the overall algorithm.

We also integrate the approach into our partitioning algorithm. We run the algorithm
as a preprocessing step before the coarsening phase and then use the clustering to restrict contractions
to nodes that belong to the same cluster.
The algorithm consists of two steps: transforming the hypergraph into its bipartite graph representation and then running
the parallel Louvain method of Staudt and Meyerhenke~\cite{PARALLEL-LOUVAIN,Louvain} for modularity
maximization, a widely used objective function for community detection~\cite{modularity.np, Newman04}.


\section{The Initial Partitioning Phase}\label{sec:initial_partitioning}

Partitioning algorithms based on the direct $k$-way partitioning scheme often use multilevel recursive
bipartitioning to obtain an initial $k$-way partition~\cite{DBLP:journals/jpdc/KarypisK98a,Schulz-DIS,KAHYPAR-K,kPaToH}, as this leads to partitions with significantly better solution quality than using flat (non-multilevel)
$k$-way partitioning methods. 
Many parallel partitioners run sequential initial partitioning algorithms in parallel~\cite{PARJOSTLE,PARMETIS,PARKWAY-1,PARKWAY-2,ZOLTAN,KAPPA,MT-KAHIP}. However, the sequential calls can become a bottleneck when the smallest
hypergraph is still large. A more scalable approach parallelizes the recursive calls after each bipartitioning operation~\cite{MT-METIS-OPT,PT-SCOTCH}.
The common approach is to statically split the thread pool along with the subproblems. Since this can lead to load imbalance when processing hypergraphs with unequal densities
in the recursive partitioning calls ,we instead generate tasks that can be dynamically load balanced using work stealing.

\paragraph{Parallel Recursive Bipartitioning}
We compute initial $k$-way partitions via parallel recursive bipartitioning using
Algorithm~\ref{pseudocode:multilevel} initialized with $k = 2$ (without flow-based refinement). For the bipartitioning case,
we replace the initial partitioning call with a portfolio of bipartitioning techniques.

Once we obtain a bipartition $\Partition = \{V_1,V_2\}$ of the input hypergraph $H$, we extract the subhypergraphs $\subhypergraph{V_1}$
and $\subhypergraph{V_2}$
and recurse on both in parallel by partitioning $\subhypergraph{V_1}$ into $\lceil \frac{k}{2} \rceil$
and $\subhypergraph{V_2}$ into $\lfloor \frac{k}{2} \rfloor$ blocks.
We ensure that the final $k$-way partition
obtained via recursive bipartitioning is $\varepsilon$-balanced by adapting the imbalance ratio for each bipartition individually~\cite{KaHyPar-R}.
Let $\subhypergraph{V'}$ be a subhypergraph that should be recursively partitioned into $k' \le k$ blocks. Then,
\begin{equation}
\label{eq:adaptive_epsilon}
\varepsilon' := \biggl( (1 + \varepsilon) \frac{c(V)}{k} \cdot \frac{k'}{c(V')} \biggr)^{\frac{1}{\lceil \log_2{k'} \rceil}} - 1
\end{equation}
is the imbalance ratio used for the bipartition of $H_{V'}$. If each bipartition is $\varepsilon'$-balanced, then it is guaranteed
that the final $k$-way partition is $\varepsilon$-balanced~\cite[Lemma~4.1 on p.~104]{KAHYPAR-DIS}.
%

\paragraph{Portfolio-Based Bipartitioning}
We implemented the same portfolio of initial bipartitioning techniques as in \Partitioner{KaHyPar} \cite{KaHyPar-R,KAHYPAR-IP},
including seven different variants of (greedy) hypergraph growing \cite{PATOH-MANUAL,HMETIS,PATOH,KarypisK98,Schulz-DIS,KaHyPar-R,KAHYPAR-DIS},
random assignment~\cite{PATOH-MANUAL,HMETIS,KaHyPar-R,KAHYPAR-DIS,MONDRIAAN}, and label propagation initial
partitioning~\cite{KaHyPar-R,KAHYPAR-DIS}.
We refer the reader to Ref.~\cite[p.~95--96]{HEUER-DIS} for more details on their implementation.
We run each algorithm independently in parallel for at least $5$ and at most $20$ times.
After $5$ runs, we only run an algorithm again if it is likely to improve the best solution $\Partition^*$ found so far.
We estimate this based on the arithmetic mean $\mu$ and standard deviation $\sigma$ of the connectivity values achieved by that algorithm so far, using the $95\%$ rule.
Assuming the connectivity values follow a normal distribution, roughly $95\%$ of the runs will fall between $\mu - 2 \sigma$ and $\mu + 2 \sigma$.
If $\mu - 2 \sigma > \ocon(\Partition^*)$, we do not run the algorithm again.
Additionally, we refine each bipartition using sequential $2$-way FM refinement~\cite{FM}.
We continue uncoarsening using the bipartition with the best connectivity value. In case of ties, we prefer the bipartition
with the best balance.

\section{Gain Computation Techniques}\label{sec:gain_computation}

Local search algorithms greedily move nodes to different blocks according to a \emph{gain value}.
The gain value reflects the change in the objective function for a particular node move. For the connectivity metric,
the gain $\gain{u}{t}$ of moving a node $u$ to a target block $V_t$
can be expressed as follows:
\[\gain{u}{t} := \omega(\{e \in \incnets(u) \mid \pinsinpart(e, \nodeblock{u}) = 1 \}) - \omega(\{e \in \incnets(u) \mid \pinsinpart(e, V_t) = 0 \}). \]
Moving node $u$ to block $V_t$ decreases the connectivity of all nets by one for which $u$ is the last remaining pin in its current block $\nodeblock{u}$.
Conversely, the move increases the connectivity of all nets $e \in \incnets(u)$ by one for which no pin $v \in e$ is assigned to the target block $V_t$.

\begin{figure}[!t]
  \centering
  \includegraphics[width=0.85\textwidth]{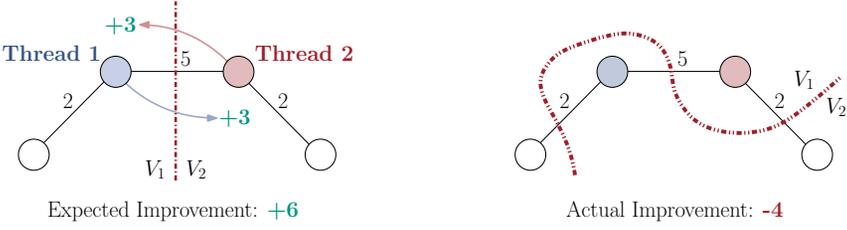}
  \caption{Example of a move conflict when two nodes are moved simultaneously. Both threads assume that the individual node moves
           removes the edge with weight $5$ from the cut. However, the edge is still cut after moving both nodes and
           the edges with weight $2$ become cut edges.}
  \label{fig:move_conflict}
\end{figure}

To achieve meaningful speedups, parallel refinement algorithms need to move nodes concurrently.
The actual gain of a node move can change between the time it is initially calculated and the time it is applied to the partition, due to concurrent node moves in its neighborhood~\cite{PARMETIS}.
As a consequence, two concurrent node moves can worsen the connectivity metric, even if their individual gains suggested an improvement, as illustrated in Figure~\ref{fig:move_conflict}.
Thus, correctly calculating gains is a fundamental challenge for parallel refinement algorithms.

These conflicts occur when two adjacent nodes change their blocks simultaneously. Common remedies include computing a node coloring and only moving nodes of the same color at a time~\cite{PARMETIS}, scheduling $2$-way refinement
algorithms on block pairs that form a matching in the quotient graph in parallel~\cite{PARJOSTLE,KAPPA},
allowing only node moves from a block $V_s$ to $V_t$ if $s < t$~\cite{MT-METIS,PARKWAY-2,ZOLTAN} (and vice versa in a second phase),
or following an optimistic strategy assuming that conflicts happen rarely in practice~\cite{PARHIP,MT-KAHIP,KAMINPAR,BIPART}.

The presented approaches still allow all individual node moves,
but combining arbitrary moves into a single move sequence might be not always possible.
This is problematic for parallelizing local search techniques as their sequential counterparts often identify a
set of moves that only yield an improvement if moved together.
While ignoring search conflicts appears to be the preferred approach, their impact on solution quality is unpredictable and deserves further consideration.

We therefore contribute several parallel gain computation techniques to compute accurate gain values and detect conflicts
between moves without restricting possible moves. We present a technique named \emph{attributed gains} to double-check the gain of a node move
in Section~\ref{sec:partition_data_structure}, a concurrent \emph{gain table} to accelerate gain calculations and communicate updates between threads in Section~\ref{sec:gain_table},
and a novel parallel algorithm for \emph{recomputing exact gains} of a sequence of node moves in Section~\ref{sec:fm:gain_recalc}.
These techniques build on our concurrent partition data structure which we describe in the next section in more detail.


\subsection{The Partition Data Structure}\label{sec:partition_data_structure}

Our partition data structure stores and maintains the block assigments $\Partition$, the block weights $c(V_i)$,
the pin count values $\pinsinpart(e,V_i)$, and
connectivity sets $\conset(e)$ for each net $e \in E$ and block $V_i \in \Partition$.

\paragraph{The Move Node Operation}


Algorithm~\ref{pseudocode:node_move} shows the updates to the partition data structure when moving a node $u$ from its source block $V_s$ to a target block $V_t$.
We only perform a node move if it does not violate the balance constraint, which we ensure by adding the weight
of node $u$ to the weight of block $V_t$ via an atomic \textfunc{fetch-and-add} instruction.
If the node move is feasible, we update the block assignment of node $u$ to block $V_t$ and subtract the node weight $u$ from
its previous block $V_s$.
If the move is infeasible, we subtract the weight again and reject the move.

\begin{algorithm2e}[!t]
	\KwIn{A node $u$ that should be moved from its source block $V_s$ to a target block $V_t$}
	\KwOut{Attributed gain value $\attrgain$}
	\caption{The Move Node Operation}\label{pseudocode:node_move}

	$c_t \gets \atomicfunc{fetch-and-add}(c(V_t), c(u))$\; \label{node_move:block_update_1}
	\If (\tcp*[f]{Revert block weight update if balance constraint violated}\label{node_move:balance_check}) { $c_t + c(u) > \balancedconstraint{\max}$ } {
		$c(V_t) \minuseqatomic c(u)$\EndOfStatement \Return{0}
	}
	$\nodeblock{u} \gets V_t$\EndOfStatement $c(V_s) \minuseqatomic c(u)$\EndOfStatement $\attrgain \gets 0$\;
	\For { $e \in I(u)$ } {
  $\FuncSty{lock}(e)$\EndOfStatement $\Phi_s \gets \Decrement\pinsinpart(e,V_s)$\EndOfStatement $\Phi_t \gets \Increment\pinsinpart(e,V_t)$\EndOfStatement$\FuncSty{unlock}(e)$\; \label{node_move:pin_count_update}
	\tcp*[l]{Update connectivity set $\conset(e)$ and attributed gain value $\attrgain$}
	\lIf (\label{node_move:pin_count_decrease}) { $\Phi_s = 0$ } { $\conset(e) \gets \conset(e) \setminus \{ V_s \}$\EndOfStatement $\attrgain \pluseq \omega(e)$ }
	\lIf (\label{node_move:pin_count_increase}) { $\Phi_t = 1$ } { $\conset(e) \gets \conset(e) \cup \{ V_t \}$\EndOfStatement $\attrgain \minuseq \omega(e)$ }
	$\FuncSty{UpdateGainTable}(e, \Phi_s, \Phi_t)$ \tcp*[r]{see Section~\ref{sec:gain_table}}\label{node_move:gain_table_update}
	}
	\Return{\attrgain}
\end{algorithm2e}

\paragraph{Data Layout}

The size of a pin count value is bounded by the size of the largest hyperedge.
To save memory, we use a packed representation with $\lceil \log({\max_{e \in E} |e|}) \rceil$ bits per entry
for the $\pinsinpart(e,V_i)$ values. Furthermore, we use a bitset of size $k$ to store the connectivity set
$\conset(e)$ of each hyperedge $e \in E$. We iterate over the connectivity set $\conset(e)$ by taking a snapshot
of its bitset and then use \emph{count-leading-zeroes} instructions. 
We compute the connectivity $\con(e) = |\conset(e)|$ of a hyperedge $e$ using \emph{pop-count} instructions
(counts the number of $1$-bits in a machine word).
To add or remove a block from the connectivity set, we flip the corresponding bit using an atomic \textfunc{xor} operation.
The move node operation can be made lock-free by updating $\pinsinpart(e, V_i)$ with atomic \textfunc{fetch-and-add} instructions, but this requires one machine word per value.
We therefore use a spin-lock for each net $e$ due to the packed representation.
%

\paragraph{Attributed Gains}
As the gain value of a node move can change between its initial calculation and actual execution due to concurrent node moves in its neighborhood,
we additionally compute an \emph{attributed gain} value for each move based on the atomic updates
of the pin count values $\pinsinpart(e,V_s)$ and $\pinsinpart(e,V_t)$ in Line~\ref{node_move:pin_count_update}
of Algorithm~\ref{pseudocode:node_move}.
We attribute a connectivity decrease by $\omega(e)$ to the move that reduces $\pinsinpart(e, V_s)$ to zero (see Line~\ref{node_move:pin_count_decrease})
and an increase by $\omega(e)$ for increasing $\pinsinpart(e,V_t)$ to one (see Line~\ref{node_move:pin_count_increase}).

Since we do not lock all incident nets $e \in \incnets(u)$ before moving a node $u$,
there is no guarantee on the order in which concurrent moves perform the pin count updates.
Hence, this scheme may distribute the connectivity reductions to different threads,
but the sum of the attributed gains of all node moves equals the overall connectivity reduction~\cite{HEUER-DIS}.

\paragraph{Attributed Gains for Label Propagation Refinement}

The most widely used refinement technique in parallel partitioning
algorithms is label propagation~\cite{JOSTLE,PARMETIS,PARHIP,MT-KAHIP,KAMINPAR,PARKWAY-2,BIPART}.
The algorithm works in rounds. In each round, it iterates over all nodes in parallel, and whenever it
visits a node $u$, it moves it to the block $V_t$ maximizing its move gain $\gain{u}{t}$ (respecting the balance constraint).
The algorithm only performs moves with positive gain and therefore cannot escape from local optima.
However, we use it in our partitioning algorithm to find all \emph{simple} node moves such that our more advanced
refinement techniques can focus on finding non-trivial improvements (for more technical details on its implementation,
see Ref.~\cite[p.~68--69]{HEUER-DIS}).

Since the label propagation algorithm performs only positive gain moves, we immediately revert a node move
if it has negative attributed gain.
Note that reverting such a node move
does not guarantee to improve the connectivity metric again as other concurrent node moves may have
changed the pin count values of the corresponding nets in the meantime. However, reverting them directly after detection
decreases the likelihood of such conflicts.
Furthermore, we use attributed gains to track the value of the connectivity metric instead of recomputing it after each round.

\subsection{The Gain Table}\label{sec:gain_table}

For our FM algorithm, we use a gain table which stores and maintains the gain values for all possible moves.
This enables repeatedly looking up gains in $O(1)$ time and is a globalized way of updating the gains of nodes owned by other threads.
Gain tables are not a new idea \cite{FromPQs, KAHYPAR-K} but have gone \say{out of fashion}
due to their memory requirements \cite{KAFFPA, MT-KAHIP}.
To the best of our knowledge, our introduction of \emph{parallel} gain tables is novel.

We use atomic \textfunc{fetch-and-add} instructions to update the gains as soon as nodes are moved.
Updates on some nodes become visible while the overall update procedure is still in flight.
Therefore, updates \emph{trickle in} over time, and some outdated or inconsistent values may be read by other threads.
Still, with concurrent node moves this is the most accurate we can be.

Recall that the gain $\gain{u}{t}$ of moving a node $u$ to a target block $V_t$ can be expressed
as follows:
\[\gain{u}{t} := \omega(\{e \in \incnets(u) \mid \pinsinpart(e, \nodeblock{u}) = 1 \}) - \omega(\{e \in \incnets(u) \mid \pinsinpart(e, V_t) = 0 \}). \]
The first term $b(u) \coloneq \omega(\{e \in \incnets(u) \mid \pinsinpart(e, \nodeblock{u}) = 1 \})$ is the \emph{benefit} of moving $u$ out of its block.
Conversely, the term $p(u, V_t) \coloneq \omega(\{e \in \incnets(u) \mid \pinsinpart(e, V_t) = 0 \})$ is the \emph{penalty} for moving $u$ into $V_t$.

\paragraph{Update Rules}

Instead of storing $\gain{u}{i}$, we store $b(u)$ and $p(u, V_i)$ separately for each node $u$, so that changes to $b(u)$ only require one update,
instead of updates to $k$ gain values.
This approach uses $(k+1)n$ memory words in total.
For each net $e \in I(u)$, we update $b(u)$ and $p(u,V_i)$ using atomic \textfunc{fetch-and-add} instructions as follows.

\begin{alignat}{4}
  \setcounter{equation}{0}
  & \text{\textbf{If} } \pinsinpart(e, V_s) = 0 \text{ \textbf{then} } \forall v \in e           & \text{ \textbf{do} } & p(v, V_s) &~\pluseqatomic~& \omega(e) \\
  & \text{\textbf{If} } \pinsinpart(e, V_s) = 1 \text{ \textbf{then} } \forall v \in e \cap V_s  & \text{ \textbf{do} } & b(v)      &~\pluseqatomic~& \omega(e) \\
  & \text{\textbf{If} } \pinsinpart(e, V_t) = 1 \text{ \textbf{then} } \forall v \in e           & \text{ \textbf{do} } & p(v, V_t) &~\minuseqatomic~& \omega(e) \\
  & \text{\textbf{If} } \pinsinpart(e, V_t) = 2 \text{ \textbf{then} } \forall v \in e \cap V_t  & \text{ \textbf{do} } & b(v)      &~\minuseqatomic~& \omega(e)
\end{alignat}

The update conditions implement the \textfunc{UpdateGainTable} procedure from Line~\ref{node_move:gain_table_update} in Algorithm~\ref{pseudocode:node_move}.

\paragraph{Benefit Pecularities}
There is a race condition on $\nodeblock{v}$ in the check $\nodeblock{v} = V_s$ (case 2) or $\nodeblock{v} = V_t$ (case 4).
When $\nodeblock{v}$ changes, we may perform a benefit update on $v$ that was also intended for a different pin of $e$ in the new $\nodeblock{v}$.
The penalty values are not affected since they are independent of the pin's current block.
Our FM algorithm is organized in rounds in which each node can be moved at most once.
Therefore, once $u$ gets moved, we do not read $b(u)$ for the rest of the round.
Due to the race condition it may still be updated, which is why we recalculate $b(u)$ \textit{after} the round is finished instead of recalculating $b(u)$ for the new block \emph{immediately} after the move.
We note that it is possible to correctly update benefits by using $k$ benefit values per node~\cite{GOTT-DIS}.

\paragraph{Correctness and Complexity of Gain Updates}


In the following, we prove that once all updates for a given set of moves are completed and no further moves are performed, the gain values are correct.

\begin{lemma}
 After performing all gain updates associated with a set of moves $M$ in parallel, each unmoved node $v \in V \setminus M$ has correct $b(v)$, and each $v \in V$ has correct $p(v, V_i)$ terms.
\end{lemma}
\begin{proof}
First, we note that the updates are correct in the sequential setting~\cite{sanchis1989mwn}.
Due to the atomic consistency of pin-count and gain updates, it suffices to prove correctness for arbitrary linearized (sequential) orders of updates.
The remaining difficulty is that different orders may yield different intermediate values.
However, due to commutativity we arrive at the same final $\pinsinpart(e,V_i)$ values.
Thus, it suffices to argue that gain updates triggered by $\pinsinpart(e, V_i) \pluseq 1$ cancel out those triggered by $\pinsinpart(e, V_i) \minuseq 1$.
This statement holds, as case 1 and 3 are complimentary, as well as case 2 and 4.
Therefore, the final $p(v, V_i)$ and $b(v)$ values only depend on the final $\pinsinpart(e, V_i)$ values.
\end{proof}

\begin{lemma}[Sanchis~\cite{sanchis1989mwn}]\label{lemma:gain-updates}
	The work of gain updates for moving all nodes once is $\Oh{\sum_{e \in E}|e| \cdot \min(k, |e|)} = \Oh{kp}$.
\end{lemma}

The core to the argument is that each of the update cases is only triggered a constant number of times per hyperedge and block~\cite{sanchis1989mwn,FM}, and costs $\mathcal{O}(|e|)$ work per update.
This hinges on moving each node at most once.
Note that due to the $\min(k, |e|)$ term, this bound matches the $\mathcal{O}(m)$ bound on plain graphs.
On real-word hypergraphs, we observed work much closer to $\Oh{p}$ since most nets have small size or few pins per block. 


\subsection{The Parallel Gain Recalculation Algorithm}\label{sec:fm:gain_recalc}

We now propose a parallel algorithm to recompute exact gain values of a sequence of node moves $M = \langle m_1, \dots, m_l \rangle$
if they are supposed to be performed in this order. Each move $m_i \in M$ is of the form $m_i = (u,V_s,V_t)$, which means
that node $u$ is moved from block $V_s$ to $V_t$.
Again, we assume that each node is moved at most once.
Recall that a move of a node $u$ from block $V_s$ to $V_t$ decreases the connectivity of a hyperedge $e$,
if $\pinsinpart(e,V_s)$ decreases to zero. Conversely, it increases the connectivity if $\pinsinpart(e,V_t)$
increases to one. The idea of the following algorithm is to iterate over the hyperedges in parallel, and
identify the node moves in $M$ that increase or decrease the connectivity of a hyperedge using Algorithm~\ref{pseudocode:recalculate_gain}.

Consider a hyperedge $e$ and a block $V_i \in \Partition$. The first observation is that if we move a pin $v \in e$ to $V_i$,
then $\pinsinpart(e,V_i)$ cannot decrease to zero anymore since each node is moved at most once. In order to decrease $\pinsinpart(e,V_i)$
to zero, we have to move all pins $u \in e \cap V_i$ out of block $V_i$ before we move the first pin $v \in e \setminus V_i$ to block $V_i$.
In this case, the last pin $u \in e$ moved out of block $V_i$ decreases the connectivity of $e$ and the first pin $v \in e$
moved to block $V_i$ increases its connectivity again. Thus, we can decide whether or not a move increases or decreases the
connectivity of a hyperedge by simply comparing the indices of the node moves in $M$, which were last moved out and first moved
to a particular block. Additionally, we need to know if the move sequence $M$ moves all pins out of block $V_i$.
To do so, we count the number of non-moved pins $v \in e$ in each block. If the number of non-moved pins is zero for a block $V_i$, then
either $\pinsinpart(e,V_i)$ was zero before, or the move sequence $M$ moved all nodes out of block $V_i$.

Algorithm~\ref{pseudocode:recalculate_gain} shows the pseudocode that identifies the node moves in $M$ that increase or decrease the connectivity
of a hyperedge $e$. The algorithm uses two loops, both iterating over the pins of hyperedge $e$.
The first loop computes the indices of the node moves
that first moved to and last moved out of each block $V_i \in \Partition$ (see Line~\ref{recalculate_gain:first_in_last_out}), in addition to the number of pins in $e$ that were not moved (see Line~\ref{recalculate_gain:non_moved_pins}).

The second loop then decides for each moved pin $u \in e$ whether or not it increases or decreases the connectivity of hyperedge $e$ by evaluating
the conditions shown in Lines~\ref{recalculate_gain:decrease_connectivity} and~\ref{recalculate_gain:increase_connectivity}.
Let $m_i := (u,V_s,V_t)$ be the corresponding node move of pin $u \in e$ in $M$.
If $M$ moves all nodes out of block $V_s$ ($\matharray{non\_moved}[V_s] = 0$) and $u$ is the last pin moved out
of block $V_s$ ($\matharray{last\_out}[V_s] = i$), while the first move that moves a pin into block $V_s$ happens strictly after
$m_i$ ($i < \matharray{first\_in}[V_s]$), then $m_i$ reduces the connectivity metric by $\omega(e)$.
Conversely, if $M$ moves all nodes out of block $V_t$ ($\matharray{non\_moved}[V_t] = 0$) and $u$ is the first pin
moved into block $V_t$ ($\matharray{first\_in}[V_t] = i$), while the last move that moves a pin out of block $V_t$ happens strictly
before $m_i$ ($i > \matharray{last\_out}[V_t]$), then $m_i$ increases the connectivity metric by $\omega(e)$.
Since we run the algorithm for each hyperedge in parallel, several threads can modify the gain value $g_i$ of a node move $m_i$ simultaneously.
We therefore use atomic \textfunc{fetch-and-add} instructions (see Line~\ref{recalculate_gain:decrease_connectivity} and~\ref{recalculate_gain:increase_connectivity}).

To further reduce the complexity of the algorithm, we only process hyperedges containing moved nodes. To do so, we iterate
over the node moves in $M$ in parallel and run Algorithm~\ref{pseudocode:recalculate_gain} only for incident edges
of moved nodes. We mark already processed hyperedges in a shared bitset using atomic \textfunc{test-and-set} instructions.

\begin{algorithm2e}[!t]
    \KwIn{Hyperedge $e$, a sequence of node moves $M = \langle m_1, \dots, m_l \rangle$ and
      a shared gain vector $\mathcal{G} = \langle g_1, \dots, g_l \rangle$ representing the recalculated gain values }
    \caption{Parallel Gain Recalculation}\label{pseudocode:recalculate_gain}

    $\matharray{first\_in} \gets [\infty,\ldots,\infty]\EndOfStatement \matharray{last\_out} \gets [-\infty,\ldots,-\infty]$ \tcp*[r]{Arrays of size $k$}
    $\matharray{non\_moved} \gets [0,\ldots,0]$ \tcp*[r]{Array of size $k$}
    \For { $u \in e$ } {
      \If (\tcp*[f]{moved nodes are marked in a bitset}) { $u$ was moved } {
        $m_i := (u,V_s,V_t) \gets$ find corresponding move in $M$\;
        $\matharray{last\_out}[V_s] \gets \max(i, \matharray{last\_out}[V_s])$\EndOfStatement $\matharray{first\_in}[V_t] \gets \min(i, \matharray{first\_in}[V_t])$\;\label{recalculate_gain:first_in_last_out}
      } \lElse {
        $\Increment\matharray{non\_moved}[\nodeblock{u}]$\label{recalculate_gain:non_moved_pins}
      }
    }

    \For { $u \in e$ } {
      \If { $u$ was moved } {
        $m_i := (u,V_s,V_t) \gets$ find corresponding move in $M$\;
        \lIf (\label{recalculate_gain:decrease_connectivity}) { $\matharray{last\_out}[V_s] = i~\land~i < \matharray{first\_in}[V_s]~\land~\matharray{non\_moved}[V_s] = 0$ } {
          $g_i \pluseqatomic \omega(e)$
        }
        \lIf (\label{recalculate_gain:increase_connectivity}) { $\matharray{first\_in}[V_t] = i~\land~i > \matharray{last\_out}[V_t]~\land~\matharray{non\_moved}[V_t] = 0$ } {
          $g_i \minuseqatomic \omega(e)$
        }
      }
    }
\end{algorithm2e}

\section{The Fiduccia-Mattheyses Algorithm}\label{sec:fm}

The Fiduccia-Mattheyses (FM) algorithm~\cite{FM} is the most widely used local search algorithm in sequential partitioning algorithms.
Most of the existing variants insert all possible moves or only the highest gain move for each boundary node into a PQ (boundary FM)
and then perform the following two steps:
(i) repeatedly perform the highest gain move subject to the balance constraint, followed by (ii) reverting moves back to the prefix
with the highest cumulative gain in the sequence of performed moves.
The revert is necessary, since moves with negative gains are allowed, so the algorithm is able to escape from local minima.
Unfortunately, calculating the same move sequence as FM is P-hard~\cite{PHARD}, i.e., it is unlikely that a parallel algorithm with poly-log depth exists.

Sanders and Schulz~\cite{KAFFPA} proposed a relaxed version that inserts only the highest gain move for a single seed node into a PQ and then
gradually expands around the node by claiming neighbors of moved node (localized FM). The algorithm not only produces better solutions
than boundary FM~\cite{DBLP:journals/tcad/HagenHK97}, it is also highly amenable to parallelization as multiple FM searches can
run in parallel, each starting from a different seed node. In the following, we present our parallel implementation of the localized FM algorithm, and
discuss its main differences to an existing parallelization~\cite{MT-KAHIP}.


\paragraph{The Parallel $k$-Way FM Algorithm}
Algorithm~\ref{pseudocode:fm_algo} shows the pseudocode of our parallel FM algorithm.
The algorithm proceeds in rounds, and each round starts with inserting all boundary nodes into a globally
shared task queue $Q$.
The threads then poll a fixed number of nodes ($= 25$) from $Q$ that they use as seed nodes for the
localized FM searches, which expand to neighbors of moved nodes.

The searches are non-overlapping, i.e., threads acquire exclusive ownership of nodes,
while hyperedges can touch multiple searches.
Node moves performed by the different searches are not visible to other threads, as they are performed locally
using thread-local hash tables.
However, once a thread finds an improvement, it immediately applies it to the global partition.
The local moves are atomically appended to a global move sequence (using one atomic \textfunc{fetch-and-add} for all local moves).
We repeatedly start localized FM searches until the task queue is empty.
Note that we initialize the searches with multiple seed nodes instead of a single node as this substantially
accelerates the algorithm in practice without sacrifices in solution quality.

Once the task queue is empty, we proceed to the second phase, where we recalculate the gains of the global move sequence
(see Section~\ref{sec:fm:gain_recalc}) and then use a parallel prefix
sum and reduce operation on the recomputed gain values to identify and revert to the best seen solution.
We perform multiple rounds until a maximum number is reached or the connectivity
metric is not improved.

\paragraph{Localized $k$-Way FM Search}


The localized FM search uses a single PQ storing the move with the highest gain for each inserted node.
We initialize the PQ with several seed nodes
and use the gain table to compute the initial best move for each node (see Line~\ref{fm_algo:init}).
Then, we repeatedly select the move with the highest gain 
and apply it to a thread-local partition $\DeltaPartition$. 
Changes on $\DeltaPartition$ are not visible to other threads for now. However, we apply the move sequence
to the global partition $\Partition$ as soon as we find an improvement (see Line~\ref{fm_algo:apply_global}), then triggering gain updates in the global gain table.

When we move a node $u$ locally, we collect the nets $e \in \incnets(u)$ affected
by gain updates. We use them to update the gain values of nodes in the PQ --
combining global gain table and $\DeltaPartition$ data, thus gradually infusing updates from other threads into the search
-- and expand the search to neighbors of moved nodes. 
A localized search terminates when the PQ becomes empty or the adaptive stopping rule of Osipov and Sanders~\cite{kaspar,KAHYPAR-K}
is triggered. The stopping rule assumes that the observed gain values follow a normal distribution and terminates a search
when it becomes unlikely to find further improvements.
We release the ownership of non-moved nodes at the end such that other searches
can acquire them again. We do not release the ownership of moved nodes to
ensure that each node is moved at most once during an FM pass.

\begin{algorithm2e}[!t]
    \caption{Parallel $k$-Way FM Algorithm}\label{pseudocode:fm_algo}

    \Function (\label{fm_algo:parallel_fm}) {$\FuncSty{FMRefinement}($\emph{Hypergraph} $H,$ \emph{$k$-way partition} $\Partition)$} {
      \While{improvement found and maximum number of rounds not reached}{
        $Q \gets$ initialize task queue with all boundary nodes\;
        \ParallelWhile { $Q$ not empty } {
          $V_{\text{\footnotesize seed}} \gets$ poll $25$ seed nodes from $Q$ \;
          $\FuncSty{LocalizedFMRefinement}(H,\Partition,V_{\text{\footnotesize seed}})$\;
        }
        recompute gains of global move sequence and revert to best prefix \tcp*[r]{see Section~\ref{sec:fm:gain_recalc}} \label{fm_algo:parallel_revert}
      }
    }
    \;
    \Function {$\FuncSty{LocalizedFMRefinement}(H,\Partition,V_{\text{\footnotesize seed}})$} {
      \For (\tcp*[f]{Initialize PQ with seed nodes}\label{fm_algo:init}) { $u \in V_{\text{\footnotesize seed}}$ } {
        $(g_t,V_t) \gets \FuncSty{ComputeMaxGainMove}(u)$\EndOfStatement
        PQ.$\FuncSty{Insert}(u, g_t, V_t)$\;
      }
      $\Delta \gets 0\EndOfStatement M \gets \emptyset$\;
      \While { PQ not empty \emph{and} search should continue }{
        $(u,g_t,V_t) \gets$ PQ.$\FuncSty{PopMaxGainMove}()$\;\label{fm_algo:extract_2}
        move $u$ to $V_t$ in thread-local partition $\DeltaPartition$ with gain table update\; \label{fm_algo:apply_thread_local}
        $\Delta \gets \Delta + g_t\EndOfStatement M \gets M \cup \{(u,V_t)\}$\;
        \If { $\Delta > 0$ \emph{or} ($\Delta = 0$ \emph{and} move improved balance) } {
          apply move sequence $M$ to global partition $\Partition$\;\label{fm_algo:apply_global}
          $\Delta \gets 0\EndOfStatement M \gets \emptyset\EndOfStatement \DeltaPartition \gets \Partition$\; \label{fm_algo:reset}
        }

        \For (\label{fm_algo:update_1}) { all nets $e \in I(u)$ involved in a gain update } {
          \For { $v \in e$ } {
            \If { $v$ is not marked } {
              \lIf (\label{fm_algo:update_2}) { PQ.$\FuncSty{Contains}(v)$ } { update gain of $v$ in PQ }
              \ElseIf (\label{fm_algo:acquire}) { try to acquire node $v$ } {
                $(g_t,V_t) \gets \FuncSty{ComputeMaxGainMove}(v)$\EndOfStatement PQ.$\FuncSty{Insert}(v, g_t, V_t)$\; \label{fm_algo:update}
              }
              mark $v$\;
            }
          }
        }
        unmark all nodes\;
      }
    }

\end{algorithm2e}

We explicitly allow moves with negative gains, which will worsen the solution quality intermediately.
At the end of each localized search, we thus revert back to the best seen solution.
If we directly applied moves to the global partition, other searches could base their decisions on states that will later be reverted.
Therefore, we apply node moves to a thread-local partition $\DeltaPartition$ first and only perform them on the global partition if they lead to an improvement, i.e., will not be reverted for now.

The thread-local partition $\DeltaPartition$ stores changes relative to the global partition in a set of hash tables.
For example, we compute the weight of a block $V_i$ by calculating $c(V_i) + \Delta c(V_i)$ where $c(V_i)$
is the weight of block $V_i$ stored in the global partition data structure and $\Delta c(V_i)$ is the weight
of all nodes that locally moved to block $V_i$ minus the weight of nodes that moved out of block $V_i$. We maintain
the block ID, pin count values, as well as benefit and penalty terms of the gain table analogously.

Applying a move sequence to the global partition makes it immediately visible to the searches on other threads.
Since $\DeltaPartition$ stores local changes relative to the global partition, the block weights and pin count values
are still correct. However, some gain values may be incorrect since the gain table updates on the global partition
do not consider moves performed locally.
This is only a small issue since thread-local deltas are cleared after applying the moves to the global partition.
In practice, the scheme drastically reduces conflicts\footnote{We have found that the recomputed gain values of
the global move sequence match the observed gain values during the localized FM searches in most cases.}.
Another reason for applying moves as soon as possible is to keep the memory footprint of the hash tables small.
The overall peak memory incurred by thread-local partition data is small, because the memory is proportional to the number of moves,
and long-running searches must find improvements to keep going.

\paragraph{Differences to \Partitioner{Mt-KaHIP}}
The FM implementation in \Partitioner{Mt-KaHIP}~\cite{MT-KAHIP} performs node moves only locally, which are therefore \emph{not
visible to other threads}. At the end of an FM pass, the move sequences found by the different searches are concatenated to a global
move sequence, for which \emph{gains are recomputed sequentially}. We improved the algorithm by making improvements immediately
visible to other threads using the thread-local partition and gain table data structure, leading to more accurate gain values.
Moreover, we removed the last sequential part of the algorithm with our parallel gain recomputation technique.

\section{Flow-Based Refinement}\label{sec:flows}
A major shortcoming of move-based local search algorithms is that they greedily move nodes to other
blocks based on a gain value considering only the block assignment of adjacent nodes.
Thus, the decision to apply a move depends only on \emph{local} information, which may not be sufficient to find some non-trivial improvements~\cite{Saab95}.
Maximum flows overcome this limitation by deriving a minimum cut
seperating two nodes~\cite{MINCUT-THEOREM} and therefore have a more \emph{global} view on the partitioning problem.
Although it seems natural to use them as local search strategy in partitioning algorithms, maximum flows were long perceived as computationally expensive and it was unclear how to derive balanced partitions~\cite{KL,yang-wong-fbb}.
This changed over the last two decades as flow-based refinement techniques were successfully implemented in the highest-quality
sequential graph and hypergraph partitioning algorithms~\cite{KAFFPA,KAHYPAR-MF,KAHYPAR-HFC,KAHYPAR-FLOW,KAHYPAR-JOURNAL}.
Today it is considered to be the most powerful improvement heuristics for (hyper)graph partitioning.
However, since they come at the cost of substantially higher running times, they can be impractical for partitioning very large hypergraphs.

\paragraph{Algorithm Overview}

In this section, we present the first parallel formulation of the sequential flow-based refinement approach
used in \Partitioner{KaHyPar}~\cite{KAHYPAR-MF,KAHYPAR-HFC}.
The high-level pseudocode of the algorithm is outlined in Algorithm~\ref{pseudocode:flow_refinement}.
Flow-based refinement works on bipartitions and can be scheduled on different block pairs to improve $k$-way partitions~\cite{KAFFPA,KAHYPAR-HFC,KAHYPAR-MF}.
We therefore start with a parallel scheduling scheme of adjacent block pairs based on the quotient graph in
Section~\ref{sec:scheduling} (see Line~\ref{pseudocode:build_quotient_graph} and~\ref{pseudocode:active_block_scheduling}).
In Section~\ref{sec:network_construction}, we describe the flow network construction algorithm that extracts a subhypergraph induced by a region $B \subseteq V$ around the boundary nodes
of two adjacent blocks, which then yields a flow network (see Line~\ref{pseudocode:construct_region} and~\ref{pseudocode:construct_flow_network}).
On each network, we run the \Partitioner{FlowCutter} algorithm~\cite{FLOW-CUTTER, yang-wong-fbb} to derive a balanced minimum cut using incremental maximum flow computations.
\Partitioner{FlowCutter} and its parallelization are discussed in Sections~\ref{sec:flowbased_refinement} and~\ref{sec:parallel_flow_algo}.
We then convert the minimum cut into a set of moves $M$ and an expected connectivity reduction $\expgain$.
If \Partitioner{FlowCutter} claims an improvement, i.e., if $\expgain \ge 0$, we apply the moves to the global partition and compute the exact reduction $\attrgain$,
based on which we either mark the blocks for further refinement, or revert the moves (see Line~\ref{pseudocode:active_block_scheduling_3} and~\ref{pseudocode:revert_moves}).

\begin{algorithm2e}[!t]
    \KwIn{Hypergraph $H = (V,E,c,\omega)$ and $k$-way partition $\Partition$ of $H$}
    \caption{Parallel Flow-Based Refinement}\label{pseudocode:flow_refinement}

    $\quotientgraph \gets \FuncSty {BuildQuotientGraph}(H, \Partition)$ \tcp*[r]{see Section~\ref{sec:scheduling}}  \label{pseudocode:build_quotient_graph}
    \ParallelWhile (\tcp*[f]{see Section~\ref{sec:scheduling}}) {$\exists$ active $(V_i,V_j) \in \quotientgraph$} { \label{pseudocode:active_block_scheduling}
      $B := \gets \FuncSty {ConstructRegion}(H,V_i,V_j)$ \tcp*[r]{see Section~\ref{sec:network_construction}} \label{pseudocode:construct_region}
      $(\flowhypergraph,\source,\sink) \gets \FuncSty {ConstructFlowNetwork}(H,B)$ \tcp*[r]{see Section~\ref{sec:network_construction}} \label{pseudocode:construct_flow_network}
      $(M, \expgain) \gets \FuncSty {FlowCutterRefinement}(\flowhypergraph,\source,\sink)$ \tcp*[r]{see Section~\ref{sec:flowbased_refinement} --~\ref{sec:parallel_flow_algo}} \label{pseudocode:flow_cutter_call}
      \If (\tcp*[f]{potential improvement}) { $\expgain \ge 0$ } {
        $\attrgain \gets \FuncSty {ApplyMoves}(H,\Partition,M)$ \tcp*[r]{see Section~\ref{sec:scheduling}} \label{pseudocode:apply_moves}
        \lIf (\tcp*[f]{found improvement}) { $\attrgain > 0$ } {
          mark $V_i$ and $V_j$ as active \label{pseudocode:active_block_scheduling_3}
        }
        \lElseIf (\tcp*[f]{no improvement}) {$\attrgain < 0$} { $\FuncSty{RevertMoves}(H, \Partition, M)$ \label{pseudocode:revert_moves} }
      }
    }
\end{algorithm2e}

\paragraph{Maximum Flows}
A flow network $\flownetwork = (\flownodeset,\flowedgeset,\capacity)$
is a directed graph with a dedicated source node
$\source \in \flownodeset$ and sink node $\sink \in \flownodeset$ in which
each edge $e \in \flowedgeset$ has capacity $\capacity(e) \ge 0$. An $\sourcesink$-flow is a function
$f: \flownodeset \times \flownodeset \rightarrow \mathbb{R}$
that satisfies the \emph{capacity constraint} $\forall u,v \in \flownodeset: f(u,v) \le \capacity(u,v)$,
the \emph{skew symmetry constraint} $\forall u,v \in \flownodeset: f(u,v) = -f(v,u)$ and the
\emph{flow conservation constraint} $\forall u \in \flownodeset \setminus \{\source,\sink\}: \sum_{v \in \flownodeset} f(u,v) = 0$.
The value of a flow $|f| := \sum_{v \in \flownodeset} f(\source,v) = \sum_{v \in \flownodeset} f(v,\sink)$ is defined as the
total amout of flow transferred from $\source$ to $\sink$.
An $\sourcesink$-flow $f$ is a maximum $\sourcesink$-flow if there exists no other $\sourcesink$-flow $f'$ with $|f| < |f'|$.
The \emph{residual capacity} is defined as $r_f(e) = \capacity(e) - f(e)$.
An edge $e$ is \emph{saturated} if $r_f(e) = 0$.
The \emph{residual network} $\flownetwork_f = (\flownodeset, \flowedgeset_f, r_f)$ with
$\flowedgeset_f := \{(u,v) \in \flownodeset \times \flownodeset \mid r_f(u,v) > 0\}$ contains all non-saturated edges.
The max-flow min-cut theorem states that the value $|f|$ of a maximum $\sourcesink$-flow equals the weight
of a minimum cut that separates $\source$ and $\sink$~\cite{MINCUT-THEOREM}.
This is also called a \emph{minimum $\sourcesink$-cut}.
The minimum $\sourcesink$-cut can be derived by exploring the nodes reachable from the source or sink via
residual edges ($r_f(e) > 0$), which is also called the \emph{source-side} or \emph{sink-side cut}.

\subsection{Parallel Active Block Scheduling}\label{sec:scheduling}

Sanders and Schulz~\cite{KAFFPA} propose the active block scheduling strategy to apply
their flow-based refinement algorithm for bipartitions on $k$-way partitions.
Their algorithm proceeds in rounds.
In each round, it schedules all pairs of adjacent blocks where at least one is marked as \emph{active}.
Initially, all blocks are marked as active.
If a search on two blocks finds an improvement, both are marked as active for the next round.

\paragraph{Parallelization}
Our parallel implementation schedules multiple flow computations on adjacent block pairs in parallel.
We do not enforce any constraints on the block pairs processed concurrently, e.g., there can
be multiple threads running on the same block and they can also share some of their nodes.
We use $\min(t,\tau \cdot k)$ threads to process the active block pairs in parallel,
where $t$ is the number of available threads in the system and
the parameter $\tau$ controls the available parallelism in the scheduler.
With higher values of $\tau$, more block pairs are scheduled in parallel, which
can lead to more interferences between searches that operate on overlapping regions.
Threads that are not involved in scheduling can join parallel flow computations.
In a parameter study~\cite[p.~108]{HEUER-DIS}, we found that $\tau = 1$ offers a good
trade-off between conflicting searches and scalability.

Initially, we push all pairs of adjacent blocks into a concurrent FIFO queue $\activequeue$.
The threads then poll from $\activequeue$ and if a search finds an improvement on a block pair $(V_i,V_j)$,
we mark both as active using a seperate bitset for each round. If either $V_i$ or $V_j$ becomes active, we
push all adjacent blocks into $\activequeue$ if they are not contained yet.
Thus, active block pairs of different rounds are stored interleaved in $\activequeue$ and the end of a
round does not induce a synchronization point as in the original algorithm~\cite{KAFFPA}.
A round ends when all of its block pairs have been processed and all prior rounds have ended.
If the relative improvement at the end of a round is less than $0.1\%$, we immediately terminate the algorithm.


\paragraph{Apply Moves}
Since concurrently scheduled flow computations can operate on overlapping regions, there are three
conflict types that can occur when applying a sequence of node moves $M$ to the global partition $\Partition$:
balance constraint violations, $\expgain \neq \attrgain$ (i.e., the expected does not
match the actual connectivity reduction), and nodes in $M$ may already be moved by other searches.

In practice, the running time to apply a sequence of node moves is negligible compared to solving flow problems~\cite[see Figure~5.21 on p.~119]{HEUER-DIS}.
Thus, we can afford to use a lock so that only one thread applies moves at a time to address these conflicts.
First, we remove all nodes from $M$ that are not in their expected block.
Afterwards, we compute the block weights as if all remaining moves were applied.
If the resulting partition is balanced, we perform the moves, during which we aggregate the attributed gains $\attrgain$ of
each move. If $\attrgain < 0$, we revert all moves.


\subsection{Flow Network Construction}\label{sec:network_construction}

\begin{figure}[!t]
  \centering
  \includegraphics[width=\textwidth]{img/flow_network}
  \caption{A hypergraph $\flowhypergraph$ (left) induced by a region $B := B_1 \cup B_2$ and the flow network $\flownetwork$ (right)
           given by the Lawler expansion of $\flowhypergraph$.}
  \label{fig:flow_network}
\end{figure}

To improve the cut of a bipartition $\Partition = \{V_1, V_2\}$, we grow a size-constrained
region $B := B_1 \cup B_2$ with $B_1 \subseteq V_1$ and $B_2 \subseteq V_2$
around the cut hyperedges of $\Partition$ via two breadth-first-searches (BFS)~\cite{KAFFPA}.
The first BFS is initialized with all boundary nodes of block $V_1$ and continues to add nodes to $B_1$ as long as
$c(B_1) \le (1 + \alpha \varepsilon) \lceil \frac{c(V)}{2} \rceil - c(V_2)$, where $\alpha$ is an input parameter.
The second BFS that constructs $B_2$ proceeds analogously.
We then contract all nodes in $V_1 \setminus B$ to the source $\source$ and
$V_2 \setminus B$ to the sink $\sink$~\cite{KAFFPA,REBAHFC}
and obtain a coarser hypergraph $\flowhypergraph = (\flownodeset, \flowedgeset)$.
The flow network $\flownetwork$ is then given by the Lawler expansion of $\flowhypergraph$~\cite{Lawler}, which is illustrated
in Figure~\ref{fig:flow_network}.
For each hyperedge $e \in \flowedgeset$, we add two nodes $\innode$ and $\outnode$ and a \emph{bridging edge}
$(\innode,\outnode)$ with capacity $\capacity(\innode,\outnode) = \omega(e)$ to $\flownetwork$. For each pin $u \in e$,
we add two edges $(u,\innode)$ and $(\outnode,u)$ with infinite capacity to $\flownetwork$.
Note that we do not construct $\flownetwork$ explicitly in our actual implementation, since our maximum flow algorithm runs
on $\flowhypergraph$ by implicitly exploiting the structure of the Lawler expansion.


The parameter $\alpha$ controls the size of the flow network.
For $\alpha = 1$, each flow computation yields a balanced bipartition
with a possibly smaller cut in the original hypergraph, since only nodes of $B$ can move
to the opposite block ($c(B_1) + c(V_2) \le (1 + \varepsilon)\lceil \frac{c(V)}{2} \rceil$ and vice versa for block $B_2$).
Larger values for $\alpha$ lead to larger flow problems with potentially smaller minimum cuts,
but also increase the likelihood of violating the balance constraint.
However, this is not a problem since the flow-based refinement routine guarantees balance through incremental minimum cut computations (see Section~\ref{sec:flowbased_refinement}).
In practice, we use $\alpha = 16$ (also used in \Partitioner{KaHyPar}~\cite{KAHYPAR-MF,KAHYPAR-HFC}).
We additionally restrict the distance of each
node $v \in B$ to the cut hyperedges to be smaller than or equal to a parameter $\delta$ ($= 2$). We observed that it
is unlikely that a node \emph{far} way from the cut is moved to the opposite block by the flow-based refinement.


\subsection{The \Partitioner{FlowCutter} Algorithm}\label{sec:flowbased_refinement}

In this section, we discuss the flow-based refinement on a bipartition.
We introduce the aforementioned \Partitioner{FlowCutter} algorithm~\cite{FLOW-CUTTER, yang-wong-fbb},
which parallelization is described in the next section.
To speed up convergence and make parallelism worthwhile, we propose an optimization named \emph{bulk piercing}.

\paragraph{Algorithm Overview}

\Partitioner{FlowCutter} solves a sequence of incremental maximum flow problems until a balanced bipartition is found.
Algorithm~\ref{pseudocode:flowcutter} shows the pseudocode for the approach.
In each iteration, first the previous flow (initially zero) is augmented to a maximum flow regarding the current source set $S$ and sink set $T$.
Subsequently, the node sets $S_r, T_r \subset \flownodeset$ of the source- and sink-side cuts are derived.
This is done via residual (parallel) BFS (forward from $S$ for $S_r$, backward from $T$ for $T_r$).
The node sets induce two bipartitions $(S_r, \flownodeset \setminus S_r)$ and $(\flownodeset \setminus T_r, T_r)$.
If neither is balanced, all nodes on the side with smaller weight are transformed to a source (if $c(S_r) \leq c(T_r)$) or a sink otherwise.
As this would yield the same cut in the next iteration, we add one additional node, called \emph{piercing node}, to the terminal set of the smaller side.
Thus, the bipartitions contributed by the currently smaller side will be more balanced with a possibly larger cut in future iterations.
Since the smaller side is grown, this process will converge to a balanced bipartition.

\begin{algorithm2e}[!t]
	\caption{The \Partitioner{FlowCutter} Algorithm}\label{pseudocode:flowcutter}
  \KwIn{Original hypergraph $H = (V,E,c,\omega)$, flow network $\flowhypergraph = (\flownodeset, \flowedgeset)$
        and a source $\source \in \flownodeset$ and sink $\sink \in \flownodeset$}
  \KwOut{Balanced Bipartition of $\flowhypergraph$}
  $S \gets \{ \source \} , T \gets \{ \sink \}$ \tcp*[r]{initialize source and sink set}
  initialize flow $f: V \times V \rightarrow \mathbb{R}_{\ge 0}$ with $\forall (u,v) \in V \times V: f(u,v) = 0$\;
	\While() {no balanced bipartition found} {
		$f \gets \FuncSty{ParallelMaxPreflow}(\flowhypergraph, S, T, f)$ \tcp*[r]{augment $f$ to a maximum preflow}
		$(S_r,T_r) \gets$ derive source- and sink-side cut $S_r, T_r \subset \flownodeset$\;
    \lIf() {$(S_r, \flownodeset \setminus S_r)$ is balanced} { \Return{$(S_r, \flownodeset \setminus S_r)$} }
    \lElseIf { ($\flownodeset \setminus T_r, T_r)$ is balanced } { \Return{$(\flownodeset \setminus T_r, T_r)$} }
		\lIf() {$c(S_r) \leq c(T_r)$} {
			$S \gets S_r \cup \FuncSty{selectPiercingNode}(S \cup S_r)$
		} \lElse() {
			$T \gets T_r \cup \FuncSty{selectPiercingNode}(T \cup T_r)$
		}
	}
\end{algorithm2e}

For our purpose, there are two important piercing node selection heuristics: \emph{avoid augmenting paths}~\cite{FLOW-CUTTER, yang-wong-fbb} and \emph{distance from cut}~\cite{KAHYPAR-HFC}.
Whenever possible, a node that is not reachable from the source or sink should be picked, i.e., $v \in \flownodeset \setminus (S_r \cup T_r)$.
Such nodes do not increase the weight of the cut, while improving balance~\cite{PicardQ82}.
As a secondary criterion, larger distances from the original cut are preferred, to reconstruct parts of it.


\paragraph{Bulk Piercing Optimization}

On larger instances, piercing only one node per iteration converges slowly. We therefore increase the amount of work in each iteration
by piercing multiple nodes, as long as we are far from balance.

To achieve a small number of iterations (e.g., poly-log) we set a goal on the weights of the sides of the bipartition, and pierce more aggressively the further we are from it.
Assume, we want to pierce the source side next and have already performed $r-1$ piercing iterations on it.
In the $r$-th iteration we want to add $\frac{1}{2^r} \cdot (\frac{c(V)}{2} - c(S))$ new weight to the source side, where $c(S)$ is the weight of the initial source-side terminals (before any piercing) and $\frac{c(V)}{2}$ is the weight of a perfectly balanced bipartition.
Thus the overall weight goal for the $r$-th iteration on the source side is set to $(\frac{c(V)}{2} - c(S)) \sum_{i=1}^{r} \frac{1}{2^i}$.
This is chosen such that we allow a lot of progress early on and become more careful as we get closer to a balanced bipartition.
We track the average weight added per node in previous iterations and from this estimate the number of required piercing nodes to reach the goal for the $r$-th iteration.
To boost measurement accuracy, we pierce only one node for the first few rounds, and then switch to bulk piercing.

\subsection{Parallel Maximum Flow Algorithm}\label{sec:parallel_flow_algo}

Maximum flow algorithms are notoriously difficult to parallelize
efficiently~\cite{ShiloachVishkin, BaumstarkSyncPushRelabel, AndersonSetubalPushRelabel, ColoringPushRelabel}.
The synchronous push-relabel approach of Baumstark \etal~\cite{BaumstarkSyncPushRelabel} is a recent algorithm that
sticks closely to sequential FIFO and thus shows good results.
We first describe the sequential push-relabel algorithm proposed by Goldberg and Tarjan~\cite{PUSH-RELABEL} and
then briefly outline its parallelization.
We conclude with implementation details and intricacies of using \Partitioner{FlowCutter} with preflows.

\paragraph{Push-Relabel Algorithm}
The \emph{push-relabel} algorithm~\cite{PUSH-RELABEL} stores a distance label $\distance(u)$ and an excess value $\excess(u) := \sum_{v \in \flownodeset} f(v,u)$ for each node.
It maintains a \emph{preflow}~\cite{PREFLOW} which is a flow where the conservation constraint is replaced by $\excess(u) \geq 0$.
The distance labels represent a lower bound for the distance of each node to the sink.
A node $u \in \flownodeset$ is \emph{active} if $\excess(u) > 0$. An edge $(u,v) \in \flowedgeset$ is \emph{admissible}
if $r_f(u,v) > 0$ and $\distance(u) = \distance(v) + 1$.
A \emph{$\text{push}(u,v)$} operation sends $\delta = \min(\excess(u),r_f(u,v))$ flow units over $(u,v)$.
It is applicable if $u$ is active and $(u,v)$ is admissible.
A \emph{$\text{relabel}(u)$} operation updates the distance label of $u$ to $\min(\{\distance(v) + 1 \mid r_f(u,v) > 0\})$, which is applicable if $u$ is active and has no admissible edges.
The distance labels are initialized to $\forall u \in \flownodeset \setminus \{ \source \}: \distance(u) = 0$ and $\distance(\source) = |V|$ and all source edges are saturated.
Efficient variants use the \emph{discharge} routine, which repeatedly scans the edges of an active node until its excess is zero.
All admissible edges are pushed and at the end of a scan, the node is relabeled.
The \emph{global relabeling} heuristic~\cite{DBLP:journals/algorithmica/CherkasskyG97} frequently assigns exact
distance labels by performing a reverse BFS from the sink to reduce relabel work in practice.
Note that a maximum preflow already induces a minimum sink-side cut, so if only a minimum cut is required, the algorithm can already stop
once no active nodes with distance label $< n$ exist.

The parallel push-relabel algorithm of Baumstark~\etal~\cite{BaumstarkSyncPushRelabel} proceeds in rounds in which all active nodes are discharged in parallel.
The flow is updated globally, the nodes are relabeled locally and the excess differences are aggregated in a second array using atomic instructions.
After all nodes have been discharged, the distance labels $\distance$ are updated to the local labels $\distance'$ and the excess deltas are applied.
The discharging operations thus use the labels and excesses from the previous round.
This is repeated until there are no nodes with $\excess(v) > 0$ and $\distance(v) < n$ left.
To avoid concurrently pushing flow on residual arcs in both directions (race condition on flow values), a deterministic winning criterion on the old distance labels is used to determine which direction to push, if both nodes are active.
If an arc cannot be pushed due to this, the discharge terminates after the current scan, as the node may not be relabeled in this round.
The rounds are interleaved with global relabeling~\cite{DBLP:journals/algorithmica/CherkasskyG97}, after linear push and relabel work, using parallel reverse BFS in the residual network.
We additionally fixed an undocumented bug in the original algorithm (not source code) for which we refer the reader to Ref.~\cite{MT-KAHYPAR-FLOWS}.

%
%

\paragraph{Intricacies with Preflows and \Partitioner{FlowCutter}}
A maximum preflow only yields a sink-side cut via the reverse residual BFS, but we also need the source-side cut.
We can run flow decomposition~\cite{DBLP:journals/algorithmica/CherkasskyG97} to push excesses back to the source.
However, flow decomposition is difficult to parallelize~\cite{BaumstarkSyncPushRelabel}.
Instead, we initialize the forward residual BFS with all active non-sink excess nodes.
This finds the reverse paths that carry flow from the source to the excess nodes, which is what we need.

Furthermore, when transforming a node with positive excess to a sink, its excess must be added to the flow value.
This only happens when piercing, as sink-side nodes have no excess. 

Finally, we want to reuse the distance labels from the previous round to avoid re-initialization overheads.
However, as the labels are a lower bound on the distance from the sink, piercing on the sink side invalidates the labels.
In this case, we run global relabeling to fix the labels and collect the existing excess nodes, before starting the main discharge loop.
When piercing on the source side the labels remain valid and new excesses are created.
These are added to the active nodes and we do not run an additional global relabeling.
The existing excess nodes are collected during regular global relabel runs.

\paragraph{Implementation Details}

Since $(\innode, \outnode)$ is the only outgoing edge of $\innode$ with non-zero capacity in the Lawler expansion (see Figure~\ref{fig:flow_network}),
the flow on edges $(u, \innode)$ is also bounded by $\omega(e)$ (instead of $\infty$). Adding these capacities is a trivial optimization, but
significantly accelerates the algorithm and increases the available parallelism. This can be explained by the fact that a hypernode $u$ does not
immediately relieve all of its excess to one of its incident nets $e \in \incnets(u)$ during the discharge routine, which is later pushed back
due to the $\omega(e)$ bound. We set the capacities $c(u,\innode)$ to $\infty$ again when deriving the source- and sink-side cut, since only
bridging edges can be cut in the Lawler expansion.

Moreover, we observed that the number of active nodes follows a power-law distribution. Due to little work in later rounds, it takes
many rounds to trigger the global relabeling step that also terminates the algorithm when a maximum preflow is found.
Therefore, we perform additional relabeling if the flow value has not changed for some rounds (500),
and only few active nodes ($< 1500$) were available in each.

\section{A Parallelization of the $n$-Level Partitioning Scheme}\label{sec:nlevel}

Our multilevel algorithm contracts a clustering of highly-connected nodes on each level, which
induces a hierarchy with a logarithmic number of levels. In contrast, \Partitioner{KaHyPar}~\cite{KAHYPAR-K,KAHYPAR-DIS} -- the currently
best sequential partitioning algorithm with regards to solution quality -- contracts only a single node on each level.
Correspondingly, in each refinement step, only a single node is uncontracted followed
by a highly-localized search for improvements around the uncontracted node.
This technique produces almost $n$ levels and is therefore known as the $n$\emph{-level partitioning scheme}.
More levels provide \say{more opportunities to refine the current solution}~\cite{ALPERT-SURVEY}
at different
granularities but also increases the running time of multilevel algorithms. Therefore, \Partitioner{KaHyPar} is
the method of choice for computing high-quality partitions but comes at the cost of substantially higher running times than
other systems -- prohibitively so for very large hypergraphs. Although $n$-level partitioning seems inherently sequential,
we present the first shared-memory parallelization of the technique which achieves good speedups and comparable
solution quality to \Partitioner{KaHyPar} in a fraction of its running time.

We start this section with a formal definition of the (un)contraction operation, and provide
a high-level overview of our $n$-level partitioning algorithm.
We then discuss and present solutions for the main challenges in this algorithm:
finding a parallel schedule of (un)contraction operations and performing them on a dynamic hypergraph data structure in parallel.
We conclude the algorithm description with a discussion on how the refinement algorithms from the previous section are
integrated into the $n$-level algorithm.

\paragraph{The Contraction and Uncontraction Operation}
Contracting a node $v$ onto another node $u$ replaces $v$ with $u$ in all nets $e \in \incnets(v) \setminus \incnets(u)$
and removes $v$ from all nets $e \in \incnets(u) \cap \incnets(v)$. The weight of node $u$ is then $c(u) + c(v)$. We call
$u$ the \emph{representative} of the contraction and $v$ its \emph{contraction partner}. Uncontracting a node $v$ reverses
the corresponding contraction operation.

\begin{algorithm2e}[!t]
  \KwIn{Hypergraph $H = (V,E)$, number of blocks $k$}
  \KwOut{$k$-way partition $\Partition$ of $H$}
  \caption{The $n$-Level Partitioning Algorithm}\label{pseudocode:nlevel}

    \While() {$V$ has too many nodes}{
      \ParallelFor(\label{pseudocode:coarsening_pass}) {$u \in V$ in random order}{
        $v \gets \operatorname{arg\,max}_{v} \sum_{e \in \incnets(u) \cap \incnets(v)} \frac{\omega{(e)}}{|e| - 1}$ \tcp*[r]{find best contraction partner for $u$}
        \lIf (\label{pseudocode:safe_contractions}) {$(v,u)$ is eligible for contraction} {
          contract $v$ onto $u$
        }
      }
    }

    $\Partition \gets$ \FuncSty{InitialPartition}($H, k$)\;
    $\batches = \langle B_1, \dots, B_l \rangle \gets \FuncSty {ConstructBatches}(\contractionforest)$\;
    \For(\tcp*[f]{\small $|B| \approx \maxbatchsize$}) {$B \in \batches$}{
      \ParallelFor() {$v \in B$}{
        uncontract $v$ from $\rep[v]$\EndOfStatement $\Partition[v] \gets \Partition[\rep[v]]$\;
      }
      $\FuncSty{LabelPropagationRefinement}(H,\Partition,B)$\EndOfStatement $\FuncSty{FMRefinement}(H,\Partition,B)$\;
    }
\end{algorithm2e}

\paragraph{Algorithm Overview}
Algorithm~\ref{pseudocode:nlevel} shows the high-level pseudocode of our $n$-level partitioning
algorithm. In the coarsening phase, we iterate in parallel over all nodes and find
the best contraction partner $v$ for each node $u$ using the heavy-edge rating function~\cite{PATOH,HMETIS,KAHYPAR-K} (similar as in our
multilevel algorithm). We then check whether $v$ can be contracted right away
onto $u$ or if there are any other pending contractions that must be performed before.
In the latter case, we transfer the responsibility
of contracting $v$ onto $u$ to the thread resolving the last dependency that defers the contraction.
Once the hypergraph is small enough, we compute an initial partition into $k$ blocks.

Uncontracting only a single node followed by a localized refinement step is
inherently sequential, which is why we have to relax the $n$-level idea in the uncoarsening phase.
We construct a sequence of batches $\batches = \langle B_1, \dots, B_l \rangle$ of contracted nodes,
such that $|B_i| \approx \maxbatchsize$ where $\maxbatchsize$ is an input parameter.
Batches are processed one after another, enabling the uncontraction of nodes in subsequent batches.
Nodes in the same batch are uncontracted in parallel.
The main challenge is to identify which nodes can or even must appear in the same batch.
After uncontracting each batch, we apply highly-localized refinement algorithms around the batched nodes.

\paragraph{A Forest-Based Scheduling of Contraction Operations}
Let us consider a sequence of contractions $\contractions := \splitatcommas{\langle (v_1,u_1),\ldots, (v_n,u_n) \rangle}$
executed exactly in this order ($v_i$ is contracted onto $u_i$).
In this sequence, each node is contracted onto at most one representative, and there are no cyclic contraction dependencies.
Thus, the sequence of contractions form a forest $\contractionforest := (V,\contractions)$ if interpreted
as a graph with directed edges $(v_i,u_i) \in \contractions$.  We will refer to $\contractionforest$ as the
\emph{contraction forest}. Our parallelization uses the observation that there exists several permutations of $\contractions$
leading to the same contraction forest $\contractionforest$. We can contract a node as soon as all of its
children in $\contractionforest$ have been contracted. To obtain parallelism, different subtrees and siblings
can be contracted independently, i.e., we traverse $\contractionforest$ in a bottom-up fashion in parallel.

In our actual algorithm, we do not know $\contractionforest$ in advance.
However, we show how to construct $\contractionforest$ dynamically and from that we derive a
parallel schedule of contraction operations. We call a contraction $(v,u)$ \emph{compatible} with existing
contractions if it satisfies the following three conditions: (i) $v$ must be a root of $\contractionforest$, (ii) adding $(v,u)$ to
$\contractionforest$ must not induce a cycle, and (iii)
the contraction of $u$ onto its parent in $\contractionforest$ must not have started yet.
We represent $\contractionforest$ using an array $\rep$ of size $n$ storing for each node its representative ($u$ is a root
if $\rep[u] = u$). Additionally, we use a zero-initialized array $\pending$, where $\pending[u]$ stores the number of
children of $u$ whose contraction is not finished. If $\pending[u] = 0$ and $\rep[u] \neq u$, we assume that the contraction of
$u$ onto $\rep[u]$ has started and prevent further contractions onto $u$. The entries $\rep[u]$ and $\pending[u]$ are only modified while
holding a node-specific lock for $u$.

If we add a contraction $(v,u)$ to $\contractionforest$, we first lock $v$ and check if $v$ is a root.
If $\rep[v] \neq v$, we discard the contraction as another thread has already selected a representative for $v$.
Otherwise, we walk the path towards the root of $u$'s tree in $\contractionforest$ to find the lowest ancestor $w$
of $u$ whose contraction has not started yet ($rep[w] = w$ or $\pending[w] > 0$, in most cases $w = u$). If $v$ is found
on this path, the contraction is discarded, as it would induce a cycle in $\contractionforest$. If no cycle is found,
we lock $w$, and check $\rep[w]$ and $\pending[w]$ again. If they changed, we find a new suitable ancestor and perform the
cycle check again. Otherwise, we set $\rep[u] = w$ and increment $\pending[w]$ by one, and unlock $v$ and $w$.
We immediately contract $v$ onto $w$ if $\pending[v] = 0$ and subsequently reduce $\pending[w]$ by one. If this
reduces $\pending[w]$ to zero, we recursively apply this process to the contraction $(w, \rep[w])$ if $\rep[w] \neq w$.

\paragraph{Batch Uncontractions}
For the uncoarsening phase, we construct a sequence of batches $\batches = \langle B_1, \ldots, B_l \rangle$ where each
batch $B_i$ contains roughly $\maxbatchsize$ contracted nodes that can be uncontracted independently in parallel.
The batch size $\maxbatchsize$ is an input parameter (set to $\maxbatchsize = 1000$ in our implementation) that interpolates between
scalability (high values) and the inherently sequential $n$-level scheme ($\maxbatchsize = 1$). Uncontracting a batch $B_i$
resolves the last dependencies required to uncontract the next batch $B_{i+1}$. After each batch uncontraction, we apply a highly-localized version of the
label propagation and FM algorithm searching for improvements in a small region around the uncontracted nodes.

\begin{figure}
	\centering
	\includegraphics[width=\textwidth]{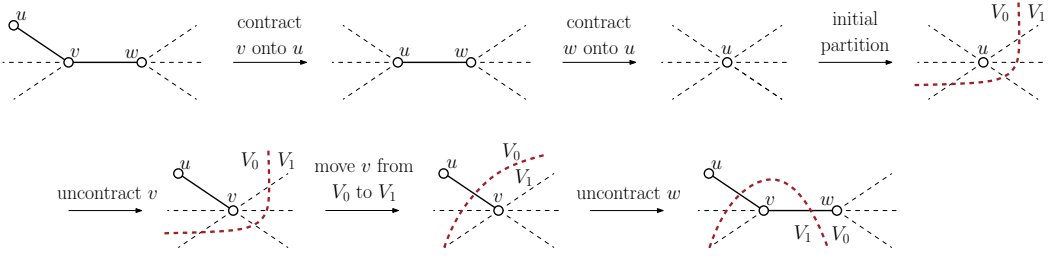}
  \caption{Uncontracting $w$ increases the cut by one since we do not uncontract $v$ and $w$ in reverse order.}\label{fig:nlevel_bug}
\end{figure}

We construct the batches via a top-down traversal of the contraction forest $\contractionforest$.
There are two constraints that we need to consider when constructing the batches:
(i) a node must appear in a batch strictly after the batch containing its representative, and (ii) siblings in
$\contractionforest$ must be uncontracted in reverse order of contraction. The second condition prevents uncontractions
increasing the cut size as illustrated in Figure~\ref{fig:nlevel_bug}, which would violate a fundamental property of
the multilevel scheme. Since contractions can be performed at the same time, it is often not possible to define a strict order in which we have to
revert the contractions. We therefore associate each contraction $(v,u)$ with a time interval $[s_v, e_v]$ by atomically incrementing
a counter before starting ($s_v$) and after finishing ($e_v$) a contraction operation. If the time interval of two nodes overlap,
we assume they were contracted at the same time, otherwise one is strictly earlier than the other. Among siblings, we compute
the transitive closure of nodes with overlapping time intervals and order them decreasingly if one is stricly earlier than the other.
We then use this as the reverse order of contractions, while we add siblings with overlapping time intervals to the same batch.
As this is only a high-level description of the batch construction algorithm, we refer the reader to Ref.~\cite[p.129--131]{HEUER-DIS} for more details.

\paragraph{The Dynamic Hypergraph Data Structure}
Figure~\ref{fig:nlevel_example} illustrates the dynamic hypergraph data structure that stores the pin-list of each net $e$ and the incident
nets $\incnets(u)$ of each node $u$ using two seperate adjacency arrays. When we contract a node $v$ onto another node $u$, we iterate over the incident nets
of $v$ and search for $u$ and $v$ in the pin-list of each net $e \in \incnets(v)$. If we do not find $u$ in $e$, we replace $v$ with
$u$ ($e \in \incnets(v) \setminus \incnets(u)$). Otherwise, we swap $v$ to the end of $e$'s pin-list and decrement
the size of $e$ by one ($e \in \incnets(u) \cap \incnets(v)$), dividing its pin-list into an \emph{active} and \emph{inactive} part.
We use a seperate lock for each net to synchronize edits to the pin-lists. We further mark the nets $e \in \incnets(u) \cap \incnets(v)$
in a bitset $X$, which we then use to update the incident nets of $u$.

The key idea for updating the incident nets is to remove $\incnets(u) \cap \incnets(v)$ from $\incnets(v)$ and concatenate $u$
and $v$ in a doubly-linked list $L_u$. All nodes contracted onto $u$ are then stored in $L_u$. We can then iterate over the incident
nets of $u$ by iterating over all entries $w \in L_u$ and the modified $\incnets(w)$ arrays.
We associate each $\incnets(w)$ array with a counter $t_w$ and each entry $e \in \incnets(w)$ with a marker $t_{w,e}$ (initially set to zero).
Entries with markers $\ge t_w$ are \emph{active}, i.e., were not remove yet.
For removing $\incnets(u) \cap \incnets(v)$ from $\incnets(v)$
(marked in bitset $X$), we iterate over all nodes $w \in L_v$ and increment $t_w$ by one. We then iterate over the previously
active entries of $\incnets(w)$ (now marked with $t_w - 1$) and if an entry is not in $X$, we set its marker to $t_w$.
Otherwise, we swap the entry to the end of the active part but keeping its marker at $t_w - 1$.
A simpler approach would be to represent the incident nets of each node as an adjacency list and add $\incnets(v) \setminus \incnets(u)$ to $\incnets(u)$,
as it is done in \Partitioner{KaHyPar}~\cite{KAHYPAR-DIS,KaHyPar-R}. However, this could lead to quadratic memory usage, and is therefore
not practical for large hypergraphs.

When uncontracting a node $v$ from its representative $u$, we first restore $L_v$ from $L_u$. To do this, we additionally store the last node in
$L_v$ at the time $v$ is contracted onto $u$. We then iterate over all nodes $w \in L_v$ and decrement their counters $t_w$ by one. This
reactivates all nets $e \in \incnets(w)$ that became inactive due to contracting $v$, i.e., were part of $\incnets(u) \cap \incnets(v)$
before the contraction. In these nets, we swap $v$ to the active part of their pin-lists again. All previously active nets $e \in \incnets(w)$ (now marked
with $t_{w,e} > t_w$) were part of $\incnets(v) \setminus \incnets(u)$ before the contraction in which we then replace $u$ with $v$ again. Note that we sort all
pins in the inactive part of a pin-list by the batches in which they are uncontracted. Then, all pins of a net $e$ part of the current batch can be restored
simultaneously by appropriately incrementing the size of $e$. Only one thread that triggers the restore case on a net performs the restore operation,
which we ensure with an atomic \textfunc{test-and-set} instruction.

\begin{figure}
	\centering
	\includegraphics[width=\textwidth]{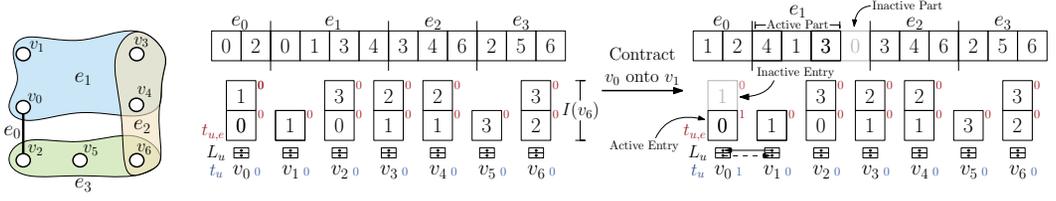}
	\caption{Contraction operation applied on the dynamic hypergraph data structure.}\label{fig:nlevel_example}
\end{figure}

\paragraph{Removing Single-Pin and Identical Nets}
We remove single-pin nets and aggregate the weight of all identical nets at one representative after a pass over all nodes
in the coarsening phase using the same algorithm as already described in Section~\ref{sec:coarsening}. This adds several synchronization points ($\approx \log{n}$)
at which we have to restore them in the uncoarsening phase. \Partitioner{KaHyPar}~\cite{KaHyPar-R,KAHYPAR-DIS} removes these nets directly
after each contraction operation. However, doing this in the parallel setting would introduce additional dependencies for batches, which is why we
decided against it.

\paragraph{Refinement}
After uncontracting a batch, we run a highly-localized version of label propagation and FM refinement initialized
with the boundary nodes of the current batch. The searches then expand to a small region around the uncontracted nodes.
We complement the localized refinement with a refinement pass on the entire hypergraph after restoring single-pin and identical nets.
Here, we run FM (initialized with all boundary nodes) and flow-based refinement. Additionally, we implemented a concurrent
gain table update procedure for batch uncontractions for which we refer the reader to Ref.~\cite[p.132--133]{HEUER-DIS} for more details.

\section{Unifying Hypergraph and Graph Partitioning}\label{sec:parallel_graph}
Hypergraph partitioning (HGP) is considered \say{inherently more complicated}~\cite{KayaaslanPCA12} and therefore more complex \say{in terms
of implementation and running time}~\cite{BulucMSS016} than graph partitioning (GP). However, the high-level description of
partitioning algorithms often does not reveal any difference between the two. For example, label propagation refinement iterates over all nodes, and
moves each node to the block with the highest gain value. While the algorithm is widely used in GP and HGP, the main difference
in its implementation lies in the representation of the graph data structure and the computation of gain values. GP tools build on
data structures using \emph{one} adjacency array to represent the neighbors of nodes, while HGP requires \emph{two} adjacency arrays storing
the pin-lists of hyperedges and the incident nets of nodes. This results in a better cache utilization and faster access times for graph algorithms.
Moreover, the gain value of a node move for the edge cut metric depends on the block assignments of neighbors for GP,
while HGP tools have to maintain or compute the pin count values of hyperedges to decide whether or not it can be removed from the cut.
In the following, we present an optimized graph and partition data structure for GP implementing the interface of our hypergraph data
structure such that we can use them as a drop-in replacement in our partitioning algorithm. We focus on the multilevel algorithm
and refer the reader to Ref.~\cite[p.150--153]{HEUER-DIS} for a description of the $n$-level graph data structure.

\paragraph{Terminology}
An undirected and weighted graph $G = (V,E,c,\omega)$ can be considered as a hypergraph where each net contains only two pins (also called an \emph{edge}).
Therefore, the definitions and notations for hypergraphs also apply to undirected graphs.
We define the weight of an edge $e = \{u,v\} \in E$ as $\omega(u,v) := \omega(e)$.
If $\{u,v\} \notin E$, then $\omega(u,v) = 0$.
An edge $\{u,u\} \in E$ is called a \emph{selfloop}.
For a subset $V' \subseteq V$, $\omega(u, V') := \sum_{v \in V'} \omega(u,v)$
is the weight of all edges connecting node $u$ to $V'$.

\subsection{Graph Data Structure}\label{sec:multilevel_graph}

We use \emph{one} adjacency array to represent an undirected graph.
The adjacency lists stores the directed edges $(u,v) \in \incnets(u)$ for each node $u$.
Note that we could reduce the memory overhead of the data structure by storing only the neighbors
$v \in \neighbors(u)$.
However, our partitioning algorithms are designed to operate on hypergraphs, which often request
the pin-list of a net using its ID (position in the adjacency array). This is only possible when storing
additional information about the source and target node for each edge.
Thus, our data structure requires twice as much memory as traditional graph data structures.

\paragraph{Contraction}
Our multilevel partitioning algorithm contracts a clustering of the nodes on each level.
The coarsening algorithm stores the clustering in an array $\rep$ where $\rep[u] = v$ stores the representative of $u$'s cluster.
For each representative $v$, we maintain the invariant that $\rep[v] = v$.

The contraction algorithm first remaps cluster IDs to a consecutive range by computing a parallel prefix sum on an array of size $n$ that
has a one at position $v$ if $v$ is a representative of a cluster and zero otherwise.
Then, we accumulate the weights and degrees of nodes in each cluster using atomic \textfunc{fetch-and-add} instructions.
Afterwards, we copy the incident edges of each cluster to a consecutive range in a temporary adjacency array by computing
a parallel prefix sum over the cluster degrees.
We then iterate over the adjacency lists of each cluster in parallel, sort them, and remove selfloops and identical edges except for one representative
at which we aggregate their weights. Finally, we construct the adjacency array of the coarse graph by computing a parallel prefix
sum over the remaining cluster degrees.

\subsection{The Partition Data Structure}\label{sec:graph_partition}

For hypergraphs, our partition data structure stores the block assigments $\Partition$, the block weights $c(V_i)$,
the pin count values $\pinsinpart(e,V_i)$, and connectivity sets $\conset(e)$ for each net $e \in E$ and block
$V_i \in \Partition$. Since a graph edge connects only two nodes, we can remove the pin count values and
connectivity sets as we can calculate them on-the-fly. However, we used the synchronized writes to the pin count
values to update the gain table and compute the attributed gain values. We therefore present alternative approaches
for both techniques that exploit the properties of graphs.

\paragraph{The Gain Table}

The connectivity metric reverts to the edge cut metric for plain graphs (since $\con(e) \le |e| = 2$). The gain value of moving
a node $u$ to another block $V_t$ is then defined as $\gain{u}{t} := \omega(u, V_t) - \omega(u, \nodeblock{u})$ (external minus internal edges).
Thus, the gain table for graph partitioning stores and maintains the $\omega(u,V_i)$ values for each node $u \in V$ and block $V_i \in \Partition$ ($n \cdot k$ entries).
If we move a node $u$ from block $V_s$ to $V_t$, we update the gain table by adding $\omega(u,v)$ to $\omega(v,V_t)$ and $-\omega(u,v)$
to $\omega(v,V_s)$ for each neighbor $v \in \neighbors(u)$ using atomic \textfunc{fetch-and-add} instructions.
Hence, the complexity of the gain table updates when each node is moved at most once is $\sum_{u \in V} 2d(u) = \Oh{m}$.

\paragraph{Attributed Gains}
For a node $u$ moved from block $V_s$ to $V_t$, we attribute a connectivity reduction or increase by
$\omega(e)$ to each net $e \in \incnets(u)$ based on the synchronized writes to $\pinsinpart(e,V_s)$
and $\pinsinpart(e,V_t)$. Since we do not longer maintain the pin count values,
we need another synchronization mechanism to decide if a node move removes an edge from the cut or makes
it a cut edge, and based on that attribute a reduction or an increase by the weight of the edge to the move.

We therefore use an array $B$ of size $m$ (initialized with $\perp$) to synchronize the node moves for each edge.
To compute the attributed gain of a node move $u$ from block $V_s$ to $V_t$, we iterate over all incident edges $e = \{u,v\} \in \incnets(u)$ and write
the target block $V_t$ of $u$ to $B[e]$ using an atomic \textfunc{compare-and-swap} operation. If the operation
succeeds, no other thread has moved its neighbor $v$ yet. In this case, $e$ becomes an internal edge if $V_t = \nodeblock{v}$
(reduces the edge cut by $\omega(e)$) and a cut edge otherwise (increases the edge cut by $\omega(e)$).
If we do not succeed in setting $B[e]$ from $\perp$ to $V_t$, another thread has already moved or is currently
moving $v$ to another block. In both cases, its target block is $B[e]$ and we can compute the attributed gain value
for edge $e$ as before by comparing $V_t$ and $B[e]$. After calculating the attributed gain value for
each net $e \in \incnets(u)$, we set the block ID of $u$ to $V_t$. Note that the algorithm only works when
each node is moved at most once, as it is done in our refinement algorithms ($B[e]$ values are reset to $\perp$ after
each refinement round).

\section{Deterministic Partitioning}\label{sec:deterministic}

A program is \emph{externally deterministic}~\cite{DBLP:conf/ppopp/BlellochFGS12} if, given the same input, it produces the same output, each time it is run.
Sequential programs are usually deterministic by default, whereas parallel programs are non-deterministic by default due to randomness in scheduling.
Yet, researchers have advocated the benefits of deterministic parallel programs for several decades~\cite{DBLP:conf/popl/Steele90,DBLP:journals/computer/Lee06,bocchino2009parallel}.
It is easier to debug the program, to reason about performance and it yields reproducible results: in experiments and applications.
Unfortunately, with the exception of \Partitioner{BiPart}~\cite{BIPART}, all published parallel partitioning algorithms so far are non-deterministic.
This stems from concurrently performed moves affecting other ongoing move decisions.

In this section, we present deterministic versions for a subset of the components in our multilevel framework: label propagation refinement, heavy-edge clustering for coarsening, and the Louvain community detection method~\cite{Louvain} (optimizing the popular modularity metric), which we used to guide coarsening decisions.
These clustering algorithms all follow the local moving scheme.
Nodes are visited asynchronously in parallel and are moved to the best cluster in their neighborhood.

To achieve determinism, we use the synchronous local moving approach which is popular in distributed Louvain implementations for community detection~\cite{SLM}.
Moves are calculated but not applied until the end of a local moving round and thus do not influence one another.
The difficult part and difference to prior work is that not all calculated moves can be applied, for example due to the balance constraint. We must select a subset that is as profitable as possible.
We also break down each round into further sub-rounds, to trade off more frequent synchronization for more accurate gains.

Except for the use of non-internally deterministic sub-routines such as sorting, group-by, and emitting elements to a collection in parallel, our algorithms are internally deterministic, i.e., additionally pass through the same internal states on each run~\cite{DBLP:conf/ppopp/BlellochFGS12}.

\paragraph{Deterministic Label Propagation Refinement}

In synchronous label propagation, we first calculate the highest gain move for each node in the current sub-round.
In a second step, we perform balance-preserving swaps between block pairs, prioritized by the gains of the calculated moves.
This generalizes a previous approach in \Partitioner{SocialHash}~\cite{SHP} to weighted hypergraphs, and thus allows the use in a multilevel framework.

For each block pair $(V_s, V_t)$, we sort the two move sequences $M_{st}$ from $V_s$ to $V_t$ and $M_{ts}$ from $V_t$ to $V_s$ by gain and then select a prefix
$M_{st}[0:i]$ and $M_{ts}[0:j]$, from each sequence to apply.
We use the node ID as tie-breaker for determinism.
Let $x(i,j) \coloneq \sum_{a = 0}^{i-1} c(M_{st}[a]) - \sum_{a = 0}^{j-1} c(M_{ts}[a])$ be the weight added to block $V_t$ and removed from block $V_s$ after swapping the nodes in the corresponding prefixes.
We call $i,j$ feasible if $- (\balancedconstraint{\max} - c(V_s)) \leq x(i,j) \leq \balancedconstraint{\max} - c(V_t)$, i.e., after the swaps the partition is still balanced.
To maximize gain, we look for the longest feasible prefixes.
This can be computed similar to merging two sorted arrays.
Keep two pointers $i,j$ to the current prefixes of $M_{st}[0:i], M_{ts}[0:j]$.
In each step advance the pointer of the sequence whose source block receives more weight, i.e., advance $j$ if $x(i,j) < 0$ or $i$ if $x(i,j) > 0$.
If $x(i,j) = 0$ advance either, if the end of the corresponding sequence is not yet reached.

The parallelization follows a common idea for parallel merging.
We first compute the cumulative gains of the sequences via parallel prefix sum operations.
Then the following algorithm is applied recursively to perform the selection.
We do binary search to find  the smallest index $q$ in the shorter sequence whose cumulative weight is not less than the cumulative weight of the middle of the longer sequence.
The two sub-sequence pairs to the left and right of the middle and $q$ can be searched independently in parallel.
Let $l$ denote the length of the longer sequence, then the algorithm does $\mathcal{O}(l)$ work and has $\mathcal{O}(\log^2(l))$ depth.
There are parallel merge algorithms with $\mathcal{O}(\log(l))$ depth, but these are more complicated and unlikely to yield faster running time in practice.

Additionally, we propose two optimizations that are helpful in practice but do not affect the theoretical running time.
If the right parts contain feasible prefixes we return them as we prefer longer prefixes, otherwise we return the result from the left parts.
If the prefixes at the splitting points are feasible, we can omit the left call.
Further, we can omit the right call if the cumulative weight at the middle of the longer sequence exceeds that at the end of the shorter sequence.

\paragraph{Deterministic Louvain Method}

There is no weight constraint on clusters in the Louvain algorithm.
Therefore, we can apply all of the calculated moves.
However, there is an intricacy with floating point weights, which we need due to the edge weight model employed~\cite{KAHYPAR-CA}.
In modularity optimization with the Louvain method, the cluster volume (weighted degree sum of the cluster) is part of the move decisions.
Usually, the volumes are updated in parallel after a node move~\cite{PARALLEL-LOUVAIN}.
Unfortunaly, floating point arithmetic is not associative: different schedules will lead to slightly different rounded values, which actually resulted in non-deterministic outcomes.
One option is to recompute the volumes after each local moving subround.
However, this is substantially slower than only considering updates from moved nodes, particularly at later stages when fewer nodes are moved.
Therefore, we have to establish an order in which the volume updates of each cluster are aggregated.
First we group the updates by cluster, then sort by node ID. Subsequently, we perform a reduction on each group with static load balancing, which is needed for determinism.

\paragraph{Deterministic Clustering for Coarsening}

During coarsening, we bound the weight of the heaviest cluster by an upper weight limit $c_{\max}$,
to ensure that initial partitioning can find a feasible solution.
The difference to refinement is that we have significantly more clusters, and only unclustered nodes (singletons) can move.
Therefore, the approach from refinement is not applicable here, which is why we use a simpler scheme.

Each unclustered node in a sub-round first determines its desired target cluster according to the heavy-edge rating function.
Then, we group the moves by the target cluster, and sort them in order of ascending node weight and use the node ID as tie-breaker for determinism.
For each group, we compute a prefix sum on the node weights, and apply the longest prefix that does not exceed the weight constraint, rejecting the remaining moves.

As an optimization to reduce the amount of work in the group-by stage (second largest bottleneck), we already sum up cluster weights during the target-cluster calculation step (main bottleneck).
If all moves into a cluster combined do not exceed the weight constraint, we simply approve them all and exclude the target cluster from the group-by stage.

For further details such as implementation details and group-by mechanisms used, as well as initial partitioning and contraction algorithms, we refer to~\cite{GOTT-DIS, MT-KAHYPAR-SDET}.

\section{Experiments}\label{sec:experiments}

All presented algorithms have been made available in the \textbf{M}ulti-\textbf{T}hreaded \textbf{Ka}rlsruhe \textbf{Hy}pergraph \textbf{Par}titioning framework
\Partitioner{Mt-KaHyPar}\footnote{\Partitioner{Mt-KaHyPar} is publicly available from \url{https://github.com/kahypar/mt-kahypar}}.
It implements a parallel multilevel and $n$-level partitioning algorithm, as well as a deterministic version of the multilevel
algorithm and optimized data structures for graph partitioning. \Partitioner{Mt-KaHyPar} optimizes the connectivity metric for
hypergraph partitioning and the edge cut metric for graph partitioning.

The following experimental evaluation is structured as follows: We first decribe the four different benchmark sets composed of
over 800 graphs and hypergraphs, and discuss the experimental setup and methodology used in our experiments. We
then evaluate the time-quality trade-offs and speedups of \Partitioner{Mt-KaHyPar}'s different partitioning configurations, and
analyze the running time of its algorithmic components in Section~\ref{sec:mt_kahypar_evaluation}.
We conclude the evaluation by comparing \Partitioner{Mt-KaHyPar} to $25$ different sequential and parallel graph and hypergraph
partitioning algorithms in Section~\ref{sec:comparison}.

\paragraph{Instances} We assembled four different benchmark sets. Two of the sets consist of graphs (G), while the other two consists of hypergraphs (HG).
The sets are further subdivided into medium-sized (M) and large instances (L). We abbreviate the name of a benchmark set, e.g., with \shortlargehg.
The baseline denotes the size of the instances, while the subscript indicates whether it contains graphs or hypergraphs. We summarize the properties
of the instances contained in the benchmark sets in Figure~\ref{fig:instances}\footnote{
We made all benchmark sets and detailed statistics of their properties publicly available
from \url{https://algo2.iti.kit.edu/heuer/talg/}.}.
Note that all graphs and hypergraphs have unit node and (hyper)edge weights.

The hypergraph instances are derived from four sources encompassing three application domains:
the ISPD98 VLSI Circuit Benchmark Suite~\cite{ISPD98} (\ISPD),
the DAC 2012 Routability-Driven Placement Contest~\cite{DAC} (\DAC),
the SuiteSparse Matrix Collection~\cite{SPM} (\SPM), and the International SAT Competition 2014 \cite{SAT14} (\SAT).
We interprete the rows and columns of a sparse matrix as nets and nodes, and a non-zero entry in a cell $(i,j)$ indicates whether
or not node $j$ is a pin of net $i$~\cite{PATOH}. We translate satisfiability formulas into three different hypergraph representations~\cite{MANN-PAPA14}.
The \Primal~resp.~\Literal~representation interpretes the variables resp.~literals as nodes, while the clauses form the hyperedges spanning the corresponding nodes.
The \Dual~representation models the clauses as nodes and variables as hyperedges.

Set~\shortmediumhg~contains all 488 hypergraphs from the well-established benchmark set of Heuer and Schlag~\cite{KAHYPAR-CA}
(18~\ISPD, 10~\DAC, 184~\SPM, 276~\Primal, \Literal, and \Dual~instances). The benchmark \largehg~is composed of 94 large hypergraphs that
we selected in order to have more inputs where parallelization is important and useful. It contains the 8 largest \SAT~instances from \mediumhg~that we enhanced with 6 even larger
satisfiability formulas from the International SAT Competition 2014~\cite{SAT14} ($3 \cdot 14 = 42$ different hypergraph representations). We also included 42 sparse matrices with at least 15 million non-zeros, randomly sampled from
the SuiteSparse Matrix Collection~\cite{SPM}. Additionally,
we added all \DAC~instances from \mediumhg. The largest hypergraph of \largehg~has roughly two billion pins.

Our graph benchmark sets are composed of instances from the 10th DIMACS Implementation Challenge~\cite{DIMACS} (\Dimacs),
the Stanford Large Network Dataset Collection~\cite{SNAP} and the Laboratory for Web Algorithms~\cite{webgraphWS} (\Social),
the DAC 2012 Routability-Driven Placement Contest~\cite{DAC} (\DAC),
the SuiteSparse Matrix Collection~\cite{SPM,wovsyd2007} (\SPM), and several randomly generated graphs~\cite{KAGEN,WSVR} (\Random).

Set \shortmediumgr~was initially assembled by Gottesbüren~\etal~\cite{KAMINPAR} (195 graphs) from which we excluded the 39 largest graphs
and additionally added 16 social networks from the Stanford Large Network Dataset Collection~\cite{SNAP} (114 \Dimacs, 30 \Social, 15 \Random, 3 \SPM, and
10 \DAC~instances). The benchmark \largegr~contains 38 out of 42 instances from a graph collection assembled by Ahkremtsev~\cite{MT-KAHIP-DIS} (four instances
were considered as too large as they were used to evaluate external memory algorithms). Additionally, we enhanced \largegr~with 15 graphs that we excluded from \mediumgr~and were
not contained in \largegr~yet (16 \Dimacs, 16 \Social, 15 \Random, and 6 \SPM~instances). The largest graph of \largegr~has roughly two billion edges.

\begin{figure}
	\centering
  \ifpdfplots
    \includegraphics{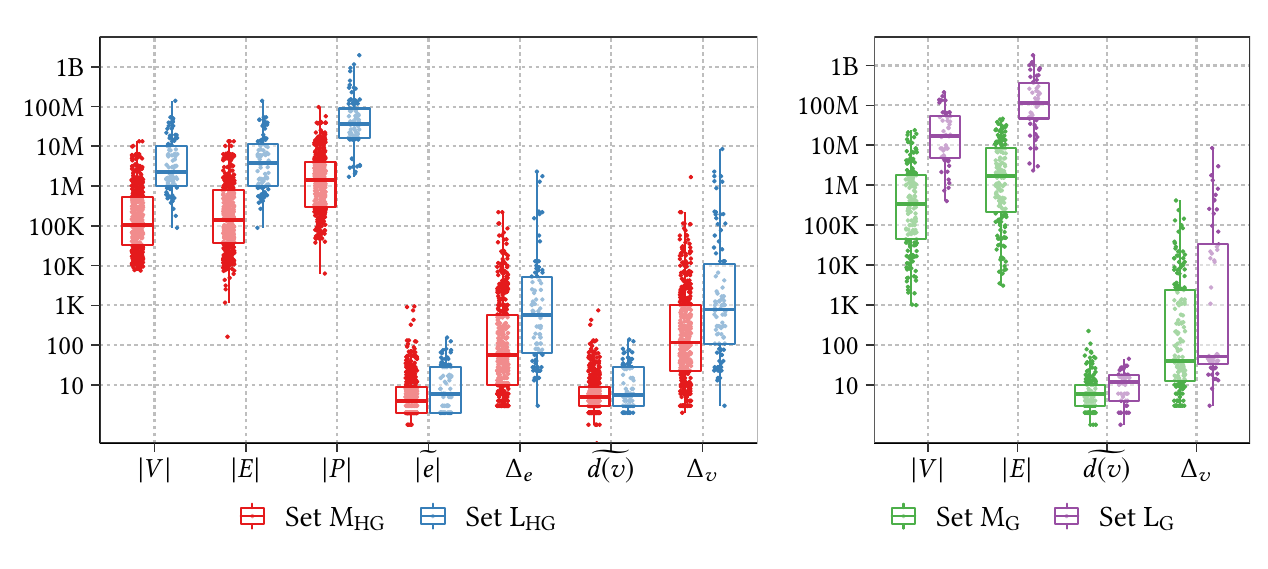}
  \fi
  \vspace{-0.25cm}
  \caption{Summary of different properties for our benchmark sets. It shows for each
    (hyper)graph (points), the number of nodes $|V|$, nets $|E|$ and pins $|P|$,
    as well as the median and maximum net size ($\medsize$ and $\maxsize{e}$) and node degree ($\meddeg$ and $\maxsize{v}$).}\label{fig:instances}
\end{figure}

\paragraph{Experimental Setup}
Experiments on medium-sized instances (\mediumgr~and~\shortmediumhg) run on a cluster of Intel Xeon Gold 6230 processors (2 sockets with 20 cores each)
running at $2.1$ GHz with 96GB RAM.
In these experiments, we partition each (hyper)graph ten times using different random seeds into
$k \in \{2,4,8,16,32,64,128\}$ blocks with an allowed imbalance of $\varepsilon = 3\%$ and a time limit of eight hours.
Experiments on large instances (\largegr~and~\shortlargehg) are done on an AMD EPYC 7702 processor (1 socket with 64 cores) running
at $2.0$--$3.5$ GHz with 1024GB RAM. Here, we partition each (hyper)graph three times using different random seeds into $k \in \{2,8,16,64\}$
blocks with an allowed imbalance of $\varepsilon = 3\%$ and a time limit of two hours. Note that we restrict the parameter space for experiments
on large instances due to limited computational resources. For graph partitioning, we configure the algorithms to optimize the
edge cut metric, while we focus on the connectivity metric for hypergraphs. We will also refer to both metrics as the solution
quality of a partition.

\paragraph{Aggregating Performance Numbers}
We call a (hyper)graph partitioned into $k$ blocks an \emph{instance}. For each instance, we aggregate running times
and the solution quality using the arithmetic mean over all seeds. To further aggregate over multiple instances, we use the geometric mean
for absolute running times and self-relative speedups. If all runs of an algorithm produced an imbalanced partition or ran into the time limit on
an instance, we consider the solution as \emph{infeasible}. In plots, we mark imbalanced solutions with \ding{55} and similarly instances
that timed out with \ClockLogo. Runs with imbalanced partitions are not excluded from aggregated running times. For runs that exceeded the time limit,
we use the time limit itself in the aggregates. When comparing running times,
we say that an algorithm $\mathcal{A}$ is faster than $\mathcal{B}$ by a factor of $x$ on average if
$x := \nicefrac{\gmean{t_\mathcal{B}}}{\gmean{t_\mathcal{A}}} > 1$ where $\gmean{t_\mathcal{A}}$ and
$\gmean{t_\mathcal{B}}$ are geometric mean running times of $\mathcal{A}$ and $\mathcal{B}$.

\paragraph{Performance Profiles}
\emph{Performance profiles} can be used to compare the solution quality of different algorithms~\cite{PERFORMANCE-PROFILES}.
Let $\mathcal{X}$ be the set of all algorithms, $\mathcal{I}$ the set of instances, and $q_{\mathcal{A}}(I)$ the quality of algorithm
$\mathcal{A} \in \mathcal{X}$ on instance $I \in \mathcal{I}$ ($q_{\mathcal{A}}(I)$ is the arithmetic mean over all seeds).
For each algorithm $\mathcal{A}$, performance profiles show the fraction of instances ($y$-axis) for which $q_\mathcal{A}(I) \leq \tau \cdot \text{Best}(I)$, where $\tau$ is on the $x$-axis
and $\text{Best}(I) := \min_{\mathcal{A}' \in \mathcal{X}}q_{\mathcal{A}'}(I)$ is the best solution produced by an algorithm $\mathcal{A}' \in \mathcal{X}$ for an instance $I \in \mathcal{I}$.
For $\tau = 1$, the $y$-value indicates the percentage of instances for which an algorithm $\mathcal{A} \in \mathcal{X}$ performs best.
Achieving higher fractions at smaller $\tau$ values is considered better.
The \ding{55}- and \ClockLogo-tick indicates the fraction of instances for which \emph{all} runs of that algorithm produced an imbalanced solution or timed out.
Note that these plots relate the quality of an algorithm to the best solution and thus do not permit a full ranking of
three or more algorithms.

\paragraph{Effectiveness Tests}
Ahkremtsev \etal~\cite{MT-KAHIP} introduce \emph{effectiveness tests} to compare solution quality when two algorithms are given a
similar running time by performing additional repetitions with the faster algorithm.
Following this approach, we generate virtual instances that we compare using performance profiles.
Consider two algorithms $\mathcal{A}$ and $\mathcal{B}$, and an instance $I$.
We first sample one run of both algorithms for instance $I$.
Let $t_\mathcal{A}^1, t_\mathcal{B}^1$ be their running times and assume that $t_\mathcal{A}^1 \le t_\mathcal{B}^1$.
We then sample additional runs without replacement for $\mathcal{A}$ until their accumulated time exceeds $t_\mathcal{B}^1$ or all runs have been sampled.
Let $t_\mathcal{A}^2, \dots, t_\mathcal{A}^l$ denote their running times.
We accept the last run with probability $(t_\mathcal{B}^1 - \sum_{i = 1}^{l-1} t_\mathcal{A}^i) / t_\mathcal{A}^l$ so that the expected time for the sampled runs of $\mathcal{A}$ equals $t_\mathcal{B}^1$.
The solution quality is the minimum out of the sampled runs.
For each instance, we generate $10$ virtual instances.

\paragraph{Statiscal Significance Tests}
We use the Wilcoxon signed-rank test~\cite{WILCOXON} to determine whether the difference of the solutions produced by two algorithms
is statiscally significant. At a $1\%$ significance level ($p \le 0.01$), a Z-score with $|Z| > 2.576$ is deemed significant~\cite[p.~180]{WilcoxonZValues}.

\subsection{Evaluation of Framework Configurations}\label{sec:mt_kahypar_evaluation}

In this section, we present a detailed evaluation of our shared-memory partitioning algorithm \Partitioner{Mt-KaHyPar}. We first describe its different
configurations and compare them regarding solution quality and running time. We then discuss the running times of the different
algorithmic components, present speedups, and evaluate the impact of our optimizations for plain graphs.

\paragraph{Framework Configurations}

The \Partitioner{Mt-KaHyPar} framework provides a multilevel (\Partitioner{Mt-KaHyPar-D}, \textbf{D}efault) and $n$-level partitioning algorithm
(\Partitioner{Mt-KaHyPar-Q}, \textbf{Q}uality), as well as configurations extending them with flow-based refinement
(\Partitioner{Mt-KaHyPar-D-F} and \Partitioner{Mt-KaHyPar-Q-F}, \textbf{F}lows). It also implements a determistic version of the multilevel
algorithm (\Partitioner{Mt-KaHyPar-SDet}, \textbf{S}peed-\textbf{Det}erministic), which does not use the FM algorithm.
The code is written in \CPP$17$, parallelized using the \ShortTBB~parallelization library~\cite{TBB}, and compiled using \gpp{9.2} with the
flags \texttt{-O3 -mtune=native -march=native}.
All of these algorithms have a large number of configuration options and were carefully tuned to provide the best trade-off between solution
quality and running time. However, a detailed parameter tuning study is beyond the scope of this paper.
We already mentioned specific choices for relevant parameters in the text and refer the reader to our conference publications~\cite{MT-KAHYPAR-SDET,MT-KAHYPAR-D,MT-KAHYPAR-Q,MT-KAHYPAR-FLOWS} and
the dissertation of Heuer~\cite[see Table~5.1 on p.~98--99]{HEUER-DIS} for a detailed overview\footnote{Parameter tuning was done on a subset \shortmediumparameterhg~of \mediumhg~that consists
of $100$ instances not contained in \largehg. We compared the quality produced by different partitioning algorithms on \mediumhg~and \shortmediumhg$\setminus$\shortmediumparameterhg~using performance profiles and found
that they do not differ~\cite[see Figure~8.1 on p.~160]{HEUER-DIS}. Thus, we decided to include the parameter tuning instances in the final evaluation to increase the evidence of the following
experimental results.}.

\paragraph{Time-Quality Trade-Off}

\begin{figure}
  \centering
  \begin{minipage}{\textwidth}
    \centering
  \ifpdfplots
    \includegraphics{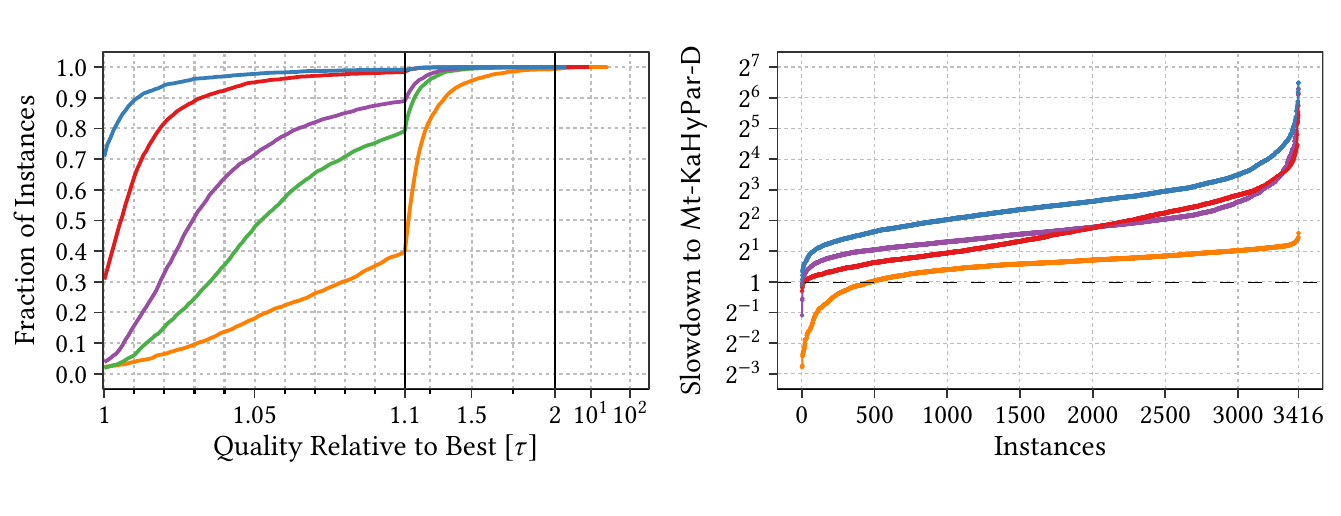}
  \fi
  \end{minipage} %
  \begin{minipage}{\textwidth}
    \vspace{-0.25cm}
    \centering
  \ifpdfplots
    \includegraphics{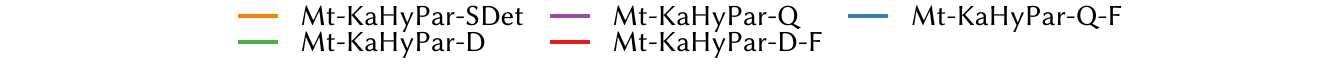}
  \fi
  \end{minipage} %
  \vspace{-0.5cm}
  \caption{Performance profiles and running times comparing the different configurations
    of \Partitioner{Mt-KaHyPar} executed with $10$ threads on \mediumhg.}\label{fig:mt_kahypar_quality}
\end{figure}

Figure~\ref{fig:mt_kahypar_quality} compares the solution quality of the partitions produced by the different configurations
of \Partitioner{Mt-KaHyPar} and their running times relative to \Partitioner{Mt-KaHyPar-D} on \mediumhg. The
configurations can be ranked from lowest to highest quality as follows: \Partitioner{Mt-KaHyPar-SDet} (\gmeantime~$1.25$s),
\Partitioner{Mt-KaHyPar-D} ($0.88$s), \Partitioner{Mt-KaHyPar-Q} ($2.99$s), \Partitioner{Mt-KaHyPar-D-F} ($2.73$s), and \Partitioner{Mt-KaHyPar-Q-F} ($5.08$s).
The ranking looks similar for running times except for \Partitioner{Mt-KaHyPar-D} which is faster than \Partitioner{Mt-KaHyPar-SDet}.
However, this changes when we compare their running times on the larger instances of \largehg.
Here, our deterministic configuration is faster than \Partitioner{Mt-KaHyPar-D} (\Partitioner{Mt-KaHyPar-SDet}: $3.14$s vs
\Partitioner{Mt-KaHyPar-D}: $4.65$s with 64 threads). For smaller instances, initial partitioning is the most time-consuming component
since we stop coarsening when we reach $160k$ nodes which can be close to the original number of nodes for some instances (e.g., $10240$ nodes for $k = 64$).
To reduce the running time of initial partitioning, \Partitioner{Mt-KaHyPar-D} adaptively adjusts the number
of repetitions of the different algorithms in the bipartitioning portfolio based on their success so far.
For larger instances, the smallest hypergraph is often significantly smaller than the input, and therefore
the running time of initial partitioning becomes negligible compared to the other phases.

The median improvement in solution quality of \Partitioner{Mt-KaHyPar-D} over \Partitioner{Mt-KaHyPar-SDet} is $6\%$, while flow-based refinement (\Partitioner{Mt-KaHyPar-D-F})
improves \Partitioner{Mt-KaHyPar-D} by $4.2\%$ in the median at the cost of a $3$ times slower running time on average.
When we compare the multilevel (\Partitioner{Mt-KaHyPar-D}) and $n$-level partitioning algorithm (\Partitioner{Mt-KaHyPar-Q}), we see that $n$-level partitioning produces
partitions that are $1.9\%$ better than those produced by our multilevel algorithm in the median, but its running time is $3.4$ times slower on average.
The differences in solution quality and running time are less pronounced when both configurations use flow-based refinement (median improvement of \Partitioner{Mt-KaHyPar-Q-F}
over \Partitioner{Mt-KaHyPar-D-F} is $0.6\%$). Note that multilevel partitioning with flow-based refinement produces better partitions than our $n$-level
configuration ($2\%$), while it is also slightly faster.

We have seen that using stronger refinement algorithms leads to substantially better solution quality at the cost of higher running times.
Moreover, traditional multilevel algorithms can produce better partitions than $n$-level algorithms when flow-based refinement is used.

\paragraph{Effectiveness Tests}

Our $n$-level algorithm computes better partitions than our multilevel algorithm without flow-based refinement, but is $3$ times slower on average.
When both configurations use flow-based refinement, the difference in solution quality becomes less pronounced.
We therefore use effectiveness tests to compare \Partitioner{Mt-KaHyPar-D(-F)} and \Partitioner{Mt-KaHyPar-Q(-F)} when both are given the same amount of time by
performing additional repetitions with the faster algorithm until the accumulated running time equals the running time of the slower algorithm.

Figure~\ref{fig:effectiveness_tests} shows the results of these experiments. As we can see, the performance lines of
\Partitioner{Mt-KaHyPar-D(-F)} and \Partitioner{Mt-KaHyPar-Q(-F)} are almost identical in the performance profiles. This means
that \Partitioner{Mt-KaHyPar-D(-F)} computes partitions of comparable quality to its $n$-level counterpart when we give more time for
additional repetitions.

In contrast to the prevalent perception in the literature that more levels lead to better partitioning results~\cite{AlpertHK97,KAHYPAR-DIS,Saab95},
we showed that already a logarithmic numbers of levels suffices to compute solutions of high quality.
However, we still see a large potential in the $n$-level scheme as it provides a greater design space for future improvements.

\begin{figure}
  \centering
  \begin{minipage}{\textwidth}
    \centering
  \ifpdfplots
    \includegraphics{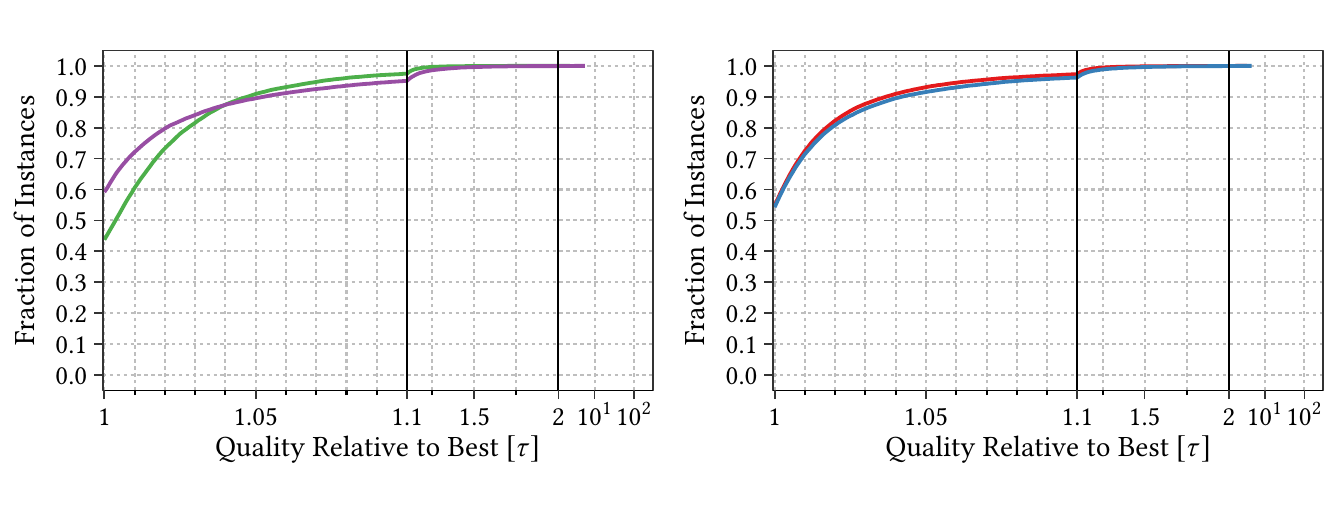}
  \fi
  \end{minipage} %
  \begin{minipage}{\textwidth}
    \vspace{-0.25cm}
    \centering
  \ifpdfplots
    \includegraphics{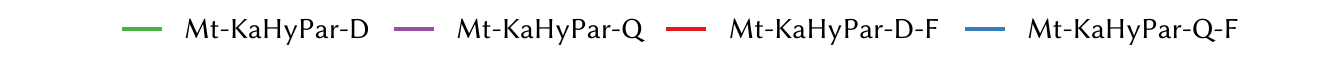}
  \fi
  \end{minipage} %
  \vspace{-0.5cm}
  \caption{Effectiveness tests comparing \Partitioner{Mt-KaHyPar-D} and
    \Partitioner{Mt-KaHyPar-Q} (left), and both configurations that extend them with flow-based refinement (right) on \mediumhg.}
	\label{fig:effectiveness_tests}
\end{figure}

\paragraph{Running Time of Algorithmic Components}
We now analyze the running times of the different algorithmic components of \Partitioner{Mt-KaHyPar} on \largehg\footnote{Since initial
partitioning uses most of the other components within multilevel recursive bipartitioning,
we evaluate running times on the larger instances of \largehg~such that initial partitioning becomes
less time-consuming as explained earlier. We evaluated the solution quality of \Partitioner{Mt-KaHyPar}'s different configurations on \mediumhg~due to the
effectiveness tests, which require a large number of repetitions per instance (10 repetitions on \mediumhg~vs~3 repetitions on \largehg.)}.
Figure~\ref{fig:component_running_times} shows the fraction of instances (x-axis) for which the share of a component on the total execution
time is $\ge y\%$ for each configuration of \Partitioner{Mt-KaHyPar}. The intersection of $x = 0.5$ with the line of a component is the median share of the component on the
overall partitioning time.

\begin{figure}
  \centering
  \begin{minipage}{\textwidth}
    \centering
  \ifpdfplots
    \includegraphics{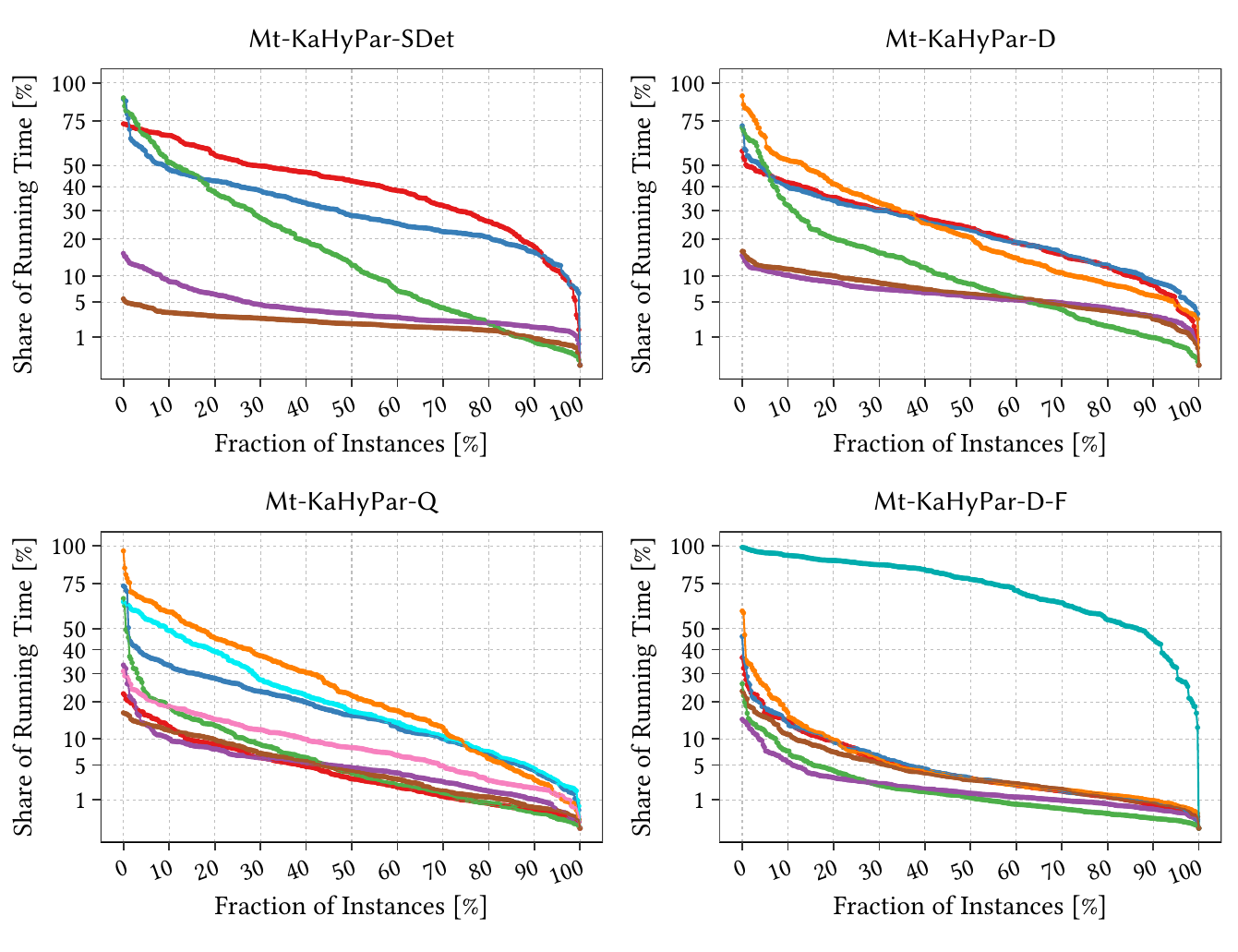}
  \fi
  \end{minipage} %
  \begin{minipage}{\textwidth}
    \vspace{-0.1cm}
    \centering
  \ifpdfplots
    \includegraphics{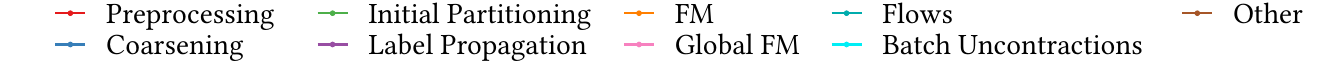}
  \fi
  \end{minipage} %
  \vspace{-0.25cm}
  \caption{Running time shares of the algorithmic components on the total execution time of
  the different configurations of \Partitioner{Mt-KaHyPar}. For \Partitioner{Mt-KaHyPar-Q}, label propagation and
  FM corresponds to their localized versions that run after each batch uncontraction, while \emph{Global FM} refers to the
  FM version that runs on the entire hypergraph after restoring identical nets.}
	\label{fig:component_running_times}
\end{figure}

The most time-consuming components of \Partitioner{Mt-KaHyPar-D} are preprocessing (consisting of the community detection algorithm
presented in Section~\ref{sec:community_detection}), coarsening,
and the FM algorithm. These components have similar shares on the total partitioning time, which is between
$21\%$ and $23\%$ in the median. However, there are some long-running outliers for the FM algorithm
on instances with many large hyperedges. Here, the FM searches tend to move more nodes due to many zero-gain moves.
The median share of initial partitioning on the total execution time is $8.3\%$. Longer running times can be observed for
instances where we do not reach the contraction limit as, e.g., social networks with highly-skewed node degree distributions.
The running time of label propagation is negligible on most of the instances.

In the deterministic version of \Partitioner{Mt-KaHyPar}, preprocessing (median share on the total execution time is $42.7\%$) and
coarsening ($28.1\%$) takes the most time, while flow-based refinement ($77.8\%$) dominates the running time of \Partitioner{Mt-KaHyPar-D-F}
(the same is holds for \Partitioner{Mt-KaHyPar-Q-F}, which is why it is omitted in the plot).
In our $n$-level partitioning algorithm, the most time-consuming components are coarsening ($16\%$), batch uncontractions ($17.3\%$), and the localized version
of the FM algorithm ($22.1\%$).

\begin{filecontents*}{data/scalability.dat}
deterministic_total_speed_up_t_4_0 = 3.9
deterministic_total_speed_up_larger_than_3_t_4 = 99.7
deterministic_preprocessing_speed_up_t_4_0 = 4.1
deterministic_preprocessing_speed_up_larger_than_3_t_4 = 100
deterministic_coarsening_speed_up_t_4_0 = 3.9
deterministic_coarsening_speed_up_larger_than_3_t_4 = 100
deterministic_initial_partitioning_speed_up_t_4_0 = 3.8
deterministic_initial_partitioning_speed_up_larger_than_3_t_4 = 98.4
deterministic_uncoarsening_speed_up_t_4_0 = 4
deterministic_uncoarsening_speed_up_larger_than_3_t_4 = 100
deterministic_total_speed_up_t_4_100 = 3.8
deterministic_total_fraction_t_4_100 = 38.8
deterministic_preprocessing_speed_up_t_4_100 = 3.9
deterministic_preprocessing_fraction_t_4_100 = 25.5
deterministic_coarsening_speed_up_t_4_100 = 3.7
deterministic_coarsening_fraction_t_4_100 = 16.5
deterministic_initial_partitioning_speed_up_t_4_100 = 3.4
deterministic_initial_partitioning_fraction_t_4_100 = 16.8
deterministic_uncoarsening_speed_up_t_4_100 = 3.9
deterministic_uncoarsening_fraction_t_4_100 = 2.1
deterministic_total_speed_up_t_16_0 = 12.8
deterministic_total_speed_up_larger_than_10_t_16 = 100
deterministic_preprocessing_speed_up_t_16_0 = 12.6
deterministic_preprocessing_speed_up_larger_than_10_t_16 = 98.7
deterministic_coarsening_speed_up_t_16_0 = 12.4
deterministic_coarsening_speed_up_larger_than_10_t_16 = 98.4
deterministic_initial_partitioning_speed_up_t_16_0 = 13.5
deterministic_initial_partitioning_speed_up_larger_than_10_t_16 = 99.2
deterministic_uncoarsening_speed_up_t_16_0 = 13
deterministic_uncoarsening_speed_up_larger_than_10_t_16 = 98.7
deterministic_total_speed_up_t_16_100 = 12.5
deterministic_total_fraction_t_16_100 = 38.8
deterministic_preprocessing_speed_up_t_16_100 = 12.5
deterministic_preprocessing_fraction_t_16_100 = 25.5
deterministic_coarsening_speed_up_t_16_100 = 12.4
deterministic_coarsening_fraction_t_16_100 = 16.5
deterministic_initial_partitioning_speed_up_t_16_100 = 12.4
deterministic_initial_partitioning_fraction_t_16_100 = 16.8
deterministic_uncoarsening_speed_up_t_16_100 = 13.8
deterministic_uncoarsening_fraction_t_16_100 = 2.1
deterministic_total_speed_up_t_64_0 = 28.8
deterministic_total_speed_up_larger_than_20_t_64 = 98.4
deterministic_preprocessing_speed_up_t_64_0 = 25.7
deterministic_preprocessing_speed_up_larger_than_20_t_64 = 93.6
deterministic_coarsening_speed_up_t_64_0 = 27.4
deterministic_coarsening_speed_up_larger_than_20_t_64 = 86.7
deterministic_initial_partitioning_speed_up_t_64_0 = 34.4
deterministic_initial_partitioning_speed_up_larger_than_20_t_64 = 96
deterministic_uncoarsening_speed_up_t_64_0 = 28.9
deterministic_uncoarsening_speed_up_larger_than_20_t_64 = 92.6
deterministic_total_speed_up_t_64_100 = 29.6
deterministic_total_fraction_t_64_100 = 38.8
deterministic_preprocessing_speed_up_t_64_100 = 27
deterministic_preprocessing_fraction_t_64_100 = 25.5
deterministic_coarsening_speed_up_t_64_100 = 29.8
deterministic_coarsening_fraction_t_64_100 = 16.5
deterministic_initial_partitioning_speed_up_t_64_100 = 33.8
deterministic_initial_partitioning_fraction_t_64_100 = 16.8
deterministic_uncoarsening_speed_up_t_64_100 = 32.2
deterministic_uncoarsening_fraction_t_64_100 = 2.1
default_total_speed_up_t_4_0 = 3.5
default_total_speed_up_larger_than_3_t_4 = 89.1
default_preprocessing_speed_up_t_4_0 = 3.1
default_preprocessing_speed_up_larger_than_3_t_4 = 68.4
default_coarsening_speed_up_t_4_0 = 3.6
default_coarsening_speed_up_larger_than_3_t_4 = 91.5
default_initial_partitioning_speed_up_t_4_0 = 3.7
default_initial_partitioning_speed_up_larger_than_3_t_4 = 85.1
default_uncoarsening_speed_up_t_4_0 = 3.4
default_uncoarsening_speed_up_larger_than_3_t_4 = 83.2
default_fm_speed_up_t_4_0 = 3.6
default_fm_speed_up_larger_than_3_t_4 = 89.9
default_total_speed_up_t_4_100 = 3.3
default_total_fraction_t_4_100 = 40.7
default_preprocessing_speed_up_t_4_100 = 2.8
default_preprocessing_fraction_t_4_100 = 9.6
default_coarsening_speed_up_t_4_100 = 3.3
default_coarsening_fraction_t_4_100 = 14.1
default_initial_partitioning_speed_up_t_4_100 = 3.2
default_initial_partitioning_fraction_t_4_100 = 10.4
default_uncoarsening_speed_up_t_4_100 = 3.5
default_uncoarsening_fraction_t_4_100 = 25
default_fm_speed_up_t_4_100 = 3.5
default_fm_fraction_t_4_100 = 23.7
default_total_speed_up_t_16_0 = 10.6
default_total_speed_up_larger_than_10_t_16 = 65.7
default_preprocessing_speed_up_t_16_0 = 9.1
default_preprocessing_speed_up_larger_than_10_t_16 = 30.6
default_coarsening_speed_up_t_16_0 = 10.9
default_coarsening_speed_up_larger_than_10_t_16 = 60.9
default_initial_partitioning_speed_up_t_16_0 = 11.5
default_initial_partitioning_speed_up_larger_than_10_t_16 = 71
default_uncoarsening_speed_up_t_16_0 = 9.8
default_uncoarsening_speed_up_larger_than_10_t_16 = 56.6
default_fm_speed_up_t_16_0 = 11.1
default_fm_speed_up_larger_than_10_t_16 = 72.9
default_total_speed_up_t_16_100 = 10.4
default_total_fraction_t_16_100 = 40.7
default_preprocessing_speed_up_t_16_100 = 9.4
default_preprocessing_fraction_t_16_100 = 9.6
default_coarsening_speed_up_t_16_100 = 11.1
default_coarsening_fraction_t_16_100 = 14.1
default_initial_partitioning_speed_up_t_16_100 = 10.8
default_initial_partitioning_fraction_t_16_100 = 10.4
default_uncoarsening_speed_up_t_16_100 = 11.6
default_uncoarsening_fraction_t_16_100 = 25
default_fm_speed_up_t_16_100 = 11.9
default_fm_fraction_t_16_100 = 23.7
default_total_speed_up_t_64_0 = 20.5
default_total_speed_up_larger_than_20_t_64 = 53.7
default_preprocessing_speed_up_t_64_0 = 17.4
default_preprocessing_speed_up_larger_than_20_t_64 = 28.7
default_coarsening_speed_up_t_64_0 = 22.9
default_coarsening_speed_up_larger_than_20_t_64 = 57.4
default_initial_partitioning_speed_up_t_64_0 = 17.9
default_initial_partitioning_speed_up_larger_than_20_t_64 = 44.1
default_uncoarsening_speed_up_t_64_0 = 17.3
default_uncoarsening_speed_up_larger_than_20_t_64 = 46.8
default_fm_speed_up_t_64_0 = 21.1
default_fm_speed_up_larger_than_20_t_64 = 57.2
default_total_speed_up_t_64_100 = 22.3
default_total_fraction_t_64_100 = 40.7
default_preprocessing_speed_up_t_64_100 = 23.1
default_preprocessing_fraction_t_64_100 = 9.6
default_coarsening_speed_up_t_64_100 = 27
default_coarsening_fraction_t_64_100 = 14.1
default_initial_partitioning_speed_up_t_64_100 = 18.8
default_initial_partitioning_fraction_t_64_100 = 10.4
default_uncoarsening_speed_up_t_64_100 = 25.6
default_uncoarsening_fraction_t_64_100 = 25
default_fm_speed_up_t_64_100 = 27.4
default_fm_fraction_t_64_100 = 23.7
quality_total_speed_up_t_4_0 = 3.7
quality_total_speed_up_larger_than_3_t_4 = 98.7
quality_preprocessing_speed_up_t_4_0 = 3.4
quality_preprocessing_speed_up_larger_than_3_t_4 = 83.1
quality_coarsening_speed_up_t_4_0 = 3.3
quality_coarsening_speed_up_larger_than_3_t_4 = 81.2
quality_initial_partitioning_speed_up_t_4_0 = 3.7
quality_initial_partitioning_speed_up_larger_than_3_t_4 = 91.9
quality_uncoarsening_speed_up_t_4_0 = 3.9
quality_uncoarsening_speed_up_larger_than_3_t_4 = 97.7
quality_batch_uncontractions_speed_up_t_4_0 = 3.9
quality_batch_uncontractions_speed_up_larger_than_3_t_4 = 99
quality_fm_speed_up_t_4_0 = 4
quality_fm_speed_up_larger_than_3_t_4 = 98.1
quality_total_speed_up_t_4_100 = 3.7
quality_total_fraction_t_4_100 = 71.8
quality_preprocessing_speed_up_t_4_100 = 2.8
quality_preprocessing_fraction_t_4_100 = 6.5
quality_coarsening_speed_up_t_4_100 = 3.4
quality_coarsening_fraction_t_4_100 = 31.5
quality_initial_partitioning_speed_up_t_4_100 = 3.8
quality_initial_partitioning_fraction_t_4_100 = 7.8
quality_uncoarsening_speed_up_t_4_100 = 3.9
quality_uncoarsening_fraction_t_4_100 = 59.1
quality_batch_uncontractions_speed_up_t_4_100 = 3.9
quality_batch_uncontractions_fraction_t_4_100 = 34.4
quality_fm_speed_up_t_4_100 = 4
quality_fm_fraction_t_4_100 = 39
quality_total_speed_up_t_16_0 = 11.9
quality_total_speed_up_larger_than_10_t_16 = 84.4
quality_preprocessing_speed_up_t_16_0 = 10
quality_preprocessing_speed_up_larger_than_10_t_16 = 71.4
quality_coarsening_speed_up_t_16_0 = 11.2
quality_coarsening_speed_up_larger_than_10_t_16 = 69.2
quality_initial_partitioning_speed_up_t_16_0 = 10.1
quality_initial_partitioning_speed_up_larger_than_10_t_16 = 63.3
quality_uncoarsening_speed_up_t_16_0 = 12
quality_uncoarsening_speed_up_larger_than_10_t_16 = 84.1
quality_batch_uncontractions_speed_up_t_16_0 = 11.2
quality_batch_uncontractions_speed_up_larger_than_10_t_16 = 69.8
quality_fm_speed_up_t_16_0 = 11.2
quality_fm_speed_up_larger_than_10_t_16 = 71.8
quality_total_speed_up_t_16_100 = 12.5
quality_total_fraction_t_16_100 = 71.8
quality_preprocessing_speed_up_t_16_100 = 10.3
quality_preprocessing_fraction_t_16_100 = 6.5
quality_coarsening_speed_up_t_16_100 = 13
quality_coarsening_fraction_t_16_100 = 31.5
quality_initial_partitioning_speed_up_t_16_100 = 11.2
quality_initial_partitioning_fraction_t_16_100 = 7.8
quality_uncoarsening_speed_up_t_16_100 = 12.7
quality_uncoarsening_fraction_t_16_100 = 59.1
quality_batch_uncontractions_speed_up_t_16_100 = 12.1
quality_batch_uncontractions_fraction_t_16_100 = 34.4
quality_fm_speed_up_t_16_100 = 12
quality_fm_fraction_t_16_100 = 39
quality_total_speed_up_t_64_0 = 23.7
quality_total_speed_up_larger_than_20_t_64 = 73.4
quality_preprocessing_speed_up_t_64_0 = 19.7
quality_preprocessing_speed_up_larger_than_20_t_64 = 64.3
quality_coarsening_speed_up_t_64_0 = 25.4
quality_coarsening_speed_up_larger_than_20_t_64 = 68.5
quality_initial_partitioning_speed_up_t_64_0 = 11.6
quality_initial_partitioning_speed_up_larger_than_20_t_64 = 18.8
quality_uncoarsening_speed_up_t_64_0 = 24.3
quality_uncoarsening_speed_up_larger_than_20_t_64 = 76.6
quality_batch_uncontractions_speed_up_t_64_0 = 23.2
quality_batch_uncontractions_speed_up_larger_than_20_t_64 = 71.1
quality_fm_speed_up_t_64_0 = 19
quality_fm_speed_up_larger_than_20_t_64 = 58.4
quality_total_speed_up_t_64_100 = 25.9
quality_total_fraction_t_64_100 = 71.8
quality_preprocessing_speed_up_t_64_100 = 28.6
quality_preprocessing_fraction_t_64_100 = 6.5
quality_coarsening_speed_up_t_64_100 = 36.4
quality_coarsening_fraction_t_64_100 = 31.5
quality_initial_partitioning_speed_up_t_64_100 = 18.2
quality_initial_partitioning_fraction_t_64_100 = 7.8
quality_uncoarsening_speed_up_t_64_100 = 26.5
quality_uncoarsening_fraction_t_64_100 = 59.1
quality_batch_uncontractions_speed_up_t_64_100 = 25.6
quality_batch_uncontractions_fraction_t_64_100 = 34.4
quality_fm_speed_up_t_64_100 = 21.8
quality_fm_fraction_t_64_100 = 39
default_flows_total_speed_up_t_4_0 = 3.1
default_flows_total_speed_up_larger_than_3_t_4 = 57.6
default_flows_preprocessing_speed_up_t_4_0 = 3.3
default_flows_preprocessing_speed_up_larger_than_3_t_4 = 83.9
default_flows_coarsening_speed_up_t_4_0 = 3.7
default_flows_coarsening_speed_up_larger_than_3_t_4 = 98
default_flows_initial_partitioning_speed_up_t_4_0 = 3.8
default_flows_initial_partitioning_speed_up_larger_than_3_t_4 = 87.5
default_flows_uncoarsening_speed_up_t_4_0 = 2.9
default_flows_uncoarsening_speed_up_larger_than_3_t_4 = 42.8
default_flows_total_speed_up_t_4_100 = 3.1
default_flows_total_fraction_t_4_100 = 54.3
default_flows_preprocessing_speed_up_t_4_100 = 3
default_flows_preprocessing_fraction_t_4_100 = 3.9
default_flows_coarsening_speed_up_t_4_100 = 3.3
default_flows_coarsening_fraction_t_4_100 = 5.6
default_flows_initial_partitioning_speed_up_t_4_100 = 3.5
default_flows_initial_partitioning_fraction_t_4_100 = 4.6
default_flows_uncoarsening_speed_up_t_4_100 = 3.1
default_flows_uncoarsening_fraction_t_4_100 = 49
default_flows_total_speed_up_t_16_0 = 7.4
default_flows_total_speed_up_larger_than_10_t_16 = 29.3
default_flows_preprocessing_speed_up_t_16_0 = 9.6
default_flows_preprocessing_speed_up_larger_than_10_t_16 = 55.6
default_flows_coarsening_speed_up_t_16_0 = 11.4
default_flows_coarsening_speed_up_larger_than_10_t_16 = 74.3
default_flows_initial_partitioning_speed_up_t_16_0 = 11.2
default_flows_initial_partitioning_speed_up_larger_than_10_t_16 = 70.4
default_flows_uncoarsening_speed_up_t_16_0 = 5.9
default_flows_uncoarsening_speed_up_larger_than_10_t_16 = 25.7
default_flows_total_speed_up_t_16_100 = 8.4
default_flows_total_fraction_t_16_100 = 54.3
default_flows_preprocessing_speed_up_t_16_100 = 10.1
default_flows_preprocessing_fraction_t_16_100 = 3.9
default_flows_coarsening_speed_up_t_16_100 = 12
default_flows_coarsening_fraction_t_16_100 = 5.6
default_flows_initial_partitioning_speed_up_t_16_100 = 12
default_flows_initial_partitioning_fraction_t_16_100 = 4.6
default_flows_uncoarsening_speed_up_t_16_100 = 8.1
default_flows_uncoarsening_fraction_t_16_100 = 49
default_flows_total_speed_up_t_64_0 = 10.6
default_flows_total_speed_up_larger_than_20_t_64 = 21.7
default_flows_preprocessing_speed_up_t_64_0 = 15.2
default_flows_preprocessing_speed_up_larger_than_20_t_64 = 6.6
default_flows_coarsening_speed_up_t_64_0 = 22.4
default_flows_coarsening_speed_up_larger_than_20_t_64 = 58.9
default_flows_initial_partitioning_speed_up_t_64_0 = 13.6
default_flows_initial_partitioning_speed_up_larger_than_20_t_64 = 28.9
default_flows_uncoarsening_speed_up_t_64_0 = 7.8
default_flows_uncoarsening_speed_up_larger_than_20_t_64 = 20.1
default_flows_total_speed_up_t_64_100 = 14.5
default_flows_total_fraction_t_64_100 = 54.3
default_flows_preprocessing_speed_up_t_64_100 = 22.3
default_flows_preprocessing_fraction_t_64_100 = 3.9
default_flows_coarsening_speed_up_t_64_100 = 32.1
default_flows_coarsening_fraction_t_64_100 = 5.6
default_flows_initial_partitioning_speed_up_t_64_100 = 19.7
default_flows_initial_partitioning_fraction_t_64_100 = 4.6
default_flows_uncoarsening_speed_up_t_64_100 = 13.9
default_flows_uncoarsening_fraction_t_64_100 = 49
\end{filecontents*}
\DTLloaddb[noheader, keys={key,value}]{speedup}{data/scalability.dat}

\begin{table}[!t]
  \centering
  \caption{Geometric mean speedups of the total execution time, preprocessing (P), coarsening (C), initial partitioning (IP), and uncoarsening (UC)
           of the different configurations of \Partitioner{Mt-KaHyPar} over \emph{all} instances and instances with a single-threaded time $\ge 100$s.  }
  \label{tbl:speedups}
  \vspace{-0.25cm}
  \begin{tabular}{lrrrrrrrrr}
     & & \multicolumn{2}{c}{\Partitioner{Mt-KaHyPar-SDet}} & \multicolumn{2}{c}{\Partitioner{Mt-KaHyPar-D}} & \multicolumn{2}{c}{\Partitioner{Mt-KaHyPar-Q}} & \multicolumn{2}{c}{\Partitioner{Mt-KaHyPar-D-F}} \\
    \multicolumn{2}{r}{Num. Threads} & All & $\ge 100$s  & All & $\ge 100$s & All & $\ge 100$s & All & $\ge 100$s \\
    \midrule
    \multirow{3}{*}{Total}
      & 4  & $\placeholder{speedup}{deterministic_total_speed_up_t_4_0}$ & $\placeholder{speedup}{deterministic_total_speed_up_t_4_100}$
           & $\placeholder{speedup}{default_total_speed_up_t_4_0}$ & $\placeholder{speedup}{default_total_speed_up_t_4_100}$
           & $\placeholder{speedup}{quality_total_speed_up_t_4_0}$ & $\placeholder{speedup}{quality_total_speed_up_t_4_100}$
           & $\placeholder{speedup}{default_flows_total_speed_up_t_4_0}$ & $\placeholder{speedup}{default_flows_total_speed_up_t_4_100}$ \\
      & 16 & $\placeholder{speedup}{deterministic_total_speed_up_t_16_0}$ & $\placeholder{speedup}{deterministic_total_speed_up_t_16_100}$
           & $\placeholder{speedup}{default_total_speed_up_t_16_0}$ & $\placeholder{speedup}{default_total_speed_up_t_16_100}$
           & $\placeholder{speedup}{quality_total_speed_up_t_16_0}$ & $\placeholder{speedup}{quality_total_speed_up_t_16_100}$
           & $\placeholder{speedup}{default_flows_total_speed_up_t_16_0}$ & $\placeholder{speedup}{default_flows_total_speed_up_t_16_100}$ \\
      & 64 & $\placeholder{speedup}{deterministic_total_speed_up_t_64_0}$ & $\placeholder{speedup}{deterministic_total_speed_up_t_64_100}$
           & $\placeholder{speedup}{default_total_speed_up_t_64_0}$ & $\placeholder{speedup}{default_total_speed_up_t_64_100}$
           & $\placeholder{speedup}{quality_total_speed_up_t_64_0}$ & $\placeholder{speedup}{quality_total_speed_up_t_64_100}$
           & $\placeholder{speedup}{default_flows_total_speed_up_t_64_0}$ & $\placeholder{speedup}{default_flows_total_speed_up_t_64_100}$ \\
    \midrule
    \multirow{3}{*}{P}
      & 4  & $\placeholder{speedup}{deterministic_preprocessing_speed_up_t_4_0}$ & $\placeholder{speedup}{deterministic_preprocessing_speed_up_t_4_100}$
           & $\placeholder{speedup}{default_preprocessing_speed_up_t_4_0}$ & $\placeholder{speedup}{default_preprocessing_speed_up_t_4_100}$
           & $\placeholder{speedup}{quality_preprocessing_speed_up_t_4_0}$ & $\placeholder{speedup}{quality_preprocessing_speed_up_t_4_100}$
           & $\placeholder{speedup}{default_flows_preprocessing_speed_up_t_4_0}$ & $\placeholder{speedup}{default_flows_preprocessing_speed_up_t_4_100}$ \\
      & 16 & $\placeholder{speedup}{deterministic_preprocessing_speed_up_t_16_0}$ & $\placeholder{speedup}{deterministic_preprocessing_speed_up_t_16_100}$
           & $\placeholder{speedup}{default_preprocessing_speed_up_t_16_0}$ & $\placeholder{speedup}{default_preprocessing_speed_up_t_16_100}$
           & $\placeholder{speedup}{quality_preprocessing_speed_up_t_16_0}$ & $\placeholder{speedup}{quality_preprocessing_speed_up_t_16_100}$
           & $\placeholder{speedup}{default_flows_preprocessing_speed_up_t_16_0}$ & $\placeholder{speedup}{default_flows_preprocessing_speed_up_t_16_100}$ \\
      & 64 & $\placeholder{speedup}{deterministic_preprocessing_speed_up_t_64_0}$ & $\placeholder{speedup}{deterministic_preprocessing_speed_up_t_64_100}$
           & $\placeholder{speedup}{default_preprocessing_speed_up_t_64_0}$ & $\placeholder{speedup}{default_preprocessing_speed_up_t_64_100}$
           & $\placeholder{speedup}{quality_preprocessing_speed_up_t_64_0}$ & $\placeholder{speedup}{quality_preprocessing_speed_up_t_64_100}$
           & $\placeholder{speedup}{default_flows_preprocessing_speed_up_t_64_0}$ & $\placeholder{speedup}{default_flows_preprocessing_speed_up_t_64_100}$ \\
    \midrule
    \multirow{3}{*}{C}
      & 4  & $\placeholder{speedup}{deterministic_coarsening_speed_up_t_4_0}$ & $\placeholder{speedup}{deterministic_coarsening_speed_up_t_4_100}$
           & $\placeholder{speedup}{default_coarsening_speed_up_t_4_0}$ & $\placeholder{speedup}{default_coarsening_speed_up_t_4_100}$
           & $\placeholder{speedup}{quality_coarsening_speed_up_t_4_0}$ & $\placeholder{speedup}{quality_coarsening_speed_up_t_4_100}$
           & $\placeholder{speedup}{default_flows_coarsening_speed_up_t_4_0}$ & $\placeholder{speedup}{default_flows_coarsening_speed_up_t_4_100}$ \\
      & 16 & $\placeholder{speedup}{deterministic_coarsening_speed_up_t_16_0}$ & $\placeholder{speedup}{deterministic_coarsening_speed_up_t_16_100}$
           & $\placeholder{speedup}{default_coarsening_speed_up_t_16_0}$ & $\placeholder{speedup}{default_coarsening_speed_up_t_16_100}$
           & $\placeholder{speedup}{quality_coarsening_speed_up_t_16_0}$ & $\placeholder{speedup}{quality_coarsening_speed_up_t_16_100}$
           & $\placeholder{speedup}{default_flows_coarsening_speed_up_t_16_0}$ & $\placeholder{speedup}{default_flows_coarsening_speed_up_t_16_100}$ \\
      & 64 & $\placeholder{speedup}{deterministic_coarsening_speed_up_t_64_0}$ & $\placeholder{speedup}{deterministic_coarsening_speed_up_t_64_100}$
           & $\placeholder{speedup}{default_coarsening_speed_up_t_64_0}$ & $\placeholder{speedup}{default_coarsening_speed_up_t_64_100}$
           & $\placeholder{speedup}{quality_coarsening_speed_up_t_64_0}$ & $\placeholder{speedup}{quality_coarsening_speed_up_t_64_100}$
           & $\placeholder{speedup}{default_flows_coarsening_speed_up_t_64_0}$ & $\placeholder{speedup}{default_flows_coarsening_speed_up_t_64_100}$ \\
    \midrule
    \multirow{3}{*}{IP}
      & 4  & $\placeholder{speedup}{deterministic_initial_partitioning_speed_up_t_4_0}$ & $\placeholder{speedup}{deterministic_initial_partitioning_speed_up_t_4_100}$
           & $\placeholder{speedup}{default_initial_partitioning_speed_up_t_4_0}$ & $\placeholder{speedup}{default_initial_partitioning_speed_up_t_4_100}$
           & $\placeholder{speedup}{quality_initial_partitioning_speed_up_t_4_0}$ & $\placeholder{speedup}{quality_initial_partitioning_speed_up_t_4_100}$
           & $\placeholder{speedup}{default_flows_initial_partitioning_speed_up_t_4_0}$ & $\placeholder{speedup}{default_flows_initial_partitioning_speed_up_t_4_100}$ \\
      & 16 & $\placeholder{speedup}{deterministic_initial_partitioning_speed_up_t_16_0}$ & $\placeholder{speedup}{deterministic_initial_partitioning_speed_up_t_16_100}$
           & $\placeholder{speedup}{default_initial_partitioning_speed_up_t_16_0}$ & $\placeholder{speedup}{default_initial_partitioning_speed_up_t_16_100}$
           & $\placeholder{speedup}{quality_initial_partitioning_speed_up_t_16_0}$ & $\placeholder{speedup}{quality_initial_partitioning_speed_up_t_16_100}$
           & $\placeholder{speedup}{default_flows_initial_partitioning_speed_up_t_16_0}$ & $\placeholder{speedup}{default_flows_initial_partitioning_speed_up_t_16_100}$ \\
      & 64 & $\placeholder{speedup}{deterministic_initial_partitioning_speed_up_t_64_0}$ & $\placeholder{speedup}{deterministic_initial_partitioning_speed_up_t_64_100}$
           & $\placeholder{speedup}{default_initial_partitioning_speed_up_t_64_0}$ & $\placeholder{speedup}{default_initial_partitioning_speed_up_t_64_100}$
           & $\placeholder{speedup}{quality_initial_partitioning_speed_up_t_64_0}$ & $\placeholder{speedup}{quality_initial_partitioning_speed_up_t_64_100}$
           & $\placeholder{speedup}{default_flows_initial_partitioning_speed_up_t_64_0}$ & $\placeholder{speedup}{default_flows_initial_partitioning_speed_up_t_64_100}$ \\
    \midrule
    \multirow{3}{*}{UC}
      & 4  & $\placeholder{speedup}{deterministic_uncoarsening_speed_up_t_4_0}$ & $\placeholder{speedup}{deterministic_uncoarsening_speed_up_t_4_100}$
           & $\placeholder{speedup}{default_uncoarsening_speed_up_t_4_0}$ & $\placeholder{speedup}{default_uncoarsening_speed_up_t_4_100}$
           & $\placeholder{speedup}{quality_uncoarsening_speed_up_t_4_0}$ & $\placeholder{speedup}{quality_uncoarsening_speed_up_t_4_100}$
           & $\placeholder{speedup}{default_flows_uncoarsening_speed_up_t_4_0}$ & $\placeholder{speedup}{default_flows_uncoarsening_speed_up_t_4_100}$ \\
      & 16 & $\placeholder{speedup}{deterministic_uncoarsening_speed_up_t_16_0}$ & $\placeholder{speedup}{deterministic_uncoarsening_speed_up_t_16_100}$
           & $\placeholder{speedup}{default_uncoarsening_speed_up_t_16_0}$ & $\placeholder{speedup}{default_uncoarsening_speed_up_t_16_100}$
           & $\placeholder{speedup}{quality_uncoarsening_speed_up_t_16_0}$ & $\placeholder{speedup}{quality_uncoarsening_speed_up_t_16_100}$
           & $\placeholder{speedup}{default_flows_uncoarsening_speed_up_t_16_0}$ & $\placeholder{speedup}{default_flows_uncoarsening_speed_up_t_16_100}$ \\
      & 64 & $\placeholder{speedup}{deterministic_uncoarsening_speed_up_t_64_0}$ & $\placeholder{speedup}{deterministic_uncoarsening_speed_up_t_64_100}$
           & $\placeholder{speedup}{default_uncoarsening_speed_up_t_64_0}$ & $\placeholder{speedup}{default_uncoarsening_speed_up_t_64_100}$
           & $\placeholder{speedup}{quality_uncoarsening_speed_up_t_64_0}$ & $\placeholder{speedup}{quality_uncoarsening_speed_up_t_64_100}$
           & $\placeholder{speedup}{default_flows_uncoarsening_speed_up_t_64_0}$ & $\placeholder{speedup}{default_flows_uncoarsening_speed_up_t_64_100}$ \\
  \end{tabular}
\end{table}

\paragraph{Scalability}
In Figure~\ref{fig:scalability_mt_kahypar} and~\ref{fig:scalability_flows}, and Table~\ref{tbl:speedups}, we summarize self-relative speedups
of \Partitioner{Mt-KaHyPar} for each configuration and the different phases of the multilevel scheme with an increasing number of threads $t \in \{1,4,16,64\}$.
The scalability experiments run on \largehg. However, we used a subset for \Partitioner{Mt-KaHyPar-Q} (77 out of 94 hypergraphs) and \Partitioner{Mt-KaHyPar-D-F} (76 out of 94 hypergraphs) to ensure
reasonable running times. This set consists of instances where \Partitioner{Mt-KaHyPar-Q}/\Partitioner{-D-F} was able to finish in under 600 seconds with $64$ threads
for all tested values of $k$. The experiment still took $6$ weeks to complete for each configuration.
Note that we only rerun the experiments for \Partitioner{Mt-KaHyPar-SDet}/\Partitioner{-D} for this work, while the speedups of \Partitioner{Mt-KaHyPar-Q}/\Partitioner{-D-F}
are based on the data from the corresponding conference publications~\cite{MT-KAHYPAR-Q,MT-KAHYPAR-FLOWS} due to the high time requirements.
In the plot, we represent the speedup (y-axis)
of each instance as a point and the centered rolling geometric mean over the points with a window size of $25$ as a line. The x-axis shows the single-threaded
running time of the corresponding configuration resp.~component\footnote{In contrast
to many other publications in the parallel partitioning community, we do not correlate
speedups to any of the common hypergraph metrics (such as the number of pins).
We found that the running time often depends on a variety of different factors. Fitting
suitable parameters for a combination of the metrics seem much more complicated
than plotting against sequential running time, which is often nicely correlated with
speedups.}.

\begin{figure}
  \centering
  \begin{minipage}{\textwidth}
    \centering
  \ifpdfplots
    \includegraphics{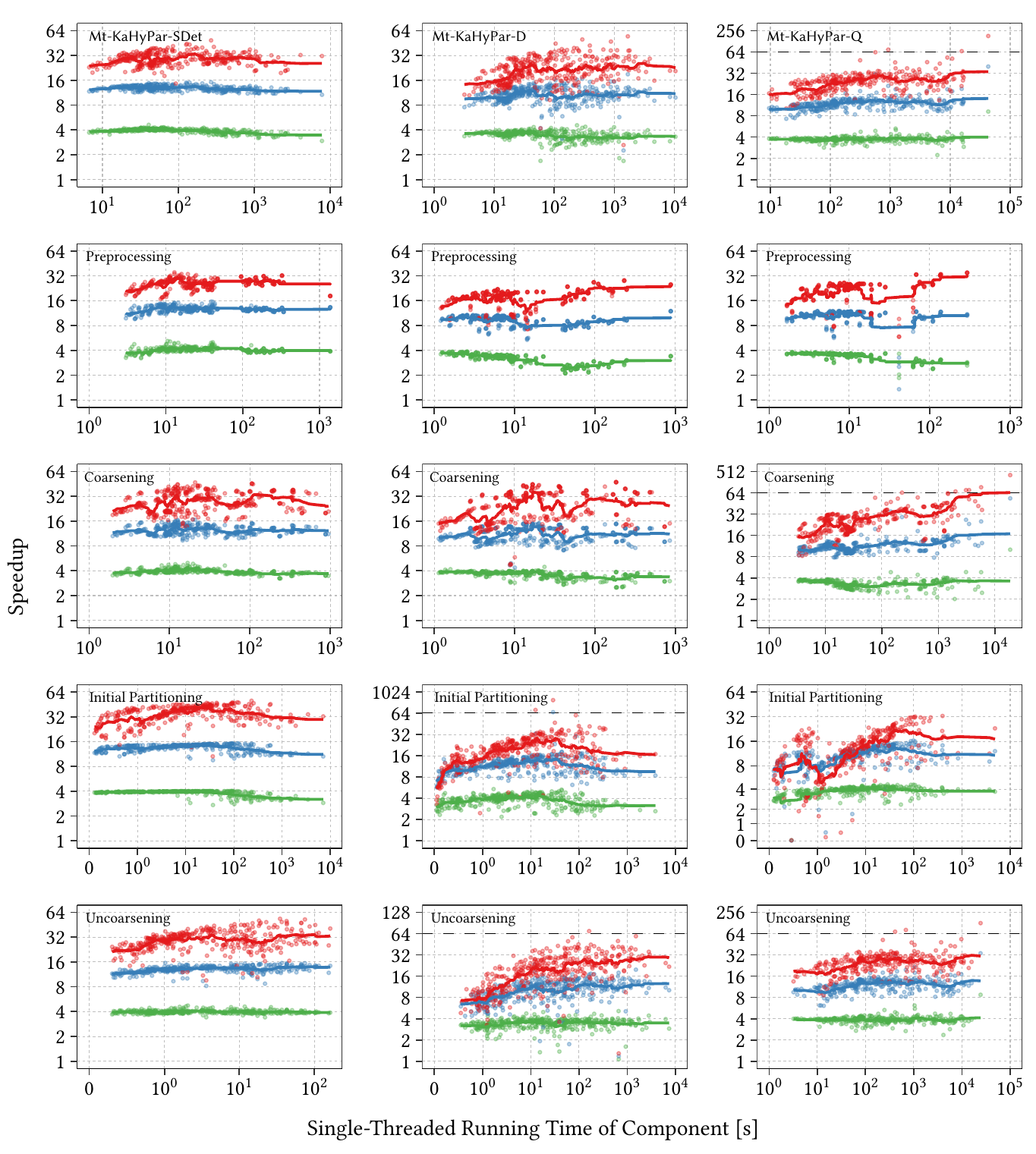}
  \fi
  \end{minipage} %
  \begin{minipage}{\textwidth}
    \vspace{0.1cm}
    \centering
  \ifpdfplots
    \includegraphics{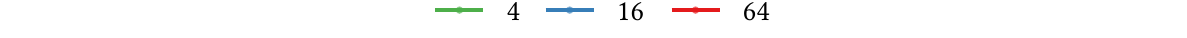}
  \fi
  \end{minipage} %
  \vspace{-0.5cm}
  \caption{Speedups of \Partitioner{Mt-KaHyPar-SDet} (left), \Partitioner{Mt-KaHyPar-D} (middle), and \Partitioner{Mt-KaHyPar-Q} (right).}
	\label{fig:scalability_mt_kahypar}
\end{figure}

\begin{figure}
  \centering
  \begin{minipage}{\textwidth}
    \centering
  \ifpdfplots
    \includegraphics{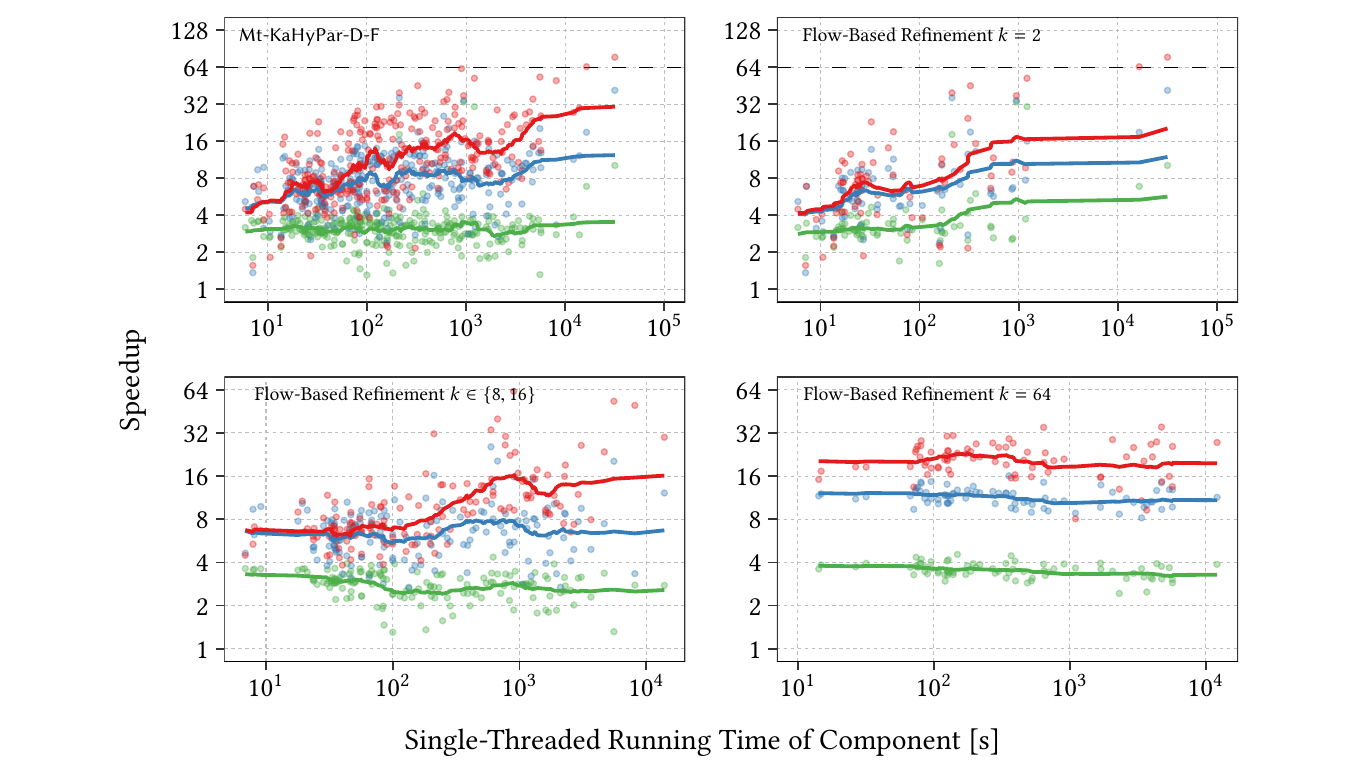}
  \fi
  \end{minipage} %
  \begin{minipage}{\textwidth}
    \vspace{0.25cm}
    \centering
  \ifpdfplots
    \includegraphics{pdf_plots/scalability_legend.pdf}
  \fi
  \end{minipage} %
  \vspace{-0.5cm}
  \caption{Speedups of \Partitioner{Mt-KaHyPar-D-F} and flow-based refinement for different values of $k$.}
	\label{fig:scalability_flows}
\end{figure}

\begin{figure}
  \centering
  \begin{minipage}{\textwidth}
    \centering
  \ifpdfplots
    \includegraphics{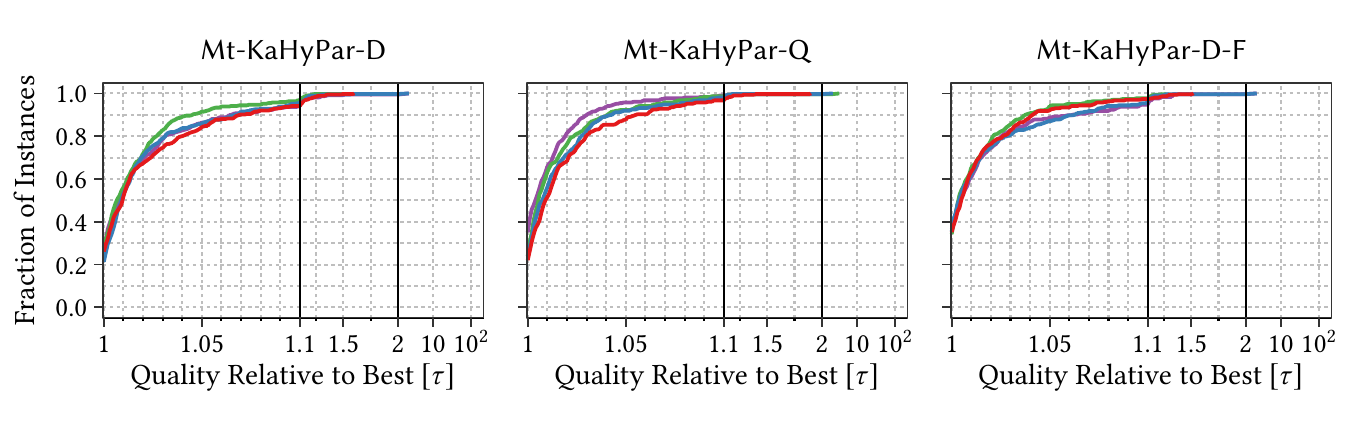}
  \fi
  \end{minipage} %
  \begin{minipage}{\textwidth}
    \vspace{-0.1cm}
    \centering
  \ifpdfplots
    \includegraphics{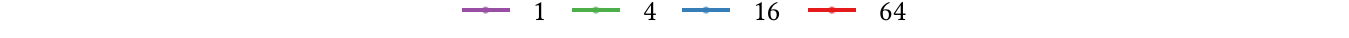}
  \fi
  \end{minipage} %
  \vspace{-0.5cm}
  \caption{Performance profiles comparing the solution quality of \Partitioner{Mt-KaHyPar}
    with an increasing number of threads on \largehg.}
	\label{fig:scalability_quality}
\end{figure}

The overall geometric mean speedup of \Partitioner{Mt-KaHyPar-D} is $\placeholder{speedup}{default_total_speed_up_t_4_0}$ for $t = 4$,
$\placeholder{speedup}{default_total_speed_up_t_16_0}$ for $t = 16$, and $\placeholder{speedup}{default_total_speed_up_t_64_0}$ for $t = 64$.
If we only consider instances with a single-threaded running time $\ge 100$s, the geometric mean speedup increases to $\placeholder{speedup}{default_total_speed_up_t_64_100}$
for $t = 64$. For $t = 4$, the speedup is at least $3$ on $\placeholder{speedup}{default_total_speed_up_larger_than_3_t_4}\%$ of the instances.
The community detection algorithm (refered to as \emph{preprocessing}) and coarsening share many similarities in their implementation and both
show reliable speedups for an increasing number of threads. For initial partitioning and the uncoarsening phase, we observe that longer single-threaded
execution times leads to substantially better speedups. The most time-consuming component of the uncoarsening phase is the FM algorithm.
The geometric mean speedup of the FM algorithm is $\placeholder{speedup}{default_fm_speed_up_t_64_0}$ for $t = 64$, which increases to
$\placeholder{speedup}{default_fm_speed_up_t_64_100}$ for instances with sequential time $\ge 100$s.

If we compare the speedups of \Partitioner{Mt-KaHyPar-SDet} to its non-deterministic counterpart \Partitioner{Mt-KaHyPar-D}, we see that it achieves much more reliable speedups.
Especially, the speedups of initial partitioning increases substantially with a geometric mean speedup of
$\placeholder{speedup}{deterministic_initial_partitioning_speed_up_t_64_0}$ for $t = 64$.
Since \Partitioner{Mt-KaHyPar-SDet} does not adaptively adjust the number of repetitions in the bipartitioning
portfolio, it performs more work in the initial partitioning phase and is not affected by non-deterministic decisions, which increases its scalability
(geometric mean running time of initial partitioning is $9.36$s in \Partitioner{Mt-KaHyPar-SDet} vs $4.23$s in \Partitioner{Mt-KaHyPar-D} for $t = 1$).
The overall geometric mean speedup of \Partitioner{Mt-KaHyPar-SDet} is $\placeholder{speedup}{deterministic_total_speed_up_t_4_0}$ for $t = 4$,
$\placeholder{speedup}{deterministic_total_speed_up_t_16_0}$ for $t = 16$, and $\placeholder{speedup}{deterministic_total_speed_up_t_64_0}$ for $t = 64$.

The coarsening and batch uncontraction algorithm are the components that differentiate our $n$-level partitioning algorithm \Partitioner{Mt-KaHyPar-Q}
from the other multilevel algorithms.
Both components exhibit good speedups, while coarsening (geometric mean speedup is
$\placeholder{speedup}{quality_coarsening_speed_up_t_64_0}$ for $t = 64$) scales slightly better than the batch uncontractions
($\placeholder{speedup}{quality_batch_uncontractions_speed_up_t_64_0}$ for $t = 64$). Moreover, the speedups of the localized version of
FM algorithm that runs after each batch uncontraction operation are less pronounced than the speedups of the FM algorithm in \Partitioner{Mt-KaHyPar-D}
(geometric mean speedup $\placeholder{speedup}{quality_fm_speed_up_t_64_0}$ vs $\placeholder{speedup}{default_fm_speed_up_t_64_0}$ for $t = 64$).
Note that we also observe super-linear speedups, which are caused by non-deterministic coarsening decisions.
The geometric mean speedup of \Partitioner{Mt-KaHyPar-Q} is $\placeholder{speedup}{quality_total_speed_up_t_4_0}$ for $t = 4$,
$\placeholder{speedup}{quality_total_speed_up_t_16_0}$ for $t = 16$, and $\placeholder{speedup}{quality_total_speed_up_t_64_0}$ for $t = 64$.

\Partitioner{Mt-KaHyPar-D-F} extends \Partitioner{Mt-KaHyPar-D} with flow-based refinement.
We therefore only show speedups for this component in Figure~\ref{fig:scalability_flows}.
Unfortunely, the speedups are less promising as for the other configurations.
The geometric mean speedup of \Partitioner{Mt-KaHyPar-D-F} is $\placeholder{speedup}{default_flows_total_speed_up_t_4_0}$ for $t = 4$,
$\placeholder{speedup}{default_flows_total_speed_up_t_16_0}$ for $t = 16$, and $\placeholder{speedup}{default_flows_total_speed_up_t_64_0}$ for $t = 64$.
However, we achieve better speedups for larger values of $k$ where all parallelism is leveraged in the scheduling algorithm, and none in
the \Algorithm{FlowCutter} algorithm. For $k = 2$, the scalability depends on our parallel maximum flow algorithm for which we observe similar
speedups as reported in Ref.~\cite{BaumstarkSyncPushRelabel} -- the work on which our parallel implementation is based on.
Thus, increasing the scalability of maximum flow algorithms is an important avenue for future research.

In Figure~\ref{fig:scalability_quality}, we compare the solution quality of the different configurations when increasing the number
of threads. We can see that using more threads adversely affects the solution quality of the partitions produced by
\Partitioner{Mt-KaHyPar-Q}, but only by a small margin (solution are $0.4\%$ better with one compared to 64 threads).
\Partitioner{Mt-KaHyPar-D} and \Partitioner{Mt-KaHyPar-D-F} produce comparable solutions when increasing the number of threads.

%

\paragraph{Effects of Graph Optimizations}

In Section~\ref{sec:parallel_graph}, we presented optimized data structures for graph partitioning used as a drop-in replacement
in our partitioning algorithm. Figure~\ref{fig:graph_vs_hypergraph} shows their impact on the solution quality and speedups for
different algorithmic components of \Partitioner{Mt-KaHyPar-D} on \largegr. As it can be seen, replacing our hypergraph with
the graph data structures does not adversely affect the solution quality of \Partitioner{Mt-KaHyPar-D} as both performance lines
are almost identical and converges quickly towards $y = 1$.

\begin{figure}
  \centering
  \ifpdfplots
    \includegraphics{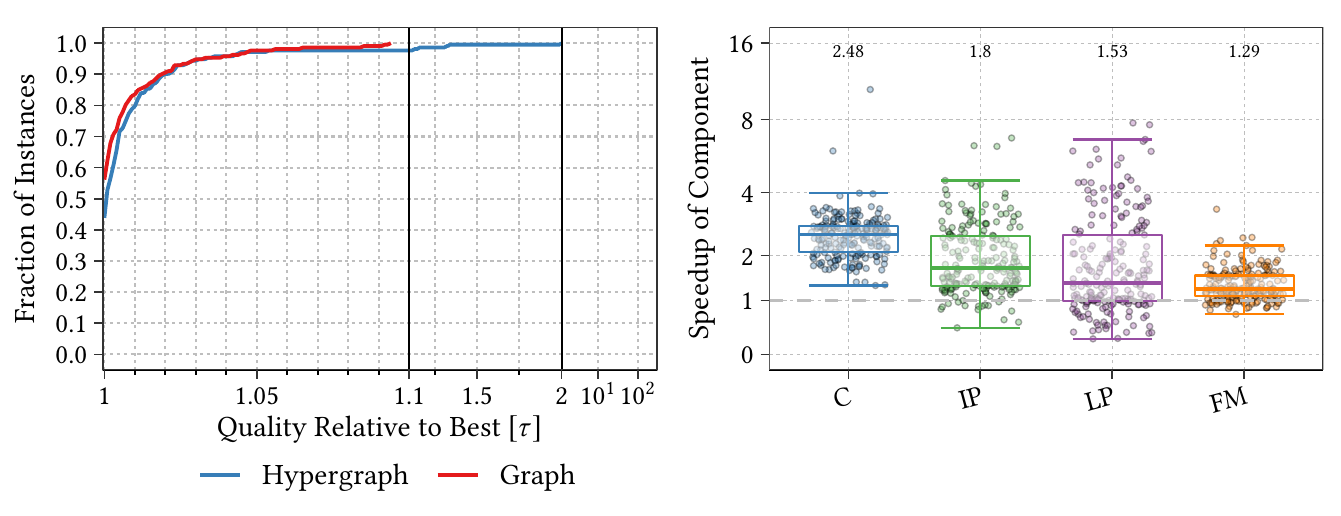}
  \fi
  \vspace{-0.5cm}
  \caption{Performance profile (left) and speedups (right) of coarsening (C), initial partitioning (IP),
    label propagation (LP), and FM refinement comparing \Partitioner{Mt-KaHyPar-D} with and without our optimized
    graph data structure on \largegr.}
	\label{fig:graph_vs_hypergraph}
\end{figure}

The coarsening algorithm benefits most from our optimized graph data structure (geometric mean speedup $2.48$).
The hypergraph version computes a clustering of the nodes by iterating over the pin-lists of nets
to aggregate ratings, and subsequently contract that clustering by collapsing two adjacency arrays (one for the
pin-lists and one for the incident nets). The cache-friendly memory layout for graphs (only one adjacency array for neighbors)
leads to faster access times to enumerate neighbors and to a simpler contraction algorithm. The FM algorithm has the least promising
speedups ($1.29$). One of the most time-consuming parts of the algorithm is retrieving and updating entries from the gain
table, which has the same asymptotic worst-case complexity in both implementations. The initial partitioning phase ($1.8$)
has better speedups than both refinement algorithms but slightly worse speedups than coarsening. This can be explained by
the fact that initial partitioning uses all algorithms within multilevel recursive bipartitioning.

The overall speedup of the graph version of \Partitioner{Mt-KaHyPar-D} over its hypergraph counterpart is $1.75$ on average
(\gmeantime~$10.8$s vs $18.94$s). In the dissertation of Heuer~\cite[p.~150--153]{HEUER-DIS}, we also present an optimized
graph data structure for $n$-level partitioning, which accelerates \Partitioner{Mt-KaHyPar-Q} by a factor of $1.91$ on average
($97.45$s vs $186.32$s). We note that flow-based refinement requires further engineering efforts to handle large graphs efficiently
(currently works only on medium-sized graphs in reasonable running times).
The \Partitioner{FlowCutter} algorithm operates implicitly on the Lawler expansion~\cite{Lawler}
($n + 2m$ nodes and $5m$ edges) of the input (hyper)graph. An optimized version would omit this transformation and
compute a maximum flow directly on the graph representation. However, this issue will be addressed in a future release
of \Partitioner{Mt-KaHyPar}.

\subsection{Comparison to Other Systems}\label{sec:comparison}

We now compare \Partitioner{Mt-KaHyPar} to existing partitioning algorithms to see if it can improve the state-of-the-art.
We did an extensive research on publicly available partitioning tools and were able to include $25$ different sequential and parallel
graph and hypergraph partitioners that we compare on over 800 graphs and hypergraphs. Thus, to the best of our knowledge, this study represents
the most comprehensive comparison of partitioning algorithms to date.
We primarly focus on multilevel algorithms as it has been shown that they provide an excellent trade-off
between solution quality and running time~\cite{DBLP:conf/sc/HendricksonL95,HAUCK}. While there are even faster
partitioning method that omit the multilevel scheme, it has been shown that they are inferior to multilevel
algorithms in terms of solution quality~\cite{HAUCK,KAHYPAR-DIS}. Moreover, algorithms that achieve even higher solution quality than
multilevel algorithms such as evolutionary algorithms~\cite{KAFFPAE,AndreS018}, diffusion-based partitioning \cite{DIBAP,MT-DIBAP,MPI-DIBAP},
and approaches based on integer linear programming~\cite{ILP} would not run in a reasonable time frame on our benchmark sets.

We first provide a description of the partitioning algorithms included in our study and explain how we configured them.
We then identify a subset of Pareto-optimal algorithms to which we then compare \Partitioner{Mt-KaHyPar}\footnote{We made
all experimental results publicly available from \url{https://algo2.iti.kit.edu/heuer/talg/}.}.

\begin{table}[t]
  \centering
  \caption{Listing of graph and hypergraph partitioning algorithms (GP and HGP) included in the experimental evaluation.
    For algorithms publicly available on GitHub, we report the first seven characters of the
    corresponding commit hash indicating the used version.}
  \label{tbl:included_partitioners}
  \begin{tabular}{lp{3cm}rp{3cm}rl}
     & \multicolumn{2}{c}{\bf Sequential} & \multicolumn{3}{c}{\bf Parallel} \\
     & Algorithm & Version & Algorithm & Version & Machine Model \\
    \midrule
    \multirow{6}{*}{\rotatebox[origin=c]{90}{\parbox[c]{1cm}{\centering \bf GP}}} &
      \Partitioner{Metis}~\cite{KarypisK98,DBLP:journals/jpdc/KarypisK98a}  & 5.1.0  & \Partitioner{KaMinPar}~\cite{KAMINPAR} & \githash{29101f61f22b917b5c6c824275040f747f2c4b3b} & Shared-Memory \\
    & \footnotesize{\Partitioner{Metis-R} and \Partitioner{Metis-K}} & & \Partitioner{Mt-Metis}~\cite{MT-METIS,MT-METIS-OPT,MT-METIS-REFINEMENT} & 0.6.0 & Shared-Memory \\
    & \Partitioner{KaFFPa}~\cite{KAFFPA,Schulz-DIS} & \githash{f239f7afd21e73d9db349aeb62d578b4b256b93b} & \Partitioner{ParMetis}~\cite{PARMETIS} & 4.0.3 & Distributed-Memory \\
    & \multicolumn{2}{l}{\footnotesize{\Partitioner{KaFFPa-Fast(S)}/\Partitioner{-Eco(S)}/\Partitioner{-Strong(S)}}} & \Partitioner{Mt-KaHIP}~\cite{MT-KAHIP,MT-KAHIP-DIS} & \githash{30de7371a2348945cbbc1ba0c2b8ae26e0ab6793} & Shared-Memory \\
    & \Partitioner{Scotch}~\cite{PellegriniR96} & 6.1.3 & \Partitioner{ParHIP}~\cite{PARHIP}   & \githash{f239f7afd21e73d9db349aeb62d578b4b256b93b} & Distributed-Memory \\
    & & & \multicolumn{2}{l}{\footnotesize{\Partitioner{ParHIP-Fast} and \Partitioner{ParHIP-Eco}}} \\
    \midrule
    \multirow{7}{*}{\rotatebox[origin=c]{90}{\parbox[c]{1cm}{\centering \bf HGP}}} &
    \Partitioner{PaToH}~\cite{PATOH} & 3.3  & \Partitioner{Zoltan}~\cite{ZOLTAN} & 3.83 & Distributed-Memory \\
    & \multicolumn{2}{l}{\footnotesize{\Partitioner{PaToH-D} and \Partitioner{PaToH-Q}}} & \Partitioner{BiPart}~\cite{BIPART} & \githash{49a59a652bf9492fd1f44d4ecebb1a7af47bbac2} & Shared-Memory \\
    & \Partitioner{hMetis}~\cite{HMETIS,HMETIS-K} & 2.0pre1 \\
    & \multicolumn{2}{l}{\footnotesize{\Partitioner{hMetis-R} and \Partitioner{hMetis-K}}} \\
    & \Partitioner{KaHyPar}~\cite{KAHYPAR-DIS} & \githash{876b7768a92614747e6427e596b89506349c7561} \\
    & \multicolumn{2}{l}{\footnotesize{\Partitioner{KaHyPar-CA}, \Partitioner{$r$KaHyPar}, and \Partitioner{$k$KaHyPar}}} \\
    & \Partitioner{Mondriaan}~\cite{MONDRIAAN} & 4.2.1 \\
  \end{tabular}
\end{table}

\paragraph{Included Algorithms}

Table~\ref{tbl:included_partitioners} lists all partitioning algorithms included in the following experimental evaluation.
Many of these algorithms provide multiple partitioning configurations
offering different trade-offs in running time and solution quality (e.g., \Partitioner{KaFFPa-Fast}/\Partitioner{-Eco}/\Partitioner{-Strong},
or the default (\Partitioner{-D}) and quality preset (\Partitioner{-Q}) of \Partitioner{PaToH}), or are
based on either recursive bipartitioning (e.g., \Partitioner{hMetis-R}) or direct $k$-way partitioning (e.g., \Partitioner{hMetis-K}).
The graph partitioner \Partitioner{KaFFPa} also provides different settings for partitioning \emph{social} networks (\Partitioner{KaFFPa-FastS}/\Partitioner{-EcoS}/\Partitioner{-StrongS}).
Thus, we include all three social configurations as well as their non-social counterparts (\Partitioner{KaFFPa-Fast}/\Partitioner{-Eco}/\Partitioner{-Strong}).
For the $n$-level algorithm \Partitioner{KaHyPar}, we include the recursive bipartitioning (\Partitioner{$r$KaHyPar}) and direct $k$-way version
(\Partitioner{$k$KaHyPar}, which uses similar algorithmic components as \Partitioner{Mt-KaHyPar-Q-F}), as well as a configuration without flow-based refinement (\Partitioner{KaHyPar-CA},
which uses similar algorithmic components as \Partitioner{Mt-KaHyPar-Q}).

Unfortunely, we were not able to include the publicly available versions of \Partitioner{Parkway}~\cite{PARKWAY-2} (distributed-memory),
\Partitioner{PT-Scotch}~\cite{PT-SCOTCH} (distributed-memory), and \Partitioner{Chaco}~\cite{DBLP:conf/sc/HendricksonL95} (sequential). These algorithms
failed with segmentation faults on most instances of our benchmark sets.

\paragraph{Algorithm Configuration}


We configure all graph partitioning algorithms to optimize the edge cut metric, while we optimize the connectivity
metric for hypergraph partitioning. We run \Partitioner{Mt-KaHyPar} using \emph{ten} threads for comparisons to sequential
algorithms as this is a typical number of available cores in a modern commodity workstation.
We add a suffix to the name of parallel algorithms indicating the number of threads
used, e.g., \Partitioner{Mt-KaHyPar 64} for 64 threads. We omit the suffix for sequential algorithms.
For graph partitioning, \Partitioner{Mt-KaHyPar} uses the partition and graph data structure presented in Section~\ref{sec:parallel_graph}.

We use the default settings provided by the authors to configure the different partitioning algorithms. However, for algorithms based on recursive bipartitioning,
we adjust the input imbalance parameter $\varepsilon$ to $\varepsilon' := (1 + \varepsilon)^{\frac{1}{\lceil \log_2{k} \rceil}}$ (based
on Equation~\ref{eq:adaptive_epsilon} by applying it to the first bipartitioning step) when we observed that most of the computed partitions are imbalanced.
This applies to \Partitioner{Metis-R}, \Partitioner{hMetis-R}, and \Partitioner{BiPart}.
We further set \Partitioner{hMetis} to optimize the \emph{sum-of-external-degree} metric $\osoed(\Partition) := \sum_{e \in \cutnets} \lambda(e) \cdot \omega(e)
= \ocon(\Partition) + \ocut(\Partition)$ (connectivity plus cut-net metric) and calculate the connectivity metric accordingly.
We additionally configure \Partitioner{Mt-Metis} to use its hill-scanning refinement algorithm~\cite{MT-METIS-REFINEMENT}.
Moreover, we do not perform multiple repetitions when running \Partitioner{Scotch} or \Partitioner{BiPart} as both do not provide
a command line parameter for setting a seed value.


\begin{table}[t]
  \centering
  \footnotesize
  \caption{Summary of algorithms (first column) outperforming others (second column). It shows the median improvement in the connectivity
           resp.~edge cut metric in percent for each baseline over the outperformed algorithm and the average slowdown of the outperformed
           relative to the baseline algorithm.}
  \label{tbl:outperformed}
  \begin{tabular}{lllcc|llcc}
     & \multicolumn{4}{c|}{\bf Sequential} & \multicolumn{4}{c}{\bf Parallel (64 threads)} \\
     & Base Algo. & Outperformed & Med. [$\%$] & Rel. Slow. & Base Algo. & Outperformed & Med. [$\%$] & Rel. Slow.  \\
     \midrule
     \multirow{5}{*}{\rotatebox[origin=c]{90}{\parbox[c]{1cm}{\centering \bf GP}}} &
     \Partitioner{Metis-K}  & \Partitioner{Metis-R} & 2.9 & 1.4 & \Partitioner{KaMinPar} & \Partitioner{Mt-Metis} & 0 & 9.11 \\
     & \Partitioner{Metis-K} & \Partitioner{KaFFPa-Fast} & 5.8 & 4.3 & \Partitioner{KaMinPar} & \Partitioner{ParMetis} & 4.4 & 211.2 \\
     & \Partitioner{Metis-K} & \Partitioner{KaFFPa-FastS} & 2.2 & 4,79 & \Partitioner{KaMinPar} & \Partitioner{ParHIP-Fast} & 2.8 & 8.18 \\
     & \Partitioner{Metis-K} & \Partitioner{Scotch} & 2.5 & 4.66 & \Partitioner{Mt-KaHIP} & \Partitioner{ParHIP-Eco} & 2.2 & 11.62 \\
     & \Partitioner{KaFFPa-EcoS} & \Partitioner{KaFFPa-Eco} & 3.2 & 1.04 &  \\
     \midrule
     \multirow{4}{*}{\rotatebox[origin=c]{90}{\parbox[c]{1cm}{\centering \bf HGP}}} &
       \Partitioner{PaToH-D} & \Partitioner{Mondriaan} & 0.6 & 5.63 & \Partitioner{Zoltan} & \Partitioner{BiPart} & 69 & 2.31 \\
     & \Partitioner{KaHyPar-CA} & \Partitioner{hMetis-R} & 0.5 & 3.31 \\
     & \Partitioner{KaHyPar-CA} & \Partitioner{hMetis-K} & 2.6 & 2.62 \\
     & \Partitioner{$k$KaHyPar} & \Partitioner{$r$KaHyPar} & 2.1 & $\sim 1$ \\
  \end{tabular}
\end{table}

\paragraph{Identifying Competitors}

Since some of the included algorithms already outperform others with regards to solution quality and running time,
we compare \Partitioner{Mt-KaHyPar} only to a subset of Pareto-optimal partitioning algorithms.
Table~\ref{tbl:outperformed} presents a summary of the results that we used to identify our main competitors.
The data is based on a detailed evaluation that can be found in the dissertation of Heuer~\cite[see Section~8.2 on p.~160--167]{HEUER-DIS}.
We added the performance profiles and running time plots used for this evaluation in Appendix~\ref{appendix:comparison}.
In the table, the algorithms in the second column are outperformed by the algorithms in the first column and are therefore excluded
from the following experimental evaluation.
The included systems can be classified into fast partitioning methods (\Partitioner{PaToH-D}, \Partitioner{Zoltan}, \Partitioner{Metis-K}, and
\Partitioner{KaMinPar}), configurations providing a good trade-off between solution quality and running time (\Partitioner{PaToH-Q}, \Partitioner{KaFFPa-EcoS},
and \Partitioner{Mt-KaHIP}), and high-quality partitioning algorithms (\Partitioner{KaHyPar-CA}, \Partitioner{$k$KaHyPar}, and \Partitioner{KaFFPa-Strong}/\Partitioner{-StrongS}).
To simplify the following evaluation, we compare the high quality algorithms to \Partitioner{Mt-KaHyPar-Q-F} (highest quality configuration)
and all others to \Partitioner{Mt-KaHyPar-D} (fastest configuration).

\paragraph{Comparison to Sequential Systems}

Figure~\ref{fig:sequential_hgp} compares \Partitioner{Mt-KaHyPar} to the sequential hypergraph partitioners
\Partitioner{PaToH} and \Partitioner{KaHyPar} on \mediumhg.
In an individual comparison,
\Partitioner{Mt-KaHyPar-D} (\gmeantime~$0.88$s) computes better partitions than
\Partitioner{PaToH-D} ($1.17$s) and \Partitioner{PaToH-Q} ($5.85$s) on $82.9\%$ and $58.34\%$ of the instances
(median improvement is $6.6\%$ and $1.2\%$)\footnote{It appears that \Partitioner{Mt-KaHyPar-D} performs slightly
worse than \Partitioner{PaToH-Q} in the performance profiles. However, if we would compare them in a performance profile
individually, we would see that the performance line of \Partitioner{Mt-KaHyPar-D} lies strictly above the line of
\Partitioner{PaToH-Q}. We therefore point out that performance profiles do not permit a full ranking between three or more algorithms.},
while it achieves a speedup of $1.32$ \wrt~\Partitioner{PaToH-D}
and $6.6$ \wrt~\Partitioner{PaToH-Q} with ten threads on average.
Thus, \Partitioner{Mt-KaHyPar-D} outperforms \Partitioner{PaToH-D} and \Partitioner{PaToH-Q}.

We can also see that the performance lines of \Partitioner{Mt-KaHyPar-Q-F} and \Partitioner{$k$KaHyPar}
-- the currently best sequential hypergraph partitioning algorithm -- are almost identical, which means that both compute
partitions of comparable solution quality. \Partitioner{Mt-KaHyPar-Q-F} ($5.08$s)
is faster than \Partitioner{KaHyPar-CA} ($28.14$s) and \Partitioner{$k$KaHyPar} ($48.97$s)
on almost all instances with ten threads ($\ge 99\%$). This shows that we achieved
the same solution quality as the currently highest-quality sequential partitioning algorithm, while being almost an
order of magnitude faster with only ten threads.
Moreover, \Partitioner{Mt-KaHyPar-Q-F} is also slightly faster than \Partitioner{PaToH-Q}, while
it computes  better partitions than \Partitioner{PaToH-Q} on $87.7\%$ of the instances
(median improvement is $6.4\%$).

\begin{figure}
  \centering
  \begin{minipage}{\textwidth}
    \centering
  \ifpdfplots
    \includegraphics{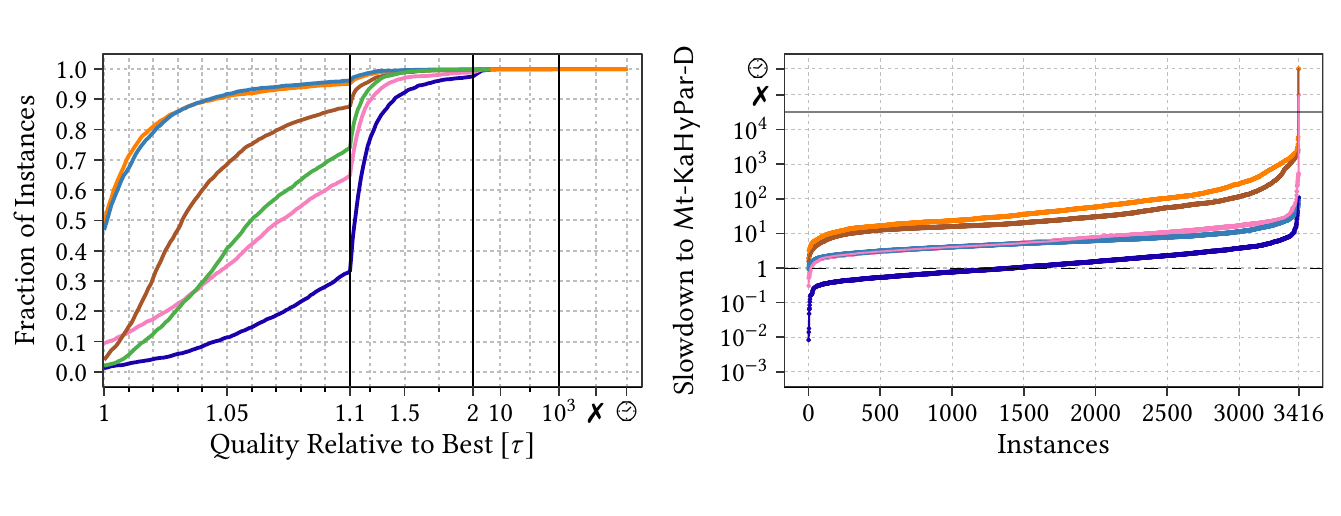}
  \fi
  \end{minipage} %
  \begin{minipage}{\textwidth}
    \vspace{-0.25cm}
    \centering
  \ifpdfplots
    \includegraphics{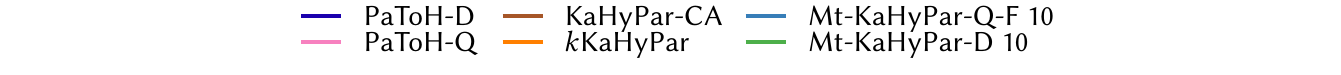}
  \fi
  \end{minipage} %
  \vspace{-0.5cm}
  \caption{Performance profiles and running times comparing \Partitioner{Mt-KaHyPar} to
           \Partitioner{PaToH} and \Partitioner{KaHyPar} on \mediumhg.}
	\label{fig:sequential_hgp}
\end{figure}

Figure~\ref{fig:sequential_gp} compares \Partitioner{Mt-KaHyPar} to the sequential graph partitioners \Partitioner{Metis-K}
and \Partitioner{KaFFPa} on \mediumgr.
\Partitioner{Mt-KaHyPar-D} (\gmeantime~$0.55$s) is slightly slower than \Partitioner{Metis-K} ($0.39$s) with ten threads but produces significantly better edge cuts
(median improvement is $5.9\%$). If we disable the FM algorithm in \Partitioner{Mt-KaHyPar-D}, we obtain a configuration that is slightly
faster than \Partitioner{Metis-K}, while the edge cuts are comparable (see Figure~\ref{fig:metis_k_vs_mt_kahypar_s} in Appendix~\ref{appendix:comparison}).

\Partitioner{Mt-KaHyPar-Q-F} ($5.22$s) is faster than \Partitioner{KaFFPa-EcoS} ($10.51$s) and produces better edge cuts
by $2.9\%$ in the median. The differences between the edge cuts computed by
\Partitioner{Mt-KaHyPar-Q-F} and \Partitioner{KaFFPa-Strong} ($162.83$s) are not statiscally
significant ($Z = -2.3101$ and $p = 0.02088$).
Out of all tested algorithms, \Partitioner{KaFFPa-StrongS} ($201.99$s) is the only algorithm producing slightly better edge cuts than
\Partitioner{Mt-KaHyPar-Q-F} (median improvement is $1\%$). However, this comes at the cost of a $38.66$ times
longer running time on average, making the quality improvement questionable in practice.

\begin{figure}
  \centering
  \begin{minipage}{\textwidth}
    \centering
  \ifpdfplots
    \includegraphics{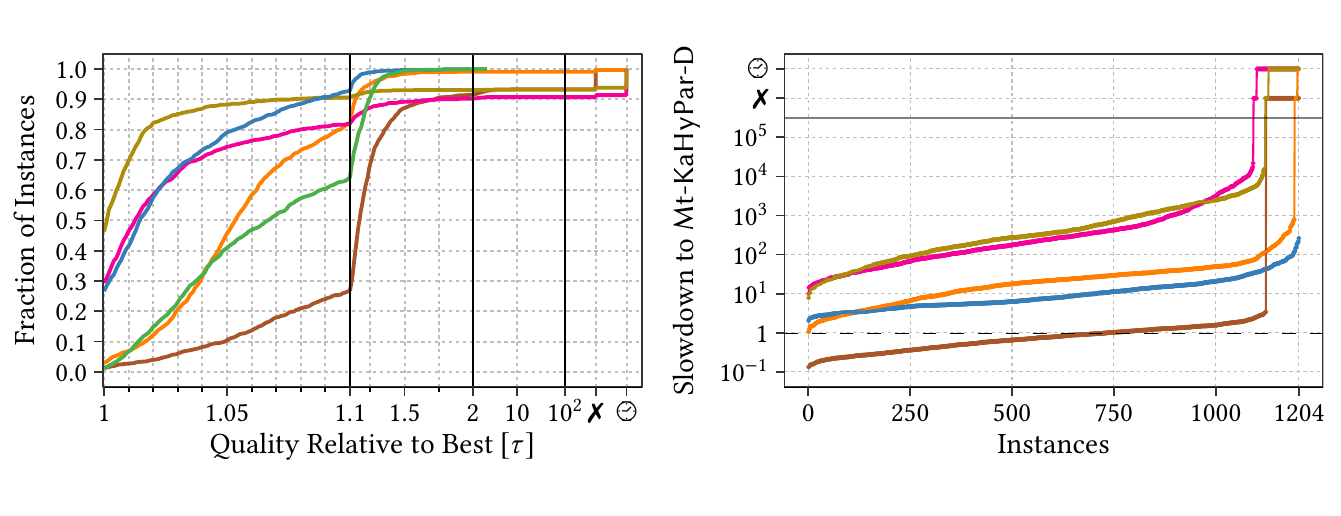}
  \fi
  \end{minipage} %
  \begin{minipage}{\textwidth}
    \vspace{-0.25cm}
    \centering
  \ifpdfplots
    \includegraphics{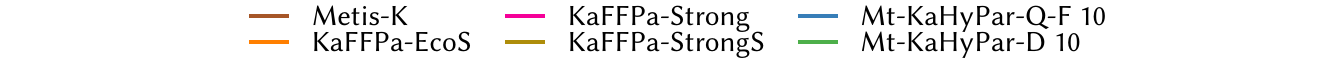}
  \fi
  \end{minipage} %
  \vspace{-0.5cm}
  \caption{Performance profiles and running times comparing \Partitioner{Mt-KaHyPar} to
           \Partitioner{Metis} and \Partitioner{KaFFPa} on \mediumgr.}
	\label{fig:sequential_gp}
\end{figure}

\begin{table}[t]
  \centering
  \small
  \caption{Geometric mean running times of different sequential (hyper)graph partitioning algorithm and \Partitioner{Mt-KaHyPar-D/-Q-F}
           with an increasing number of threads on \mediumgr~(left) and \shortmediumhg~(right).}
  \label{tbl:gmean_times}
  \begin{tabular}{lr|rrr|lr|rrr}
              &        &         & \multicolumn{2}{c|}{\Partitioner{Mt-KaHyPar}} &             &        &         & \multicolumn{2}{c}{\Partitioner{Mt-KaHyPar}} \\
  Seq. Algo.  & $t[s]$ & Threads & \Partitioner{-D} & \Partitioner{-Q-F}  &
  Seq. Algo.  & $t[s]$ & Threads & \Partitioner{-D} & \Partitioner{-Q-F}  \\
  \midrule
  \Partitioner{Metis-K}        & $0.39$    & 16      & $0.45$ & $4.23$  & \Partitioner{PaToH-D}      & $1.17$        & 16      & $0.74$    & $3.98$ \\
  \Partitioner{Metis-R}        & $0.55$    & 10      & $0.55$ & $5.22$  & \Partitioner{PaToH-Q}      & $5.86$        & 10      & $0.88$    & $5.08$  \\
  \Partitioner{KaFFPa-Fast}    & $1.69$    & 8       & $0.61$ & $5.74$  & \Partitioner{Mondriaan}    & $6.62$        & 8       & $1.08$    & $5.58$  \\
  \Partitioner{Scotch}         & $1.84$    & 4       & $0.98$ & $8.98$  & \Partitioner{KaHyPar-CA}   & $28.14$       & 4       & $1.83$    & $9.07$  \\
  \Partitioner{KaFFPa-FastS}   & $1.88$    & 2       & $1.69$ & $15.64$ & \Partitioner{$r$KaHyPar}   & $46.10$       & 2       & $3.33$    & $16.04$  \\
  \Partitioner{KaFFPa-EcoS}    & $10.51$   & 1       & $3.00$ & $28.56$ & \Partitioner{$k$KaHyPar}   & $48.98$       & 1       & $6.24$    & $29.52$  \\
  \Partitioner{KaFFPa-Eco}     & $10.94$   &         &        &         & \Partitioner{hMetis-K}     & $73.75$       &         &           &   \\
  \Partitioner{KaFFPa-Strong}  & $162.83$  &         &        &         & \Partitioner{hMetis-R}     & $93.21$       &         &           &   \\
  \Partitioner{KaFFPa-StrongS} & $201.99$  &         &        &         &             &              &         &           &   \\
  \end{tabular}
\end{table}

As we have seen, \Partitioner{Mt-KaHyPar-D} is faster than most of the sequential algorithms using ten threads.
This raises the question whether or not the result still holds when we use less threads. We therefore
compare the running times of \Partitioner{Mt-KaHyPar-D/-Q-F} with an increasing number of threads to the different sequential
algorithms on \mediumgr~and \shortmediumhg~in Table~\ref{tbl:gmean_times}\footnote{Note that increasing the number of threads
does not affect the solution quality of \Partitioner{Mt-KaHyPar-D/-Q-F}, as shown in Figure~\ref{fig:scalability_quality}.}.
On \mediumgr, \Partitioner{Mt-KaHyPar-D} is faster than most of the sequential algorithms using two threads.
\Partitioner{Metis-K} is still faster than \Partitioner{Mt-KaHyPar-D}, but their running times become comparable when we use $16$ threads.
The sequential time of \Partitioner{Mt-KaHyPar-Q-F} is almost an order of magnitude faster than the running time of the best sequetial
partitioner \Partitioner{KaFFPa-StrongS}, and it becomes faster than \Partitioner{KaFFPa-EcoS} when we use four threads.
On \mediumhg, we have to run \Partitioner{Mt-KaHyPar-D} with eight threads to achieve comparable speed to \Partitioner{PaToH-D}.
However, this number decreases to two threads when we compare their running times on the larger instances of \largehg~\cite[see Fig.~4.17]{GOTT-DIS}.
The sequential time of \Partitioner{Mt-KaHyPar-D} is comparable to \Partitioner{PaToH-Q}, and \Partitioner{Mt-KaHyPar-Q-F} is
significantly faster than its sequential counterpart \Partitioner{$k$KaHyPar} when we use only one thread.

\paragraph{Comparison to Parallel Systems}

Figure~\ref{fig:parallel_hgp} compares \Partitioner{Mt-KaHyPar} to the hypergraph partitioners
\Partitioner{Zoltan} (distributed-memory), \Partitioner{BiPart} (deterministic shared-memory), and
\Partitioner{PaToH} (sequential) on \largehg. Note that \Partitioner{PaToH-D} is fast enough to conduct the experiments on \largehg~in a reasonable time frame,
while this is not the case for any of the other sequential partitioners.
Despite the fact that \Partitioner{Zoltan} has been shown to outperform \Partitioner{BiPart} (see Table~\ref{tbl:outperformed}),
we have included it for a direct comparison to our deterministic configuration \Partitioner{Mt-KaHyPar-SDet}.

The median improvement of \Partitioner{Mt-KaHyPar-SDet} (\gmeantime~$3.14$s) over
\Partitioner{BiPart} ($29.19$s) -- the only existing competitor for deterministic partitioning --
is $200\%$, while it is almost an order of magnitude faster. Our deterministic algorithm
also outperforms \Partitioner{Zoltan} ($12.63$s, median improvement is $12\%$) and
the Wilcoxon signed-ranked test reveals that there is no statistically significant difference between the
solutions produced by \Partitioner{Mt-KaHyPar-SDet} and \Partitioner{PaToH-D} ($51.2$s,
$Z = 1.7314$ and $p = 0.08337$).

\Partitioner{Mt-KaHyPar-D} ($4.64$s) is slightly slower than \Partitioner{Mt-KaHyPar-SDet}, but it computes solutions
that are $23\%$ resp.~$6.6\%$ better than those of \Partitioner{Zoltan} resp.~\Partitioner{PaToH-D} in the median and
is still significantly faster than both algorithms. When flow-based refinement is used (\Partitioner{Mt-KaHyPar-D-F}, not shown in the plots),
we achieve a median improvement over \Partitioner{Zoltan} of $34\%$.
This shows that \Partitioner{Mt-KaHyPar} can partition extremely large hypergraph with high solution quality,
which was previously only possible with sequential codes on medium-sized instances.

\begin{figure}
  \centering
  \begin{minipage}{\textwidth}
    \centering
  \ifpdfplots
    \includegraphics{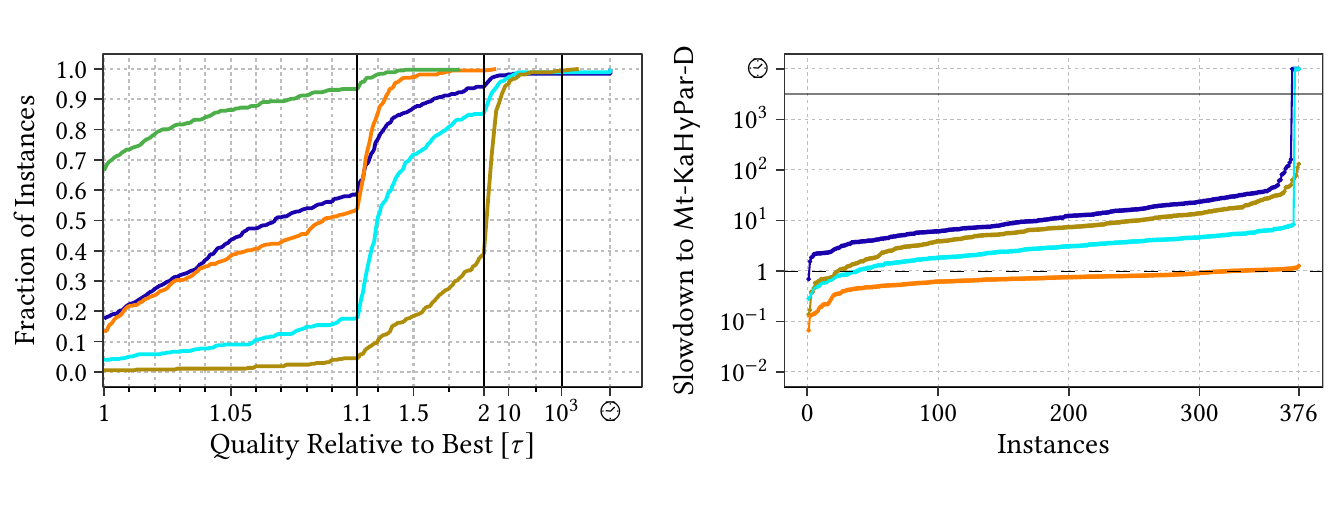}
  \fi
  \end{minipage} %
  \begin{minipage}{\textwidth}
    \vspace{-0.25cm}
    \centering
  \ifpdfplots
    \includegraphics{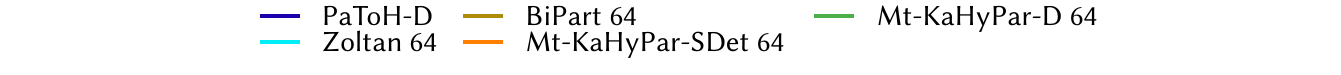}
  \fi
  \end{minipage} %
  \vspace{-0.5cm}
  \caption{Performance profiles and running times comparing \Partitioner{Mt-KaHyPar} to
           \Partitioner{PaToH} and \Partitioner{Zoltan} on \largehg.}
	\label{fig:parallel_hgp}
\end{figure}

Figure~\ref{fig:parallel_gp} compares \Partitioner{Mt-KaHyPar} to the parallel graph partitioner
\Partitioner{KaMinPar} (shared-memory) and \Partitioner{Mt-KaHIP} (shared-memory, also implements
a parallel version of the FM algorithm) on \largegr. We can see that \Partitioner{Mt-KaHyPar-D} ($10.8$s)
computes on most of the instances the best solutions. The median improvement of \Partitioner{Mt-KaHyPar-D}
over \Partitioner{Mt-KaHIP} ($13.69$s) is $2.1\%$, while it is also slightly faster. Out of
all tested algorithms, \Partitioner{KaMinPar} ($2.69$s) is the only algorithm that is faster than \Partitioner{Mt-KaHyPar-D},
but the edge cuts produced by \Partitioner{KaMinPar} are worse than those of \Partitioner{Mt-KaHyPar-D} by
$9.9\%$ in the median. On larger graph instances, \Partitioner{KaMinPar} is the method of choice when speed is more important
than quality, and \Partitioner{Mt-KaHyPar} should be used if one aims for high solution quality.

\begin{figure}
  \centering
  \begin{minipage}{\textwidth}
    \centering
  \ifpdfplots
    \includegraphics{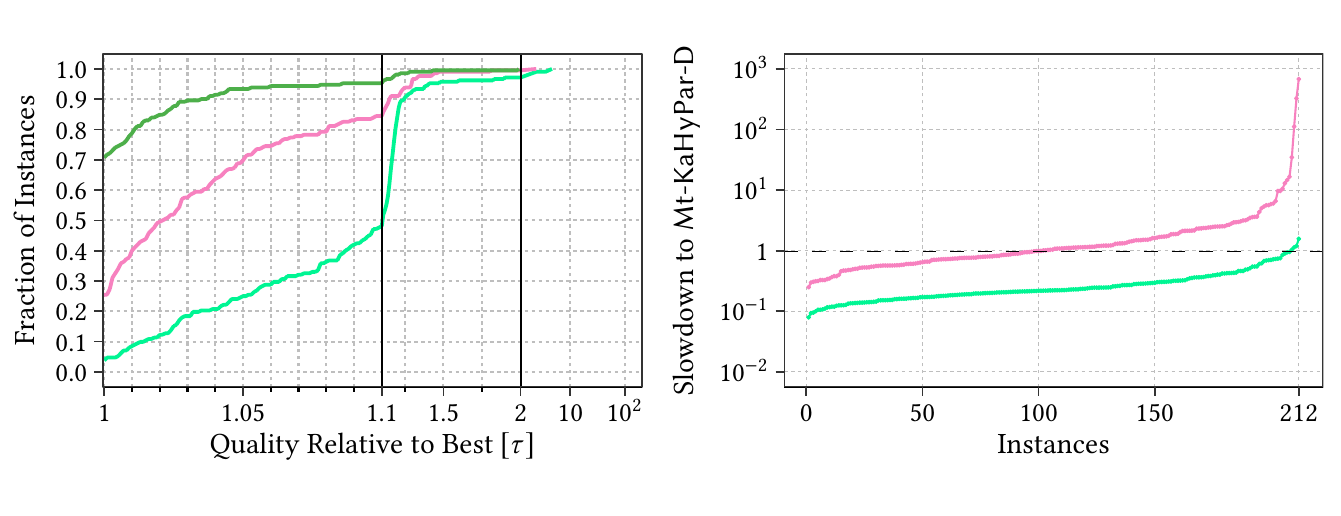}
  \fi
  \end{minipage} %
  \begin{minipage}{\textwidth}
    \vspace{-0.25cm}
    \centering
  \ifpdfplots
    \includegraphics{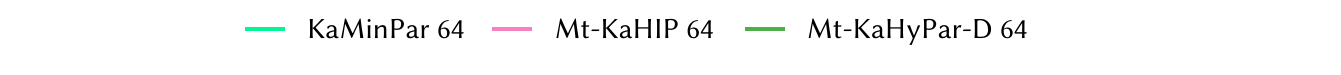}
  \fi
  \end{minipage} %
  \vspace{-0.5cm}
  \caption{Performance profiles and running times comparing \Partitioner{Mt-KaHyPar} to
           \Partitioner{KaMinPar} and \Partitioner{Mt-KaHIP} on \largegr.}
	\label{fig:parallel_gp}
\end{figure}

\paragraph{Limitations}
In this study, we partitioned (hyper)graphs in up to $128$ blocks with an allowed imbalance of $\varepsilon = 3\%$.
We want to point out that there are still settings where the results of this evaluation do not apply.
For example, \Partitioner{KaMinPar} is specifically designed for partitioning graphs into a
large number of blocks (e.g., $k \in \Oh{\sqrt{n}}$). In this setting, existing algorithms struggle to find
balanced solutions or do not complete in a reasonable time frame~\cite{KAMINPAR}.
We are integrating KaMinPar's deep multilevel partitioning scheme in Mt-KaHyPar and hope to offer support for very large $k$ in the near future.
Another limitation is the restriction of our algorithms to running in-memory on a single machine, and thus instances are restricted to the size of currently available RAM.
Finally, partitioning (hyper)graphs with a tight balance constraint (e.g., $\varepsilon \approx 0$) poses additional
challenges for traditional refinement algorithms as this drastically reduces the set of possible moves.

\section{Conclusion}\label{sec:conclusion}

We have presented the \emph{first} set of shared-memory algorithms for partitioning hypergraphs. Our solver \Partitioner{Mt-KaHyPar}
produces solutions on par with the best sequential codes, while it is faster than most of the existing parallel algorithms.
We demonstrated this achievement in our extensive experimental evaluation with 25 sequential and parallel
graph and hypergraph partitioners tested on over $800$ (hyper)graphs.
We contributed parallel formulations for all phases of the multilevel scheme: a parallel clustering-based coarsening
algorithm guided by the community structure of the input hypergraph obtained via a parallel community detection algorithm,
initial partitioning via parallel recursive bipartitioning using work-stealing, the first fully-parallel FM implementation,
and a parallelization of flow-based refinement.
Perhaps the most suprising result is the efficient parallelization of the $n$-level partitioning scheme, even though
we showed that traditional multilevel algorithms can compute comparable solutions when flow-based refinement is used.
Furthermore, we presented multiple techniques to accurately (re)compute gain values
for concurrent node moves, which had not been addressed in
parallel partitioning algorithms before. We also proposed data structure optimizations
for plain graphs, making \Partitioner{Mt-KaHyPar} the state-of-the-art solver for graph partitioning.
Additionally, we devised a deterministic version of our multilevel algorithm based on the synchronous local moving scheme.

Given that quality improvements often come at the cost of significantly longer running times, it may be interesting to
evaluate the quality-time trade-off of existing tools for applications before advancing the
field of high-quality partitioning. For instances that do not fit into the main memory of a single machine,
translating the techniques presented in this work into the distributed-memory setting is also
an important area for future research. We see further algorithmic improvements in a localized version
of flow-based refinement that runs after each batch uncontraction in the $n$-level scheme as well as
improving clustering decisions in the coarsening phase.

\begin{acks}
The authors thank Michael Hamann, Daniel Seemaier, Christian Schulz and Dorothea Wagner for helpful discussions over the course of this research.
This work was supported in part by DFG grants WA654/19-2 and SA933/11-1. The
authors acknowledge support by the state of Baden-Württemberg through bwHPC.
\end{acks}


\bibliographystyle{ACM-Reference-Format}
\bibliography{references}

\clearpage

\appendix

\section{Comparison to Other Systems}\label{appendix:comparison}

\begin{figure*}[!htb]
  \centering
  \begin{minipage}{\textwidth}
    \centering
  \ifpdfplots
    \includegraphics{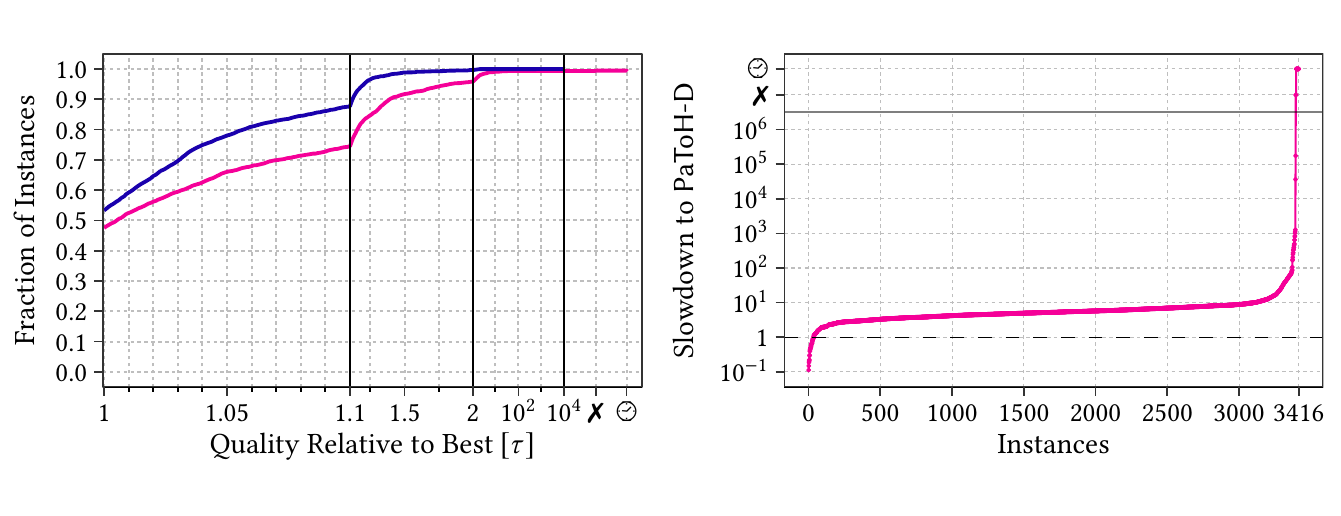}
  \fi
  \end{minipage} %
  \begin{minipage}{\textwidth}
    \vspace{-0.25cm}
    \centering
  \ifpdfplots
    \includegraphics{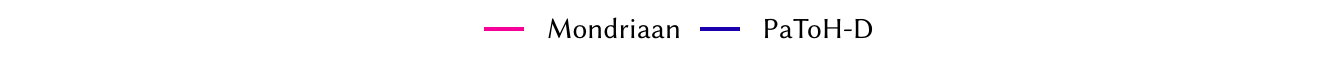}
  \fi
  \end{minipage} %
  \vspace{-0.5cm}
  \caption{Performance profiles and running times comparing \Partitioner{PaToH-D} and
           \Partitioner{Mondriaan} on \mediumhg.}
	\label{fig:patoh_vs_mondriaan}
\end{figure*}

\begin{figure*}[!htb]
  \centering
  \begin{minipage}{\textwidth}
    \centering
  \ifpdfplots
    \includegraphics{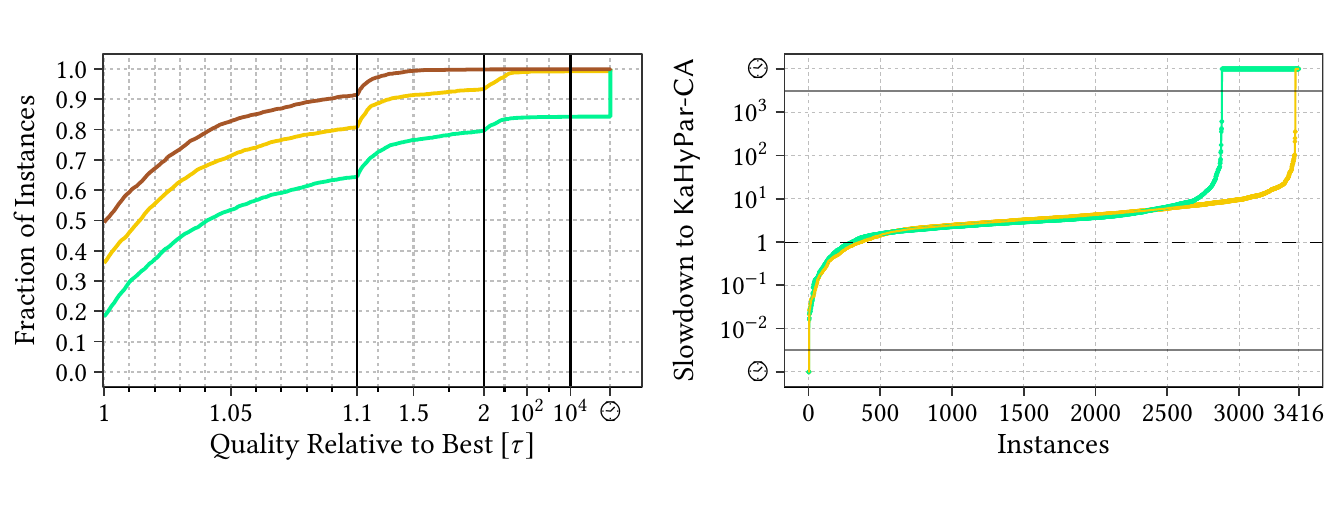}
  \fi
  \end{minipage} %
  \begin{minipage}{\textwidth}
    \vspace{-0.25cm}
    \centering
  \ifpdfplots
    \includegraphics{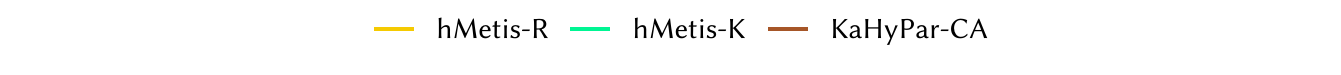}
  \fi
  \end{minipage} %
  \vspace{-0.5cm}
  \caption{Performance profiles and running times comparing \Partitioner{KaHyPar-CA} and
           \Partitioner{hMetis} on \mediumhg.}
	\label{fig:kahypar_vs_hmetis}
\end{figure*}

\begin{figure*}[!htb]
  \centering
  \begin{minipage}{\textwidth}
    \centering
  \ifpdfplots
    \includegraphics{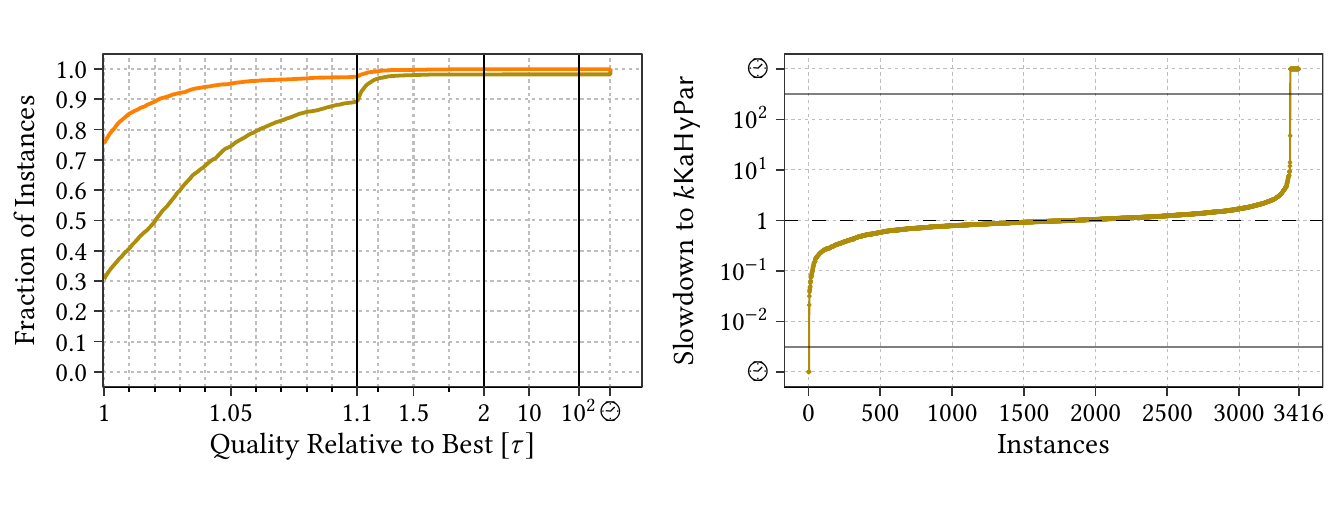}
  \fi
  \end{minipage} %
  \begin{minipage}{\textwidth}
    \vspace{-0.25cm}
    \centering
  \ifpdfplots
    \includegraphics{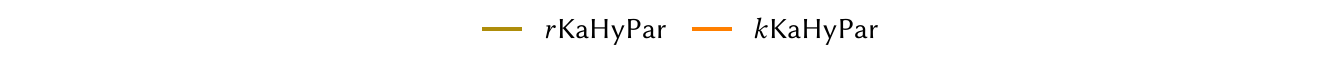}
  \fi
  \end{minipage} %
  \vspace{-0.5cm}
  \caption{Performance profiles and running times comparing \Partitioner{$r$KaHyPar} and
           \Partitioner{$k$KaHyPar} on \mediumhg.}
	\label{fig:kahypar_r_vs_k}
\end{figure*}

\begin{figure*}[!htb]
  \centering
  \begin{minipage}{\textwidth}
    \centering
  \ifpdfplots
    \includegraphics{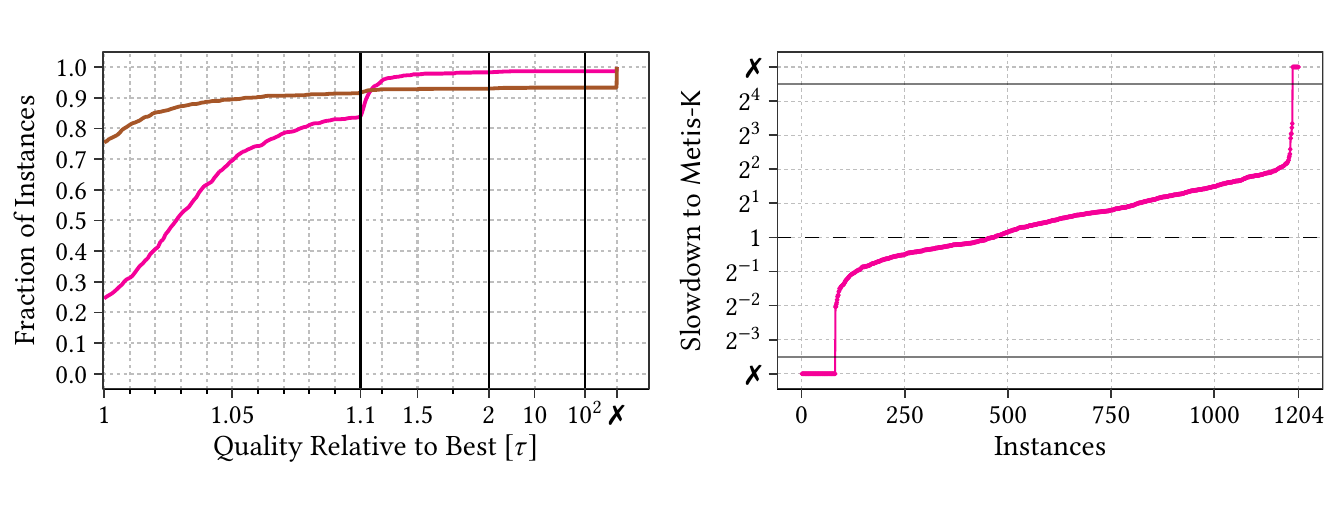}
  \fi
  \end{minipage} %
  \begin{minipage}{\textwidth}
    \vspace{-0.25cm}
    \centering
  \ifpdfplots
    \includegraphics{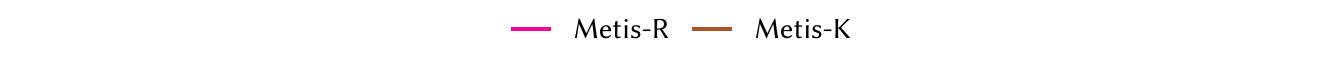}
  \fi
  \end{minipage} %
  \vspace{-0.5cm}
  \caption{Performance profiles and running times comparing \Partitioner{Metis-R} and
           \Partitioner{Metis-K} on \mediumgr.}
	\label{fig:metis_k_vs_r}
\end{figure*}

\begin{figure*}[!htb]
  \centering
  \begin{minipage}{\textwidth}
    \centering
  \ifpdfplots
    \includegraphics{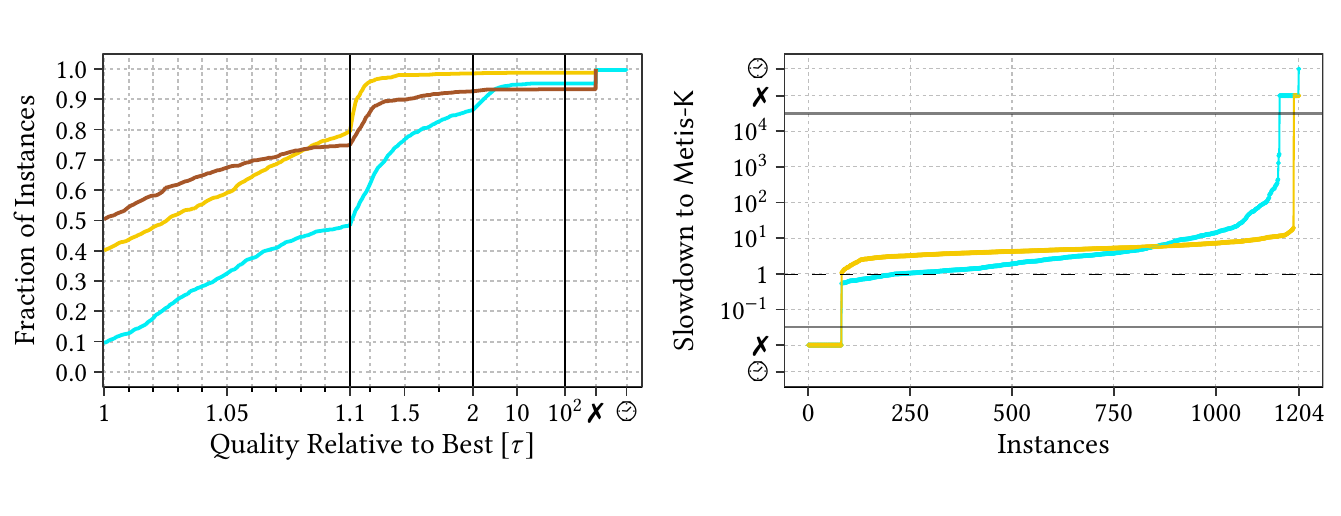}
  \fi
  \end{minipage} %
  \begin{minipage}{\textwidth}
    \vspace{-0.25cm}
    \centering
  \ifpdfplots
    \includegraphics{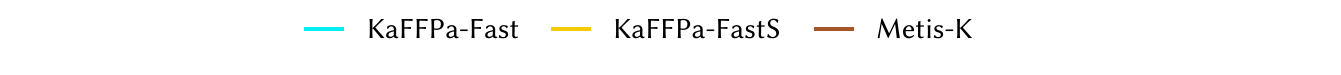}
  \fi
  \end{minipage} %
  \vspace{-0.5cm}
  \caption{Performance profiles and running times comparing \Partitioner{Metis-K} and
           \Partitioner{KaFFPa-Fast}/\Partitioner{-FastS} on \mediumgr.}
	\label{fig:metis_k_vs_kaffpa_fast}
\end{figure*}

\begin{figure*}[!htb]
  \centering
  \begin{minipage}{\textwidth}
    \centering
  \ifpdfplots
    \includegraphics{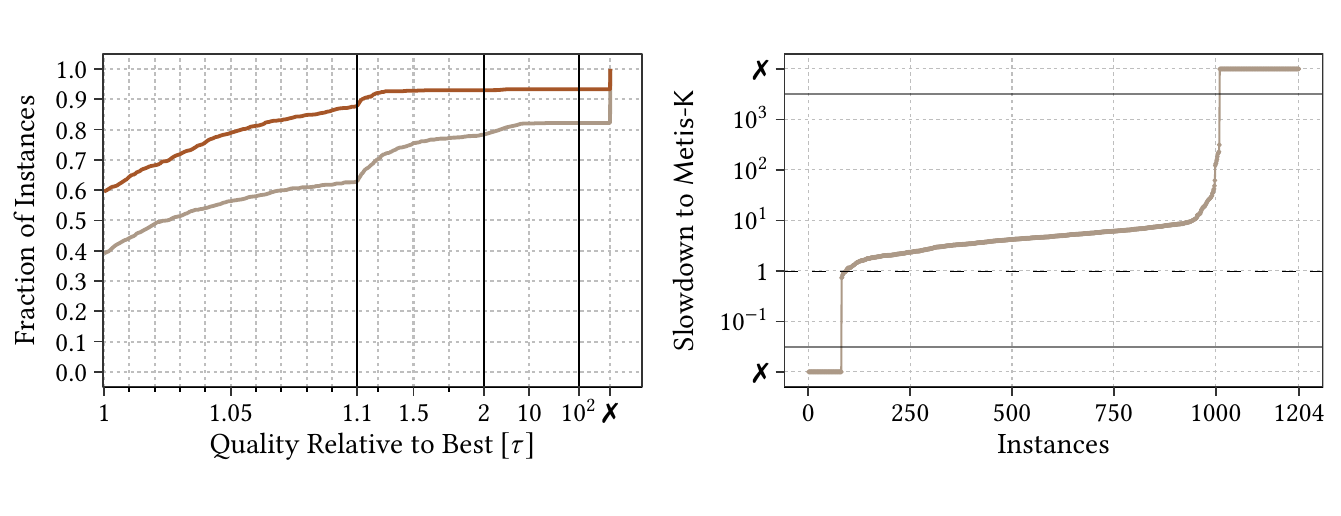}
  \fi
  \end{minipage} %
  \begin{minipage}{\textwidth}
    \vspace{-0.25cm}
    \centering
  \ifpdfplots
    \includegraphics{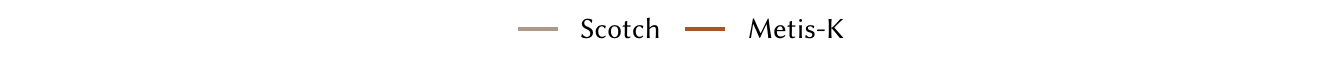}
  \fi
  \end{minipage} %
  \vspace{-0.5cm}
  \caption{Performance profiles and running times comparing \Partitioner{Metis-K} and
           \Partitioner{Scotch} on \mediumgr.}
	\label{fig:metis_k_vs_scotch}
\end{figure*}

\begin{figure*}[!htb]
  \centering
  \begin{minipage}{\textwidth}
    \centering
  \ifpdfplots
    \includegraphics{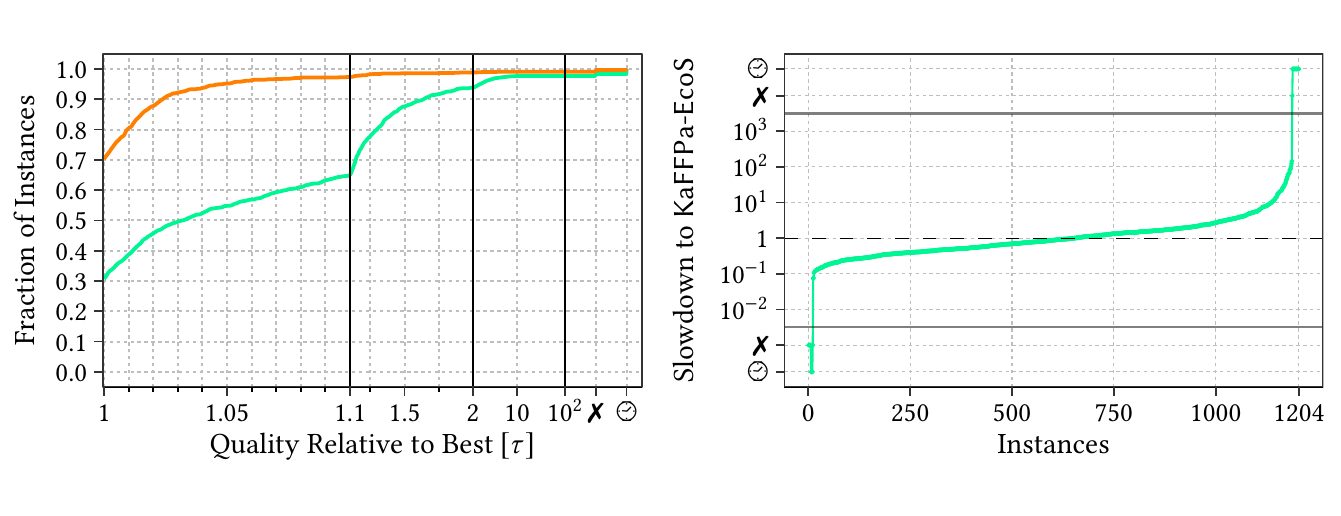}
  \fi
  \end{minipage} %
  \begin{minipage}{\textwidth}
    \vspace{-0.25cm}
    \centering
  \ifpdfplots
    \includegraphics{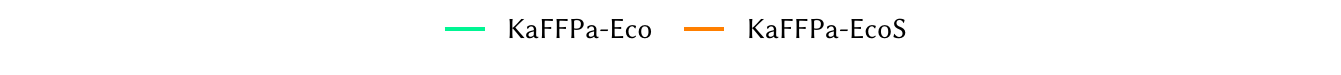}
  \fi
  \end{minipage} %
  \vspace{-0.5cm}
  \caption{Performance profiles and running times comparing \Partitioner{KaFFPa-Eco} and
           \Partitioner{KaFFPa-EcoS} on \mediumgr.}
	\label{fig:kaffpa_eco_vs_ecos}
\end{figure*}

\begin{figure*}[!htb]
  \centering
  \begin{minipage}{\textwidth}
    \centering
  \ifpdfplots
    \includegraphics{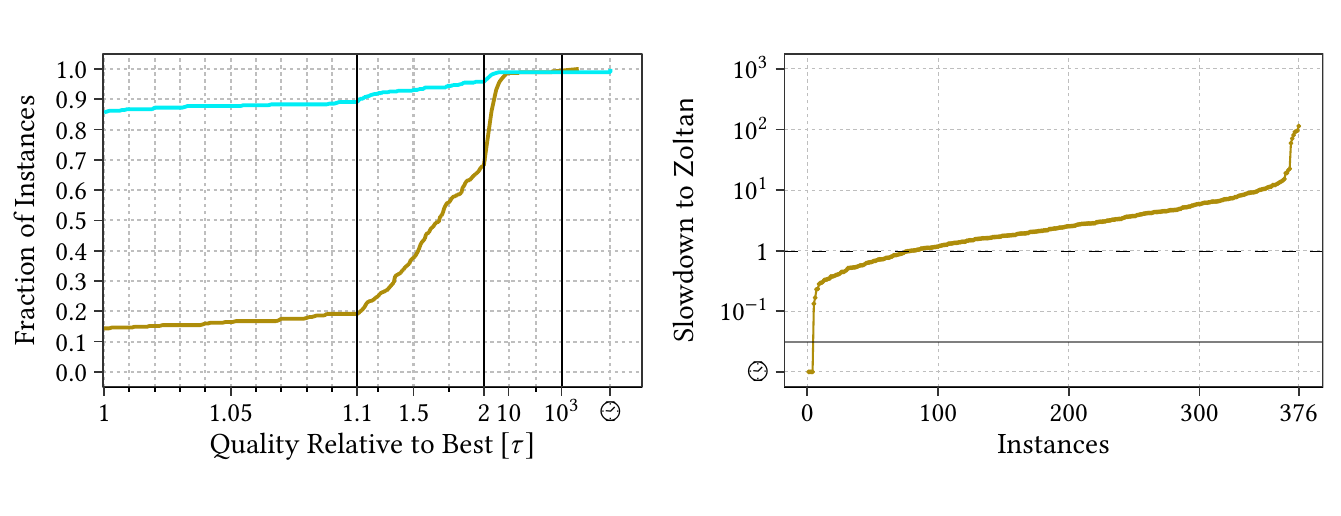}
  \fi
  \end{minipage} %
  \begin{minipage}{\textwidth}
    \vspace{-0.25cm}
    \centering
  \ifpdfplots
    \includegraphics{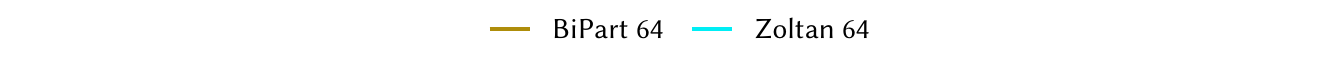}
  \fi
  \end{minipage} %
  \vspace{-0.5cm}
  \caption{Performance profiles and running times comparing \Partitioner{Zoltan} and
           \Partitioner{BiPart} on \largehg.}
	\label{fig:zoltan_vs_bipart}
\end{figure*}

\begin{figure*}[!htb]
  \centering
  \begin{minipage}{\textwidth}
    \centering
  \ifpdfplots
    \includegraphics{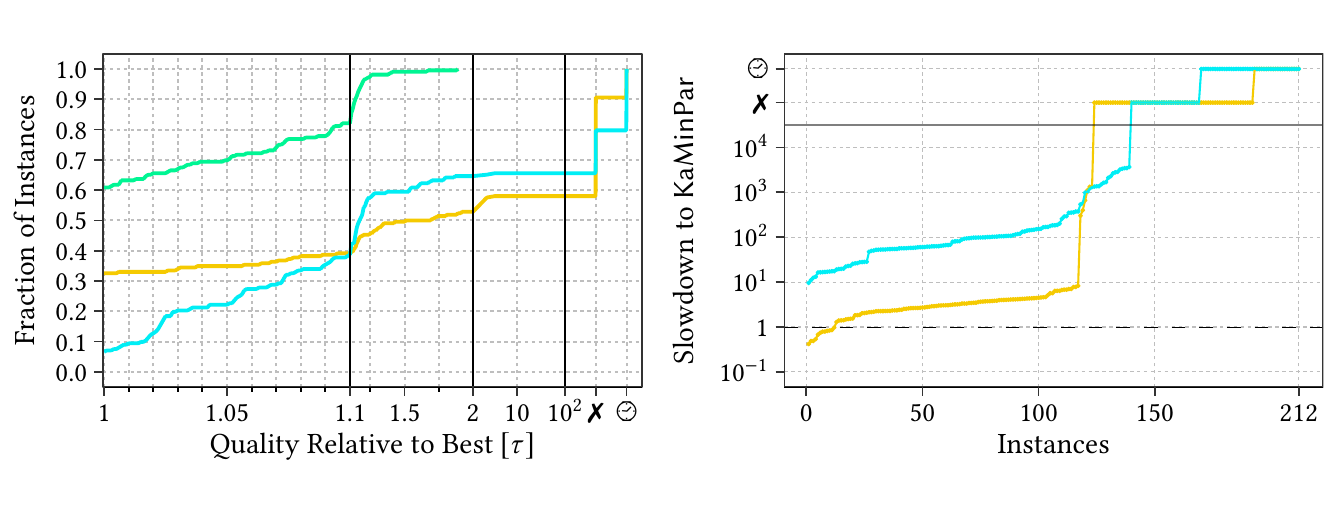}
  \fi
  \end{minipage} %
  \begin{minipage}{\textwidth}
    \vspace{-0.25cm}
    \centering
  \ifpdfplots
    \includegraphics{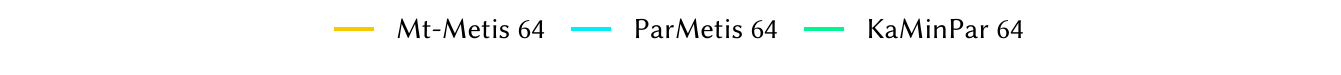}
  \fi
  \end{minipage} %
  \vspace{-0.5cm}
  \caption{Performance profiles and running times comparing \Partitioner{KaMinPar} to
           \Partitioner{Mt-Metis} and \Partitioner{ParMetis} on \largegr.}
	\label{fig:kaminpar_vs_metis}
\end{figure*}

\begin{figure*}[!htb]
  \centering
  \begin{minipage}{\textwidth}
    \centering
  \ifpdfplots
    \includegraphics{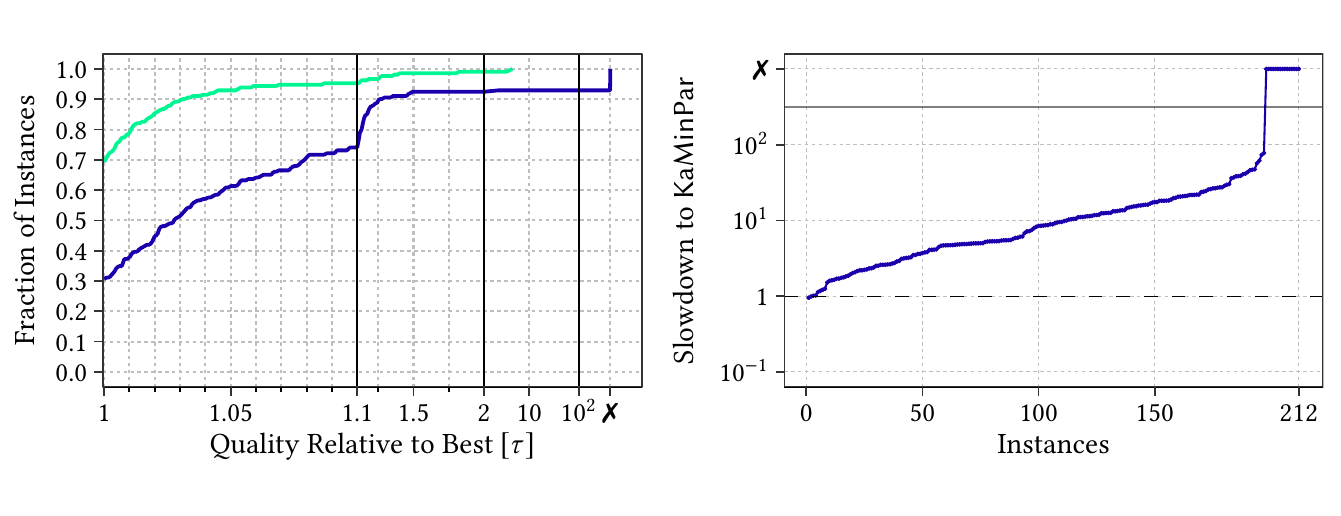}
  \fi
  \end{minipage} %
  \begin{minipage}{\textwidth}
    \vspace{-0.25cm}
    \centering
  \ifpdfplots
    \includegraphics{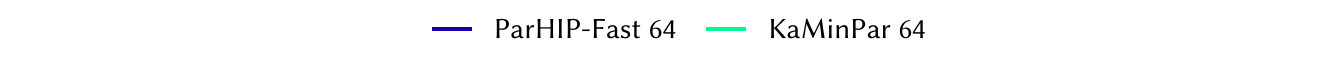}
  \fi
  \end{minipage} %
  \vspace{-0.5cm}
  \caption{Performance profiles and running times comparing \Partitioner{KaMinPar} to
           \Partitioner{ParHIP-Fast} on \largegr.}
	\label{fig:kaminpar_vs_parhip}
\end{figure*}

\begin{figure*}[!htb]
  \centering
  \begin{minipage}{\textwidth}
    \centering
  \ifpdfplots
    \includegraphics{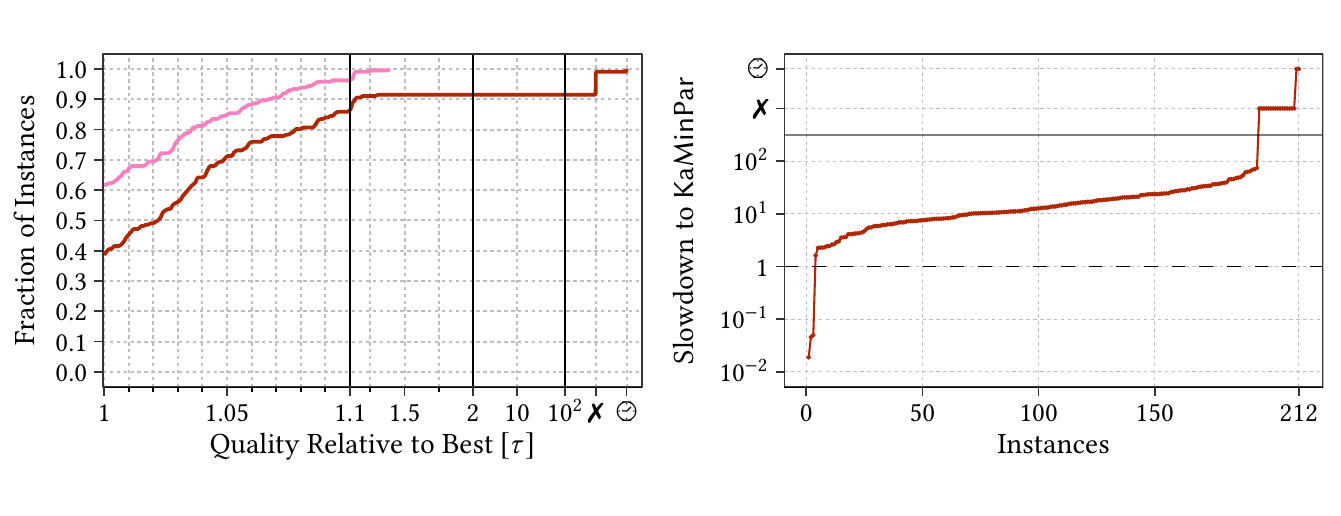}
  \fi
  \end{minipage} %
  \begin{minipage}{\textwidth}
    \vspace{-0.25cm}
    \centering
  \ifpdfplots
    \includegraphics{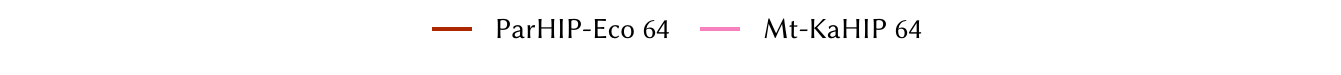}
  \fi
  \end{minipage} %
  \vspace{-0.5cm}
  \caption{Performance profiles and running times comparing \Partitioner{Mt-KaHIP} to
           \Partitioner{ParHIP-Eco} on \largegr.}
	\label{fig:mt_kahip_vs_parhip}
\end{figure*}

\begin{figure*}[!htb]
  \centering
  \begin{minipage}{\textwidth}
    \centering
  \ifpdfplots
    \includegraphics{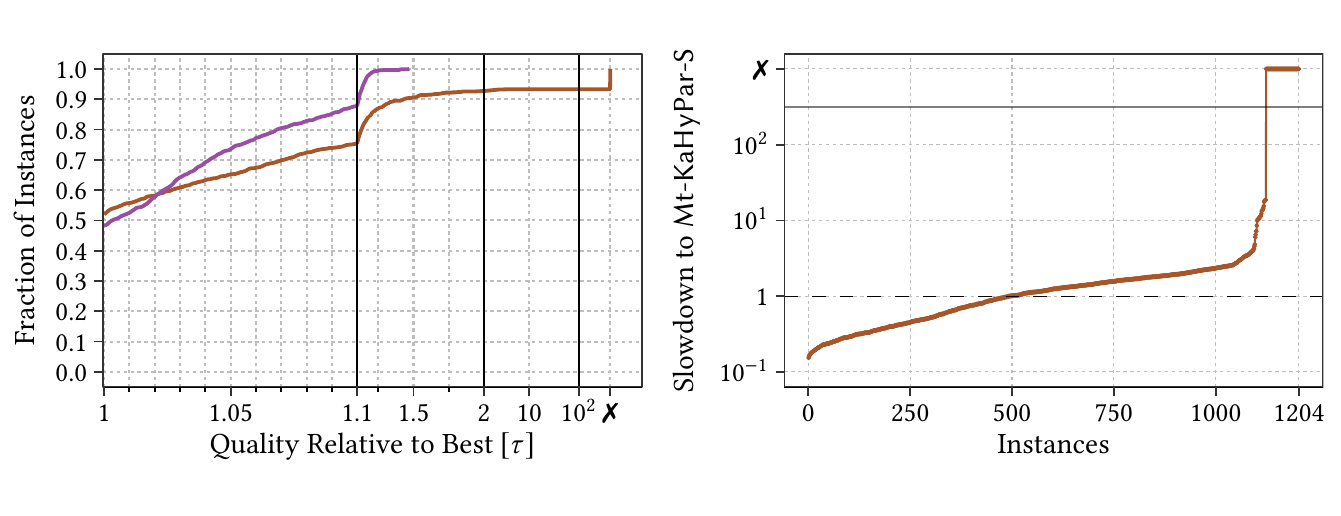}
  \fi
  \end{minipage} %
  \begin{minipage}{\textwidth}
    \vspace{-0.25cm}
    \centering
  \ifpdfplots
    \includegraphics{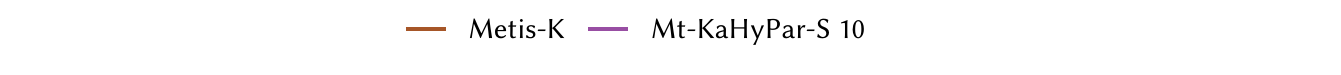}
  \fi
  \end{minipage} %
  \vspace{-0.5cm}
  \caption{Performance profiles and running times comparing \Partitioner{Mt-KaHyPar-S} (\textbf{S}peed, \Partitioner{Mt-KaHyPar-D} without FM refinement) to
           \Partitioner{Metis-K} on \mediumgr.}
	\label{fig:metis_k_vs_mt_kahypar_s}
\end{figure*}

\end{document}